\documentclass[10pt, a4paper]{article}

\usepackage[margin=1in]{geometry}

\usepackage{setspace}
\doublespace   

\usepackage{amsmath,amsfonts}
\usepackage{algorithmic}
\usepackage{array}
\usepackage{caption}
\usepackage[caption=false,font=normalsize,labelfont=sf,textfont=sf]{subfig}
\usepackage{textcomp}
\usepackage{stfloats}
\usepackage{url}
\usepackage{verbatim}
\usepackage{graphicx}
\usepackage{cite}

\usepackage{amsthm}
\newtheorem{remark}{Remark}
\newtheorem{lemma}{Lemma}
\newtheorem{proposition}{Proposition}
\newtheorem{theorem}{Theorem}

\usepackage{hyperref}
\usepackage{cleveref}
\usepackage{multirow}
\usepackage{amssymb}
\usepackage{setspace}

\usepackage[linesnumbered,ruled,vlined]{algorithm2e}
\SetKwInput{KwInput}{Input}
\SetKwInput{KwOutput}{Output}
\SetKwInOut{Parameter}{Parameter}
\SetKwFor{Case}{case}{}{end case}%
\SetKwProg{Function}{Function}{}{}
\SetKwProg{Initialization}{initialization}{}{}
\SetKwProg{return}{return}{}{}
\SetKwProg{Upon}{upon receiving}{ do}{end}

\usepackage[table,xcdraw]{xcolor}
\definecolor{node_1}{RGB}{139,87,42} 
\definecolor{node_2}{RGB}{144, 19, 254}
\definecolor{node_3}{RGB}{74, 144, 226}
\definecolor{node_4}{RGB}{208, 2, 27}
\definecolor{node_5}{RGB}{126, 211, 33}

\providecommand{\keywords}[1]{\textbf{\textit{Index terms --}} #1}

\begin{document}

\title{Secure and Private Vickrey Auction Protocols: \\A Secure Multiparty Computation Approach}
\author{Lucy Klinger, Mengfan Lyu and Lei Zhang
\thanks{
L. Klinger is affiliated with the Beijing International Center for Mathematical Research at Peking University, China. 
Email: lucyk@bicmr.pku.edu.cn.
M. Lyu is affiliated with the School of Mathematics and Statistics at the University of Sydney, Australia. 
Email: M.Lyu@maths.usyd.edu.au.
L. Zhang is affiliated with the Beijing International Center for Mathematical Research, the Center for Quantitative Biology, and the Center for Machine Learning Research 
at Peking University, China. 
Email: zhangl@math.pku.edu.cn.
}
}
\date{April 26, 2023}
\maketitle

\begin{abstract}
Recent attention on secure multiparty computation and blockchain technology has garnered new interest in developing auction protocols in a decentralized setting. 
In this paper, we propose a secure and private Vickrey auction protocol that addresses the challenge of maintaining auction correctness while preserving the confidentiality of bidders' information. 
Our protocol leverages secure multiparty computation (SMPC) techniques that enable decentralized computation of the final payment price without disclosing additional information.
The total computational complexity of our proposed protocol, including preparation, is $O(n^2 k)$, and the complexity for the main auction is $O(n k)$.
Furthermore, we show that our protocol is resistant to collusion and the auction outcome is publicly verifiable, which prevents dishonest bidders from altering this outcome.
\end{abstract}

\bigskip


\section{Introduction}\label{section_1}

Auction is a mechanism for allocating resources through competition among self-interested bidders. 
Second-price sealed-bid auctions, known as Vickrey auctions, have drawn significant interest from researchers, especially in the fields of economics and game theory. 
In such auctions, the highest bidder wins the auction but only pays the second-highest bid. 
This auction scheme, named after William Vickrey, who received a Nobel Prize in 1996, has been shown to be incentive-compatible, meaning that the best response, regardless of the behavior of others, is to bid the amount that the good is actually worth to the bidder. 
Consequently, the object being auctioned will go to the bidder who values it most, regardless of the specific auction format used. 
For example, the expected payment made by the winning bidder will be the same whether it is an English auction, Dutch auction, or first-price auction \cite{vickrey1961counterspeculation}. 
In contrast, first-price sealed auctions may result in inefficient allocations when bidders try to outbid each other slightly to pay as little as possible. 
This can lead to the bidder with the highest valuation for the item underestimating other bids,
potentially causing another bidder to win the object ---in an inefficient allocation.

Despite the efficiency of Vickrey auctions in determining the market value of an item, 
they are not commonly utilized in practice due to security and privacy concerns, 
particularly in the context of digital transactions and the widespread adoption of online environments \cite{rothkopf1995two, rothkopf1990vickrey, sandholm1996limitations, brandt2001cryptographic}.  
The main challenges for existing Vickrey auction protocols include ensuring the correctness of auction outcomes, 
preserving the confidentiality of bidders' information, resisting collusion, and being robust to malicious behavior. 
These challenges are important because they impact auction correctness and privacy. 
Violations of privacy may affect the role of the Nash equilibrium in the Vickrey auction 
so that the auction is no longer incentive-compatible and its mechanism design becomes inefficient.

Auction correctness is a fundamental aspect of the Vickrey auction process, 
ensuring an accurate execution of the auction that identifies the highest bidder and determines the selling price.
Collusion presents a significant challenge to this, 
particularly in traditional Vickrey auction protocols that rely on trusted third parties to handle bids. 
This reliance can introduce a single point of failure and raise concerns about the potential for collusion or corruption. 
For example, the presence of an auctioneer poses considerable risks, 
such as when the auctioneer overstates the second-highest bid to increase the seller's revenue, declares a non-existent winning bidder due to low bids, or colludes with the winning bidder to understate the selling price \cite{10.1007/s10207-006-0001-y}.  
Malicious behavior can also affect correctness: attackers can compromise underlying cryptographic mechanisms by submitting fake bids or tampering with the bid evaluation process.

Privacy violations in Vickrey auctions may not only impact auction correctness but also affect the Nash equilibrium of the auction. 
For instance, collusion and malicious behavior may enable bidders to access information about other bids. 
As a result, bidders might be hesitant to bid truthfully, fearing that the auctioneer or other bidders could expose their bids and reveal their private valuations. 
Moreover, even the disclosure of partial information about bids can influence participants' strategic behavior. 
Fear of information leaks can alter an assumption of the Vickrey model, namely, that bidders have independent private valuations \cite{vickrey1961counterspeculation}.  
Consequently, this private valuation model turns into a common value auction model, where a bidder's valuation depends on the information held by other bidders. 
This model is generally not incentive-compatible: bidders may misrepresent their valuations to gain an advantage, leading to inefficiencies in the mechanism design \cite{krishna2009auction}. 
Therefore, privacy concerns are an essential factor in designing and implementing efficient Vickrey auctions.

In order to address these security and privacy concerns, 
our aim is to develop a new fully private Vickrey auction protocol that employs secure multiparty computation (SMPC).

\section{Literature Review}\label{section_2}

Secure and private auctions have attracted considerable attention in cryptography and auction theory. 
A subfield of cryptography, secure multiparty computation (SMPC), focuses on protocols that allow multiple parties to perform computations on private inputs without revealing the inputs themselves \cite{10.1145/3387108}. 
Many existing auction protocols involve multiple auctioneers  \cite{franklin1996design, sako2000auction, kurosawa2002bit, bogetoft2008multiparty, galbraith2019proceedings} or trusted third parties \cite{naor1999privacy, juels2003two, abe2002m+}, relying on non-collusion assumptions to ensure security and privacy, posing challenges in inherently distrustful settings. 

Alternative auction protocols without trusted third parties employ techniques such as homomorphic encryption and garbled circuits to ensure secure computation, as demonstrated in protocols presented in \cite{bogetoft2006practical, elkind2004interleaving, decker2003secure, bradford2008protocol}. 
However, these methods can be computationally expensive and may face challenges, such as the potential requirement for pairwise secret channels \cite{brandt2005efficient}, 
which may not always be practical in decentralized settings.

Bidder-resolved auction protocols, auctioneer-free solutions that use secret sharing techniques, are among those presented in  \cite{kikuchi1999multi, brandt2001cryptographic, brandt2002secure, bag2019seal}. 
In particular, 
Brandt introduced a Vickrey auction protocol based on cryptographic primitives, 
one which distributes computation among bidders and ensures full privacy in a decentralized setting \cite{brandt2002secure}. 
However, the protocol's total computational complexity is 
$O(n^2 2^k)$, with $O(n 2^k)$ complexity for the main auction protocol,
where $n$ denotes the number of bidders and $k$ represents the number of bits in the binary format of the maximum bid.
This exponential computational complexity may affect its suitability for auctions involving large value bids and limit practical application.

To address this issue, recent auction protocols, 
such as the one presented in \cite{bag2019seal}, 
use binary radix representations of bids to reduce the complexity to $O(n k)$. 
Originally designed for a first-price sealed-bid auction, these kinds of protocols can be applied to Vickrey auctions: 
assume the highest bid is
$[~p_{(1),1} ~p_{(1),2}\dots p_{(1),k}~]_{\text{base 2}}$
and second highest bid is
$[~p_{(2),1} ~p_{(2),2}\dots p_{(2),k}~]_{\text{base 2}}$,
the winning bidder is excluded from the winning iteration $j_w = \min \{~j\mid p_{(1),j} \neq p_{(2) ,j} ~\}$ and exits the protocol while losing bidders continue with the remaining iterations. 
However, this step reveals the lower boundary $2^{j_w}$ of the highest bid
and has limited implementation when full privacy is required. 
A dishonest participant could collude with the seller to exploit this partial information in a repeated auction, thereby increasing the seller's revenue.

In this context, 
we propose a novel protocol that addresses some limitations of previous auction models and offers enhanced efficiency and privacy. 
Our goal is to develop a scalable Vickrey auction protocol that operates in a decentralized environment, ensuring correctness, full privacy, and security against collusion or malicious behavior. 
Our protocol allows bidders to calculate the final payment price without a central authority and allow all operations to be publicly verifiable, enabling third-party observers to confirm the auction's integrity. 
Additionally, the protocol is designed to prevent dishonest bidders from manipulating auction outcomes and to ensure full privacy. 
Even in the event of collusion, no additional information can be extracted. 

Our protocol has several key advantages. 
First, 
it achieves improved computational efficiency, 
with a total complexity of $O(n^2 k)$, and $O(n k)$ complexity for the main auction. 
This is notably faster than the YMB-SHARE protocol proposed in \cite{brandt2002secure}, which has a total complexity of $O(n^2 2^k)$.
Second, our protocol ensures full privacy for all bidders, distinguishing it from other models such as \cite{bag2019seal}. 
Although the latter model suggests a scheme that could be extended with comparable complexity to Vickrey auctions, it reveals partial information. In contrast, our protocol only reveals the second-highest bid and guarantees full privacy for all bidders.

This paper is organized as follows. 
 In \Cref{section_1}, we provide an overview of security and privacy issues concerning traditional Vickrey auctions.
 In \Cref{section_2}, we summarize existing research in the field of SMPC for Vickrey auction protocols. 
 In \Cref{section_3}, we present a step-by-step protocol for a secure and private Vickrey auction with no auctioneer. 
 In \Cref{section_4}, we prove that the proposed protocol maintains correctness and full privacy throughout the auction process. \Cref{section_5} presents numerical results to support our findings, while \Cref{section_6} concludes the paper with a summary of our findings and suggests future research directions.

\section{Protocol Design}\label{section_3}

In this section, we will provide a step-by-step explanation of conducting a Vickrey auction protocol without an auctioneer, one that ensures correctness, privacy, and security against collusion or malicious behavior in a decentralized setting.

Consider a set of bidders denoted by $I = \{\alpha_1, \alpha_2, \dots, \alpha_n\}$, 
with $l \in \{1, 2, \dots, n\}$ as the index for bidders. 
Each bidder has a private valuation of ${p_l}$, 
and the best response is to bid this amount,
regardless of the behavior of other bidders. 
To ensure the scalability of the auction, 
we take inspiration from \cite{bag2019seal} and examine the binary representation of each bid,
beginning with the first digit $j=1$ and progressing towards the right until $j=k$. 
In this context, 
the bid is denoted as  $p_l=\left[~p_{l,1}~p_{l,2}\dots p_{l,k}~\right]_{\text{base 2}}$, 
where $p_{l,j}\in\{0,1\}$ represents the $j^\text{th}$ digit of the binary representation with a total of $k$ bits.

The primary objectives of our protocol are
\begin{enumerate}
\item
correctness: 
accurately determine the winner and the selling price (i.e., the second-highest bid) 
\item
full privacy: 
reveal no additional information other than the selling price during and after the auction. 
\item
security against collusion or malicious behavior: ensure that no coalition of up to $n-1$ participants can manipulate the equilibrium design of the auction mechanism or obtain any additional information about others
\item
decentralization: eliminate the need for a central authority, enabling a distributed and trustless environment
\item
public verifiability: ensure that all operations are publicly verifiable by allowing third-party observers to confirm the integrity of the auction outcome.
\end{enumerate}

\subsection{Intuition}

In order to achieve our primary objectives, 
we present a high-level intuition consisting of the following phases:
\begin{enumerate}
\item \label{phase_1} Bid commitment
\item \label{phase_2} Determination of output price
\item \label{phase_3} Second-highest bid verification
\item \label{phase_4} Winner determination
\end{enumerate}

In Phase \ref{phase_1}, 
each bidder commits their bids on a public channel, 
such as a bulletin board or a blockchain.

In Phase \ref{phase_2}, 
we initially formulated an equation using two commonly known variables, 
allowing bidders to jointly detect bid placements for each digit $j$ regarding iterations where $j=1,2,\dots,k$. 
This equation holds when no one bids for the current digit (returning an output of `0') but does not hold when someone bids for the current digit (returning an output of `1'). 
To determine the second-highest bid for each digit, it is necessary to conceal the winning bid and bids lower than the second-highest bid. 
To conceal information about the highest bidder, each bidder must generate a privately known variable that is used in an equation only recognized by the highest bidder.
With this equation, the highest bidder can then secretly manipulate the value of one of the original common variables using a cryptographic tool so it seems like they did not bid for the current digit, while inputting `1' for all subsequent digits to maintain their lead.
On the other hand, to conceal information about losing bidders, they must input `0' as their bids for all subsequent digits. 
Doing so allows them to maintain the appearance of still being in the race without revealing their loss. 
By hiding, for each digit, bids other than the second-highest bid, the output price can be determined without disclosing additional information.

In Phase \ref{phase_3},
we verify that the output price is the second-highest bid by comparing that bid to all the bids committed to in Phase \ref{phase_1}.

In Phase \ref{phase_4}, 
which occurs after the payment price is determined, winners can prove their success in the auction by demonstrating use of a privately known variable. This indicates a willingness to pay a price that is at least one bid higher than the second-highest bid while maintaining their privacy.

An example of this protocol's intuition can be found in Table \ref{table_example}.


\begin{table*}[ht] 
\center\resizebox* {!} {6cm}{ 
\begin{tabular}{  c | c c | c c | c c | c c | c c | c c | c c | c  c } 
& \multicolumn{2}{c|}{$j=1$} 
& \multicolumn{2}{c|}{$j=2$} 
& \multicolumn{2}{c|}{$j=3$} 
& \multicolumn{2}{c|}{$j=4$} 
& \multicolumn{2}{c|}{$j=5$} 
& \multicolumn{2}{c|}{$j=6$} 
& \multicolumn{2}{c|}{$j=7$} 
& \multicolumn{2}{c|}{$j=8$} 
\\ \cline{1-17} 
\multirow{2}{*}{{\color{node_1}  $p_1=143$}}
& {\color{node_1}  1} &  
& {\color{node_1}  0} &  
& {\color{node_1}  0} &  
& {\color{node_1}  0} &  
& {\color{node_1}  1} &  
& {\color{node_1}  1} &  
& {\color{node_1}  1} &  
& {\color{node_1}  1} &  
\multicolumn{1}{c|}{} 
 \\ \cline{2-17} 
& \rotatebox[origin=c]{180}{$\Lsh$} & \multicolumn{1}{|c|}{\cellcolor {gray!30} 1} 
& \rotatebox[origin=c]{180}{$\Lsh$} & \multicolumn{1}{|c|}{\cellcolor {gray!30} 0} 
& \rotatebox[origin=c]{180}{$\Lsh$} & \multicolumn{1}{|c|}{\cellcolor {gray!30} 0} 
& \rotatebox[origin=c]{180}{$\Lsh$} & \multicolumn{1}{|c|}{\cellcolor {gray!30} 0} 
& \rotatebox[origin=c]{180}{$\Lsh$} & \multicolumn{1}{|c|}{\cellcolor {gray!30} 0} 
& \rotatebox[origin=c]{180}{$\Lsh$} & \multicolumn{1}{|c|}{\cellcolor {gray!30} 0} 
& \rotatebox[origin=c]{180}{$\Lsh$} & \multicolumn{1}{|c|}{\cellcolor {gray!30} 0} 
& \rotatebox[origin=c]{180}{$\Lsh$} & \multicolumn{1}{|c|}{\cellcolor {gray!30} 0} 
\\ \cline{1-17} 
\multirow{2}{*}{{\color{node_2}  $p_2=124$}}   
& {\color{node_2}  0} &  
& {\color{node_2}  1} &  
& {\color{node_2}  1} &  
& {\color{node_2}  1} &  
& {\color{node_2}  1} &  
& {\color{node_2}  1} &  
& {\color{node_2}  0} &  
& {\color{node_2}  0} &  
\multicolumn{1}{c|}{} 
 \\ \cline{2-17} 
& \rotatebox[origin=c]{180}{$\Lsh$} & \multicolumn{1}{|c|}{\cellcolor {gray!30} 0} 
& \rotatebox[origin=c]{180}{$\Lsh$} & \multicolumn{1}{|c|}{\cellcolor {gray!30} 0} 
& \rotatebox[origin=c]{180}{$\Lsh$} & \multicolumn{1}{|c|}{\cellcolor {gray!30} 0} 
& \rotatebox[origin=c]{180}{$\Lsh$} & \multicolumn{1}{|c|}{\cellcolor {gray!30} 0} 
& \rotatebox[origin=c]{180}{$\Lsh$} & \multicolumn{1}{|c|}{\cellcolor {gray!30} 0} 
& \rotatebox[origin=c]{180}{$\Lsh$} & \multicolumn{1}{|c|}{\cellcolor {gray!30} 0} 
& \rotatebox[origin=c]{180}{$\Lsh$} & \multicolumn{1}{|c|}{\cellcolor {gray!30} 0} 
& \rotatebox[origin=c]{180}{$\Lsh$} & \multicolumn{1}{|c|}{\cellcolor {gray!30} 0} 
\\ \cline{1-17} 
\multirow{2}{*}{{\color{node_3}  $p_3=217$}}   
& {\color{node_3}  1} &  
& {\color{node_3}  1} &  
& {\color{node_3}  0} &  
& {\color{node_3}  1} &  
& {\color{node_3}  1} &  
& {\color{node_3}  0} &  
& {\color{node_3}  0} &  
& {\color{node_3}  1} &  
\multicolumn{1}{c|}{} 
 \\ \cline{2-17} 
& \rotatebox[origin=c]{180}{$\Lsh$} & \multicolumn{1}{|c|}{\cellcolor {gray!30} 1} 
& \rotatebox[origin=c]{180}{$\Lsh$} & \multicolumn{1}{|c|}{\cellcolor {gray!30} 1} 
& \rotatebox[origin=c]{180}{$\Lsh$} & \multicolumn{1}{|c|}{\cellcolor {gray!30} 0} 
& \rotatebox[origin=c]{180}{$\Lsh$} & \multicolumn{1}{|c|}{\cellcolor {gray!30} 1} 
& \rotatebox[origin=c]{180}{$\Lsh$} & \multicolumn{1}{|c|}{\cellcolor {gray!30} 1} 
& \rotatebox[origin=c]{180}{$\Lsh$} & \multicolumn{1}{|c|}{\cellcolor {gray!30} 0} 
& \rotatebox[origin=c]{180}{$\Lsh$} & \multicolumn{1}{|c|}{\cellcolor {gray!30} 0} 
& \rotatebox[origin=c]{180}{$\Lsh$} & \multicolumn{1}{|c|}{\cellcolor {gray!30} 1} 
\\ \hline
\multirow{2}{*}{{\color{node_4}  $p_4=222$}}   
& {\color{node_4}  1} &  
& {\color{node_4}  1} &  
& {\color{node_4}  0} &  
& {\color{node_4}  1} &  
& {\color{node_4}  1} &  
& {\color{node_4}  1} &  
& {\color{node_4}  1} & 
& {\color{node_4}  0} &  
\multicolumn{1}{c|}{} 
 \\ \cline{2-17} 
& \rotatebox[origin=c]{180}{$\Lsh$} & \multicolumn{1}{|c|}{\cellcolor {gray!30} 1} 
& \rotatebox[origin=c]{180}{$\Lsh$} & \multicolumn{1}{|c|}{\cellcolor {gray!30} 1} 
& \rotatebox[origin=c]{180}{$\Lsh$} & \multicolumn{1}{|c|}{\cellcolor {gray!30} 0} 
& \rotatebox[origin=c]{180}{$\Lsh$} & \multicolumn{1}{|c|}{\cellcolor {gray!30} 1} 
& \rotatebox[origin=c]{180}{$\Lsh$} & \multicolumn{1}{|c|}{\cellcolor {gray!30} 1} 
& \rotatebox[origin=c]{180}{$\Lsh$} & \multicolumn{1}{|c|}{\cellcolor {gray!30} 0} 
& \cellcolor {yellow!80} $\rotatebox[origin=c]{180}{$\Lsh$}_{\color{node_4} 1}$& \multicolumn{1}{|c|}{\cellcolor {gray!30} 0} 
& $\rotatebox[origin=c]{180}{$\Lsh$}_{\color{node_4} 1}$ & \multicolumn{1}{|c|}{\cellcolor {gray!30} 1} 
\\ \cline{1-17} 
\multirow{2}{*}{{\color{node_5}  $p_5=86$}}   
& {\color{node_5}  0} &  
& {\color{node_5}  1} &  
& {\color{node_5}  0} &  
& {\color{node_5}  1} &  
& {\color{node_5}  0} &  
& {\color{node_5}  1} &  
& {\color{node_5}  1} &  
& {\color{node_5}  0} &  
\multicolumn{1}{c|}{} 
 \\ \cline{2-17} 
& \rotatebox[origin=c]{180}{$\Lsh$} & \multicolumn{1}{|c|}{\cellcolor {gray!30} 0} 
& \rotatebox[origin=c]{180}{$\Lsh$} & \multicolumn{1}{|c|}{\cellcolor {gray!30} 0} 
& \rotatebox[origin=c]{180}{$\Lsh$} & \multicolumn{1}{|c|}{\cellcolor {gray!30} 0} 
& \rotatebox[origin=c]{180}{$\Lsh$} & \multicolumn{1}{|c|}{\cellcolor {gray!30} 0} 
& \rotatebox[origin=c]{180}{$\Lsh$} & \multicolumn{1}{|c|}{\cellcolor {gray!30} 0} 
& \rotatebox[origin=c]{180}{$\Lsh$} & \multicolumn{1}{|c|}{\cellcolor {gray!30} 0} 
& \rotatebox[origin=c]{180}{$\Lsh$} & \multicolumn{1}{|c|}{\cellcolor {gray!30} 0} 
& \rotatebox[origin=c]{180}{$\Lsh$} & \multicolumn{1}{|c|}{\cellcolor {gray!30} 0} 
\\ \cline{1-17} 
 \multirow{2}{*}{$\left[~p_0~\right]_{\text{base 2}}$}   
  \\ \cline{2-17} 
& &\multicolumn{1}{|c|}{\cellcolor {gray!90} 1 }    
& &\multicolumn{1}{|c|}{\cellcolor {gray!90} 1 }    
& &\multicolumn{1}{|c|}{\cellcolor {gray!90} 0 }    
& &\multicolumn{1}{|c|}{\cellcolor {gray!90} 1 }    
& &\multicolumn{1}{|c|}{\cellcolor {gray!90} 1 }    
& &\multicolumn{1}{|c|}{\cellcolor {gray!90} 0 }    
& \multicolumn{1}{|c|}{\cellcolor {yellow!80}}  &\multicolumn{1}{|c|}{\cellcolor {gray!90} 0 }    
& &\multicolumn{1}{|c|}{\cellcolor {gray!90} 1 }    
\\ \cline{1-17} 
\end{tabular}} 
\vspace{0.3cm}
 \caption{
This table provides a visual illustration of the intuition behind our proposed Vickrey auction protocol using an example. 
Each of the five colors corresponds to a binary representation of each of the following bids: $p_1=143$, $p_2=124$, $p_3=217$, $p_4=222$, and $p_5=86$. 
The protocol examines each bid's binary representation, starting at the first digit, $j=1$, and then moving right until $j=8$. 
Light grey cells for each digit $j={1,2,\dots,8}$ demonstrate how the protocol conceals information. 
When a bidder's bid falls below the second-highest bid for the current digit, 
they are forced to input `0' as their bids for all subsequent digits. 
The winner conceals information by simulating a bid of `0' for the current digit and is required to input `1' for all subsequent digits. 
Otherwise, the original bid is displayed as the input number. 
The final row outputs the second-highest bid $\left[p_0\right]_{^\text{\tiny base 2}}$ in binary representation, 
which is indicated by the dark grey cell. 
As seen in the yellow cell, 
the winner $\alpha_4$ is willing to pay a price that is at least one bid higher than the output price: 
$\left[p_4\right]_{^\text{\tiny base 2}}  \geq \left[p_0\right]_{^\text{\tiny base 2}} +1 = [1  1  0  1  1  0  1  0]$. 
 (See \Cref{appendix_1} for full details.)
} 
\label{table_example}
\end{table*}

\subsection{Decentralized Setting}\label{sub_Codes_Creation}

In a decentralized setting without a central authority, 
we treat the Vickrey auction protocol as an SMPC problem. 
In this scenario, bidders collaborate to determine the second-highest bid without revealing additional information. 
Our design of the problem assumes the availability of an authenticated public channel for all participants or third-party observers who wish to verify the auction's integrity. 
This channel can be implemented using a public bulletin board or a blockchain.

Throughout the protocol, 
participants must prove their knowledge of discrete logarithms at various stages without disclosing them, thus preventing fraudulent behavior. 
This can be achieved using a zero-knowledge proof, which is commonly employed in related studies including \cite{bag2019seal,hao20092,brandt2005efficient}. 
For instance, Schnorr's signature, found in \cite{schnorr1991efficient}, can be utilized here because it is both concise and non-interactive.

To ensure security and privacy in a decentralized setting, 
the protocol performs all calculations in a finite field $\mathbb{F}$, 
with $g$ defined as a generator in the multiplicative group of $\mathbb{F}$. 
In addition, we define permutations $\phi$ and $\psi$ on the set of bidders $I$, 
where $\phi(\alpha_l)=\alpha_{l-1}$ and $\psi(\alpha_l)=\alpha_{l+1}$, with $\alpha_0=\alpha_n$ and $\alpha_{n+1}=\alpha_1$.

To support the decentralized SMPC aspect of the protocol, 
each bidder $\alpha_l$ is responsible for generating their own codes by choosing non-zero values for $a_{l, i, j}$, $c_{l, i}$, and $e_{l, j}$. 
The indices $i \in \{1, 2, \dots, n\}$ and  $j \in \{1, 2, \dots, k\}$ are used to provide encryption for different individuals and for each digit in the binary representation, respectively. 
It is essential to note that the codes $a_{l, l, j}$, $c_{l, i}$, and $e_{l, j}$ are considered to be private codes because $\alpha_l$ never shares them with anyone else.

In the following subsections, 
the codes $a_{l,i,j}$ and $e_{l,j}$ are used for performing various steps, 
such as generating keys and sharing bids, 
while the code $c_{l,i}$ is used for bid commitment and verification of the final price. 
We will also discuss their specific functions in greater detail.

\subsection{Key Generation}\label{sub_Keys_Generations}

To ensure auction privacy, 
each bidder $\alpha_l$ uses the codes created in \Cref{sub_Codes_Creation} to generate key values, check keys, and fake keys. 
These cryptographic tools play a critical role in calculating the correct final payment and maintaining the confidentiality of each bidder’s information.

Key values enable bidders to secretly verify their status as winners by ensuring that they are the only one who had bid `1' for a specific digit $j$.
This process is further explained in \Cref{sub_Checking_for_a_Sole_Participant}. 
Check keys and fake keys are used to compute the final payment price of the auction without revealing their information to anyone else. 
Further details are provided in \Cref{sub_2nd_Price}.

The key generation process is detailed in Algorithm \ref{keys_generation_algo}, 
where key values, check keys, and fake keys are computed iteratively in a ring transfer among the bidders. 
An example of this process using five bidders is illustrated in \autoref{figure_ring}.

\begin{figure}
\vspace{-0.6cm}
\centering
\subfloat[]
{\includegraphics[width=0.4\textwidth]{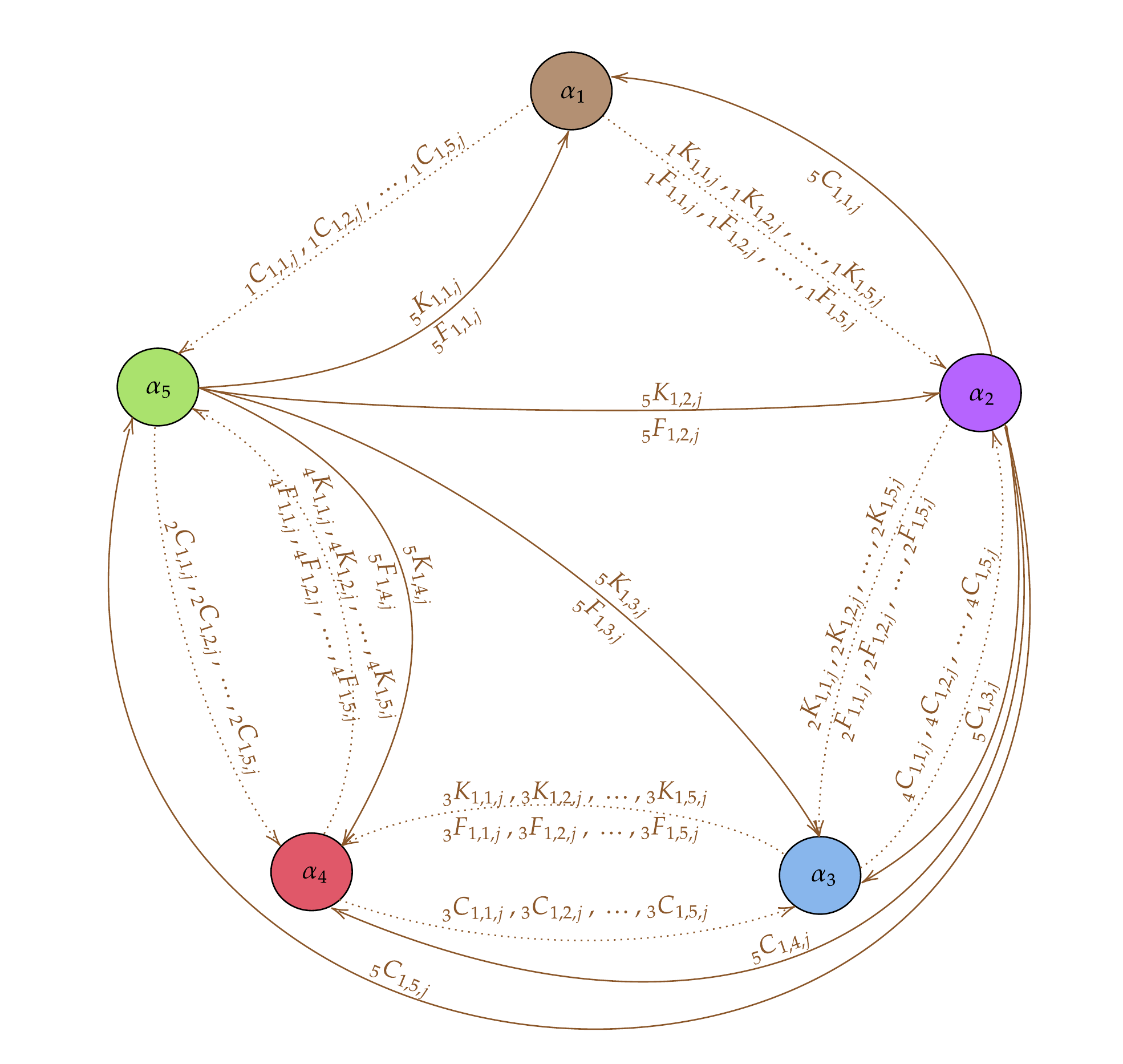}
\label{figure_ring_1}}
\qquad
\qquad
\subfloat[]
{\includegraphics[width=0.4\textwidth]{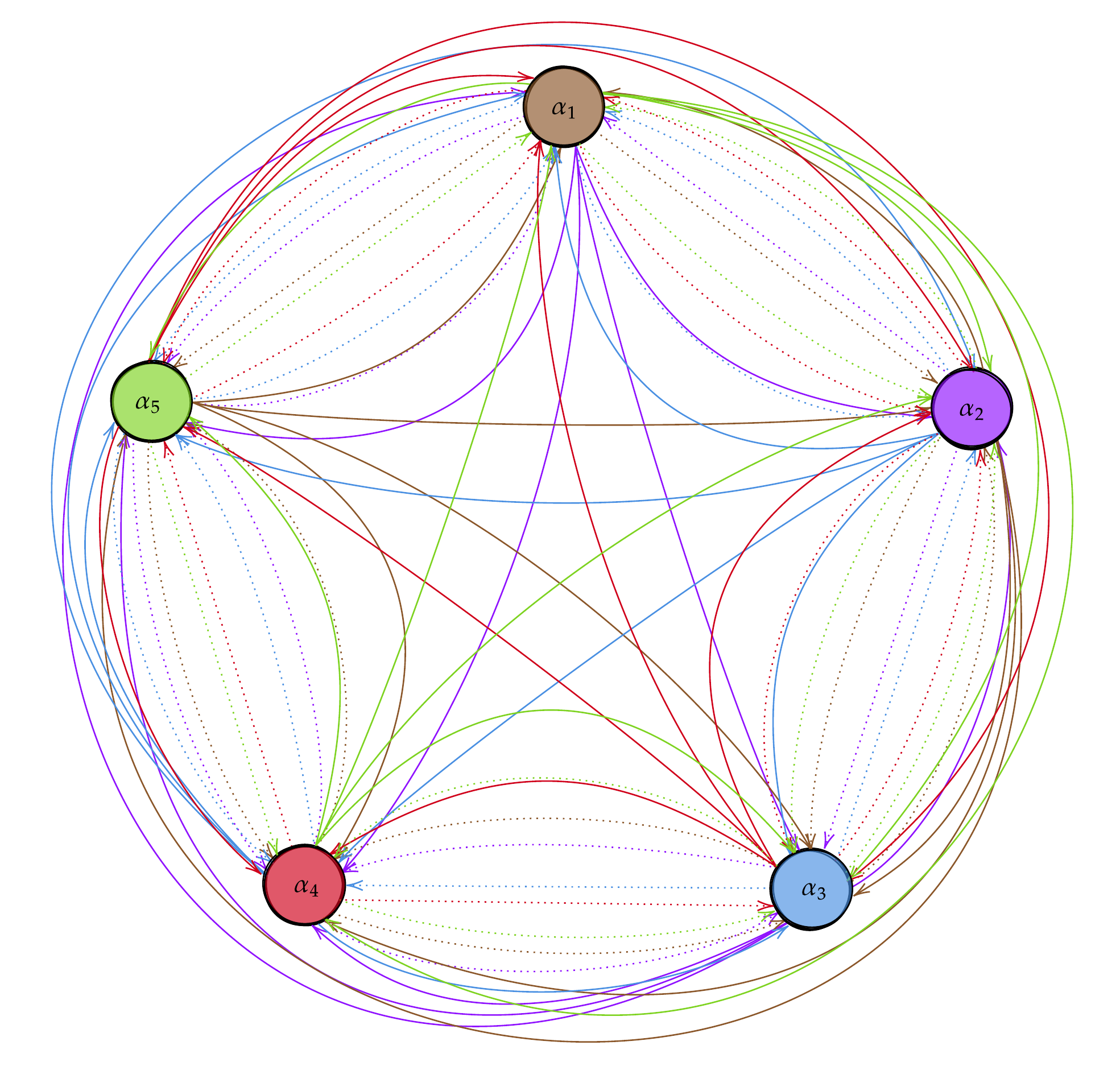}
 \label{figure_ring_2}}
\caption{
Key generation among five bidders using a ring transfer. 
(a) depicts the transformation of ${ }_r K_{1,i,j}$, ${ }_r F_{1,i,j}$ and ${ }_r C_{1,i,j}$ as initially generated by bidder $\alpha_1$, while (b) illustrates the process of key generation using a ring transfer among the five bidders $\alpha_l$ where $l=1,2,\dots,5$. The contributions of each bidder are distinguished by different colors, representing the path of initial values generated by each bidder.
}
\label{figure_ring}
\end{figure}
Initially,
each bidder $\alpha_l$ computes 
\begin{subequations}
\begin{align}
{ }_1 K_{l,i,j}&=g^{\left(a_{l,i,j}-\sum_{u \neq i} a_{l,u,j}\right) e_{l,j}} ~ \text{for $\forall i$ and $\forall j$} \label{eq_1a_proof} \\
{ }_1 F_{l,i,j}&=g^{\left((2n-3)a_{l,i,j}-\sum_{u \neq i} a_{l,u,j}\right) e_{l,j}^2}  ~ \text{for $\forall i$ and $\forall j$} \\
{ }_1 C_{l,i,j}&=g^{\left(a_{l,i,j}-\sum_{u \neq i} a_{l,u,j}\right) e_{l,j}^2}  ~ \text{for $\forall i$ and $\forall j$}
\end{align}
\end{subequations}
based on their private codes created in \Cref{sub_Codes_Creation}. 
Then, bidder $\alpha_l$ sends ${ }_1 K_{l,i,j}$ and ${ }_1 F_{l,i,j}$ to the next bidder $\psi(\alpha_l)$,
and ${ }_1 C_{l,i,j}$ to the previous bidder $\phi(\alpha_l)$ in the ring transfer.
Upon receiving ${ }_r K_{u,i,j}$ or ${ }_r F_{u,i,j}$ from bidder $\phi(\alpha_l)$, 
 bidder $\alpha_l$ computes 
\begin{equation} \label{eq_2_proof}
{ }_{r+1} K_{u,i,j}=\left({ }_{r} K_{u,i,j}\right)^{e_{l,j}}
\end{equation}
or 
\begin{equation}
{ }_{r+1} F_{u,i,j}=\left({ }_{r} F_{u,i,j}\right)^{e_{l,j}^2} \,,
\end{equation}
using their private codes.
If the iteration count is $r<n-1$, 
bidder $\alpha_l$ sends the computed values ${ }_{r+1} K_{u,i,j}$ or ${ }_{r+1} F_{u,i,j}$ to the next bidder $\psi(\alpha_l)$, 
or to the original bidder $i$ if $r = n-1$.
Similarly, when bidder $\alpha_l$ receives ${ }_{r'} C_{v,i,j}$ from bidder $\psi(\alpha_l)$, 
 they use their private code to compute 
 \begin{equation}
{ }_{r'+1} C_{v,i,j}=\left({ }_{r'} C_{v,i,j}\right)^{e_{l,j}^2}
\end{equation}
and send them to bidder $\phi(\alpha_l)$ if the $r < n-1$, 
or to the original bidder $i$ if $r=n-1$.
This process continues until the last bidder in the ring transfer,
set $K_{u,l,j}:={ }_n K_{u,l,j}$, 
$C_{v,l,j}:={ }_n C_{v,l,j} $, 
and $F_{u,l,j}:={ }_n F_{u,l,j}$ as their final values,
and committed $K_{u,l,j}$ to the seller by sending a cryptographic hash. 
Finally, each bidder $\alpha_l$ computes the key values 
\begin{equation}\label{K_formula}
K_{l,j}=\prod_{u=1}^n K_{u,l,j} \,,
\end{equation} 
check keys 
\begin{equation}\label{C_formula}
C_{l,j}=\prod_{v=1}^n C_{v,l,j}
\end{equation}
and fake keys 
\begin{equation}\label{F_formula}
F_{l,j}=\prod_{u=1}^n F_{u,l,j}
\end{equation} 
for each $j$.

\bigskip
\begin{algorithm}[h]
\caption{Key Generation}
\label{keys_generation_algo}
\DontPrintSemicolon
\SetAlgoLined

\KwInput{codes $a_{l,i,j}$ and $e_{l,j}$ for $\forall~l \in\{1,2, \dots, n\}$, $\forall~i \in\{1,2, \dots, n\}$  and $\forall~j\in\{1,2, \dots, k\}$.}
\KwOutput{keys $K_{l,j}$, check keys $C_{l,j}$ and fake keys $F_{l,j}$ $\forall~l \in\{1,2, \dots, n\}$ and $\forall~j\in\{1,2, \dots, k\}$.}

\vspace{0.1cm}
\Initialization{$r \gets 1$}{
    Send  ${ }_1 F_{l,i,j} \gets g^{\left((2n-3)a_{l,i,j}-\sum_{u \neq i} a_{l,u,j}\right) e_{l,j}^2}$
    and ${ }_1 K_{l,i,j} \gets g^{\left(a_{l,i,j}-\sum_{u \neq i} a_{l,u,j}\right) e_{l,j}}$ to bidder $\psi(\alpha_l)$ \; 
    Send ${ }_1 C_{l,i,j} \gets g^{\left(a_{l,i,j}-\sum_{u \neq i} a_{l,u,j}\right) e_{l,j}^2}$ to bidder $\phi(\alpha_l)$ for each bidder $\alpha_l$, each $i$ and each $j$
}

\vspace{0.1cm}
\Upon{ ${ }_r K_{u,i,j}$, ${ }_r F_{u,i,j}$ or ${ }_{r'} C_{v,i,j}$}{
    \vspace{0.05cm}
Bidder $\alpha_l$ computes one  of the following updates:\;
     ${ }_{r+1} K_{u,i,j} \gets \left({ }_r K_{u,i,j}\right)^{e_{l,j}}$\;
     ${ }_{r+1} F_{u,i,j} \gets \left({ }_r F_{u,i,j}\right)^{e_{l,j}^2}$\;
     ${ }_{r'+1} C_{v,i,j} \gets \left({ }_{r'} C_{v,i,j}\right)^{e_{l,j}^2}$\;
     \vspace{0.1cm}
    \eIf{$r<n-1$ or $r'<n-1$}{
        Send  ${ }_{r+1} K_{u,i,j}$ or ${ }_{r+1} F_{u,i,j}$ to bidder $\psi(\alpha_l)$ \;
        or Send ${ }_{r'+1} C_{v,i,j}$ to bidder $\phi(\alpha_l)$ \;
    }{
        Send  ${ }_{n} K_{u,i,j}$, ${ }_{n} F_{u,i,j}$, or ${ }_{n} C_{v,i,j}$ to bidder $\alpha_i$ \; 
        Bidder $\alpha_l$ sets one of the following:\;
         $K_{u,l,j} \gets { }_{n} K_{u,l,j}$\;
         $F_{u,l,j} \gets { }_{n} F_{u,l,j}$\; 
         $C_{v,l,j} \gets { }_n C_{v,l,j}$\;
         Bidder $\alpha_l$ sends a hash of $K_{u,l,j}$ to the seller \;
    }
}

\return{the following for each bidder $\alpha_l$ and each $j$}{
       $K_{l,j} \gets \prod_{u=1}^n K_{u,l,j}$\;
       $F_{l,j} \gets \prod_{u=1}^n F_{u,l,j}$ \;
       $C_{l,j} \gets \prod_{v=1}^n C_{v,l,j}$}

\end{algorithm} 
\bigskip

\begin{lemma}\label{lemma_1}
The formulas presented in Equations \ref{K_formula}, \ref{C_formula} and \ref{F_formula} can be represented as
\begin{subequations} \label{form1KCF}
    \begin{align}\hspace{-0.1cm}
K_{l,j} &=g^{\left(\sum_{u=1}^n a_{u,l,j}-\sum_{h \neq l} \sum_{u=1}^n a_{u,h,j}\right) \prod_{u=1}^n e_{u,j}}; \label{form1KCF:1} \\                      
C_{l,j} &=g^{\left(\sum_{u=1}^n a_{u,l,j}-\sum_{h \neq l} \sum_{u=1}^n a_{u,h,j}\right) \prod_{u=1}^n e_{u,j}^2}; \label{form1KCF:2} \\
F_{l,j} &=g^{\left((2n-3)\sum_{u=1}^n a_{u,l,j}-\sum_{h \neq l} \sum_{u=1}^n a_{u,h,j}\right) \prod_{u=1}^n e_{u,j}^2}. \label{form1KCF:3}            
    \end{align}
\end{subequations}
\end{lemma}

\begin{proof}
Using a recursive substitution of \autoref{eq_2_proof} in the ring transfer of the key generation algorithm and applying the initial condition of \autoref{eq_1a_proof}.
We can simplify \autoref{K_formula} and obtain
\begin{equation}\label{K_formula_proof_1}
K_{l,j}=\prod_{u=1}^n K_{u,l,j} 
=\prod_{u=1}^n  g^{\left(a_{u,l,j}-\sum_{h \neq l} a_{u,h,j}\right) \prod_{u=1}^n e_{u,j}}
\end{equation}
Applying the product-to-sum rule, \autoref{K_formula_proof_1} becomes
\begin{equation}  \label{form1K}
              K_{l,j}=g^{\left(\sum_{u=1}^n a_{u,l,j}-\sum_{h \neq l} \sum_{u=1}^n a_{u,h,j}\right) \prod_{u=1}^n e_{u,j}}; 
\end{equation}
Similarly, we have expressions 
\begin{equation} \label{form1C}
    \begin{aligned}
        C_{l,j}&=\prod_{u=1}^n C_{u,l,j} \\ &=\prod_{u=1}^n  g^{\left(a_{u,l,j}-\sum_{h \neq l} a_{u,h,j}\right) \prod_{u=1}^n e_{u,j}^2}\\
              &=g^{\left(\sum_{u=1}^n a_{u,l,j}-\sum_{h \neq l} \sum_{u=1}^n a_{u,h,j}\right) \prod_{u=1}^n e_{u,j}^2};
    \end{aligned}
\end{equation}
and
\begin{equation} \label{form1F}
    \begin{aligned}
        F_{l,j}&=\prod_{u=1}^n F_{u,l,j} \\&=\prod_{u=1}^n  g^{\left((2n-3)a_{u,l,j}-\sum_{h \neq l} a_{u,h,j}\right) \prod_{u=1}^n e_{u,j}^2}\\
              &=g^{\left((2n-3)\sum_{u=1}^n a_{u,l,j}-\sum_{h \neq l} \sum_{u=1}^n a_{u,h,j}\right) \prod_{u=1}^n e_{u,j}^2}.
    \end{aligned}
\end{equation}
\end{proof}
The equations presented in the previous proof demonstrate how keys are generated from codes $a_{u,l,j}$ and $e_{u,j}$. 

In \Cref{sub_Price_Determination}, 
we will provide a detailed explanation of the purpose and function of these key representations and how they help ensure auction security and privacy.

\subsection{Bid Commitment}\label{sub_Bid_Commitment}

To ensure the integrity of the final outcome and provide transparency in the auction, 
a bid commitment step is employed. 
This step involves committing all bids ${p_l}$ for all bidders to a public channel while keeping the original values confidential, 
which serves as proof that the final output price is a legitimate bid.
We choose the ring transfer method to commit bids as it is superior to directly publishing hash values. 
Direct publishing of hash values can be vulnerable to brute-force attacks due to the limited number of possible bids. 

To accomplish this, 
each bidder $\alpha_l$ follows the steps outlined in Algorithm \ref{bid_commitment_algo}.
Initially, the bidder computes
\begin{equation}\label{price_commitment_1}
{ }_1 \mathbf{P}=g^{{p_l}  c_{l,1}}
\end{equation} 
and send it to bidder $\psi(\alpha_l)$.
Upon receiving ${ }_r \mathbf{P}$,
bidder $\alpha_l$ then needs to complete the ring transfer by passing the computed value
\begin{equation}\label{price_commitment_2} 
{ }_{r+1} \mathbf{P}={ }_r \mathbf{P}^{c_{l,r+1}}
\end{equation} 
to the next bidder $\psi(\alpha_l)$.
This process continues until the last bidder in the ring transfer, 
at which point bidder $\alpha_l$ publishes 
a final commitment to the public channel\begin{equation}\label{price_commitment_3}
\mathbf{P}_{l}={ }_{n} \mathbf{P}\,.
\end{equation}

\bigskip
\begin{algorithm}[h]
\caption{Bid Commitment}
\label{bid_commitment_algo}
\DontPrintSemicolon
\SetAlgoLined

\KwInput{bids $p_l$ and codes $c_{l,1}$ for $\forall~l \in\{1,2, \dots, n\}$.}
\KwOutput{published $\mathbf{P}_{l}$ for $\forall~l \in\{1,2, \dots, n\}$.}
\vspace{0.1cm}
\Initialization{$r \gets 1$}{
    Send ${ }_1 \mathbf{P} \gets g^{{p_l} c_{l,1}}$ to bidder $\psi(\alpha_l)$ for each $\alpha_l$\;
}
\vspace{0.1cm}
\Upon{${ }_r \mathbf{P}$}{
    Bidder $\alpha_l$ compute ${ }_{r+1} \mathbf{P} \gets { }_r \mathbf{P}^{c_{l,r+1}}$\;
    \eIf{$r < n - 1$}{
        Send ${ }_{r+1} \mathbf{P}$ to bidder $\psi(\alpha_l)$\;
    }{
        Publish $\mathbf{P}_{l} \gets { }_{n} \mathbf{P}$ to the public channel\;
    }
}
\end{algorithm} 
\bigskip


In \Cref{sub_Final_Price_Verification}, we will explain how we utilize these committed bids to verify the final output price's legitimacy.

\subsection{Bid Sharing}\label{sub_Bids_Sharing}

Bid sharing is presented in Algorithm \ref{bids_sharing_algo}.  
In this step of the auction protocol,
each bidder $\alpha_l$ sends their codes $a_{l,i,j}$ to bidder $\alpha_i$ for all $i \neq l$. 
Using these codes, each bidder computes the indicators
\begin{equation}\label{Y_form}
Y_{l,j}=\sum_{i=1}^n a_{i,l,j}
\end{equation}
or
$
N_{l,j} = -Y_{l,j}
$
for each digit $j$ of their bid in binary representation. 
These indicators correspond to the $j$th digit of the bid $\left[~p_{l,1}~p_{l,2}\dots p_{l,k}~\right]_{\text{base 2}}=p_l$ that bidder $\alpha_l$ committed to in the bid commitment step (\Cref{sub_Bid_Commitment}), 
which is denoted by $p_{l,j}$. 
Specifically, the indicators $Y_{l,j}$ and $N_{l,j}$ represent a bid of `$p_{l,j}=1$' and `$p_{l,j}=0$', respectively.

In the next step, 
each bidder $\alpha_l$ randomly chooses $b_{i,l,j}$ for each $i$ and $j$, 
such that 
\begin{equation} \label{def_share_bids}
    \sum_{i=1}^n b_{i,l,j}=\begin{cases}
    {Y_{l,j},} & {\text{if}}\ p_{l,j}=1 \\ 
    {N_{l,j},} & {\text{otherwise.}}
    \end{cases}
\end{equation}
Once bidders generate their random variables, 
$\alpha_l$ sends corresponding $b_{i,l,j}$ values to each bidder $\alpha_i$ for all $i\neq l$, 
and receives $b_{l,i,j}$ from each $\alpha_i$ for all $i\neq l$ and $j$. 
This completes the bid sharing step of the auction protocol.

\bigskip
\begin{algorithm}[h]
\caption{Bid Sharing for Bidder $\alpha_l$}
\label{bids_sharing_algo}
\SetKwFunction{Random}{Random} 
\DontPrintSemicolon
\SetAlgoLined

\KwInput{$a_{l,i,j}$ for $\forall~i$ and $\forall~j$.}
\KwOutput{$\alpha_l$ receives bids $b_{l,i,j}$ for $\forall i\neq l$ and $\forall j$.}
\vspace{0.1cm}
Send $a_{l,i,j}$ to bidder $\alpha_i$ for each $i \neq l$ and each $j$ \;
Compute $Y_{l,j} \gets \sum_{i=1}^n a_{i,l,j}$ and $N_{l,j} \gets -Y_{l,j}$ \;
$b_{i,l,j} \gets \Random(\mathbb{F})$ for $i \in\{1,2, \dots, n-1\}$\;
    
\eIf{$p_{l,j}=1$}{
$b_{n,l,j} \gets Y_{l,j} - \sum_{i=1}^{n-1} b_{i,l,j}$ for each $j$ \;

            }{
$b_{n,l,j} \gets N_{l,j} - \sum_{i=1}^{n-1} b_{i,l,j}$ for each $j$ \;

        }
Send $b_{i,l,j}$ to bidder $\alpha_i$ for each $i \neq l$ and each $j$.

\vspace{0.1cm}
\Function{\Random $(\text{a finite field } \mathbb{F})$}{
    $x \gets$ Randomly choose elements from $\mathbb{F}$ \;
    \Return{$x$}}
\end{algorithm} 
\bigskip

\begin{remark}
By using the formula for $Y_{l,j}$ defined in \autoref{Y_form}, 
the formulas for $K_{l,j}$, $C_{l,j}$, and $F_{l,j}$ in Lemma \autoref{lemma_1} can be simplified as follows:
\begin{subequations} \label{form2}
\begin{align}
K_{l,j}&=g^{\left(Y_{l,j}-\sum_{h \neq l} Y_{h,j}\right) \prod_{u=1}^n e_{u,j}};\\ 
C_{l,j}&=g^{\left(Y_{l,j}-\sum_{h \neq l} Y_{h,j}\right) \prod_{u=1}^n e_{u,j}^2};\\
F_{l,j}&=g^{\left((2n-3)Y_{l,j}-\sum_{h \neq l} Y_{h,j}\right) \prod_{u=1}^n e_{u,j}^2}.
\end{align}
\end{subequations}
\end{remark}

\subsection{Price Determination}\label{sub_Price_Determination}

In this step, our aim is to determine the second-highest bid without disclosing any additional information of participants.

The algorithm operates by examining the binary representation of the bids, 
beginning with the first digit $j=1$ and progressing towards the right until $j=k$.
At each digit $j$, the process is divided into three stages, 
each of which is detailed in separate subsections. 
In the first stage (\Cref{sub_Checking_for_a_Sole_Participant}), 
all participants check if they are the only ones bidding `1' for the current digit.
In the second stage (\Cref{sub_2nd_Price}), 
the protocol computes the second-highest bid while preserving privacy by using either a fake key or a check key, depending on the result of the first stage. 
If only one participant is bidding `1' for the current digit, 
the protocol uses a fake key that makes it seem as if nobody is bidding `1' at the current digit. Otherwise, the protocol uses a check key to preserve privacy.
Finally, in the third stage (\Cref{sub_Bids_Adjustment}), 
all participants adjust their bids based on the key used in the second stage to ensure privacy while accurately determining the second-highest bid.

\subsubsection{Checking for a Sole Participant}\label{sub_Checking_for_a_Sole_Participant}
In this step, the values $B_{j}$ and $P_{j}$ are introduced based on the discrete logarithm implemented in the ring transfer, function to maintain bid confidentiality.
$B_{j}$ is designed to help bidders secretly check if they are the sole participant for the current digit $j$. 
If so, their key value $K_{l,j}$ will match $B_{j}$ after the ring transfer. 
This information then allows bidder $\alpha_l$ to use a fake key to manipulate the value of $D_{j}$,
which makes all bids appear as if all bids are `0'.
Moreover, $P_{j}$ will be utilized to calculate the second-highest bid of the current digit $j$ in the following step without disclosing any additional information of participants.

The calculation of $B_{j}$ and $P_{j}$ values is performed by the bidders using shared bid values $b_{i,l,j}$ as outlined in Algorithm \ref{ring_transfer_BP_algo}. 
An example of this process is illustrated in \autoref{figure_ring_BP}.
The algorithm begins with bidder $\alpha_l$ computing the following equations
\begin{equation} \label{B_proof_eq_1}
{ }_{1} B_{l,j}=g^{\sum_{i=1}^n b_{l,i,j} e_{l,j}}
\end{equation}
and
\begin{equation} 
{ }_{1} P_{l,j}=g^{\sum_{i=1}^n b_{l,i,j} e_{l,j}^3} \,.
\end{equation}
Bidder $\alpha_l$ then sends ${ }_{1} B_{l,j}$ to bidder $\psi(\alpha_l)$ and ${ }_{1} P_{l,j}$ to bidder $\phi(\alpha_l)$.
Each time a value ${ }_r B_{i,j}$ or ${ }_{r'} P_{i,j}$ is received, 
bidder $\alpha_l$ computes the next value using 
\begin{equation} \label{B_proof_eq_2}
{ }_{r+1} B_{i,j}=\left({ }_{r} B_{i,j}\right)^{e_{l,j}}
\end{equation} 
or
\begin{equation} 
{ }_{r'+1} P_{i,j}=\left({ }_{r'} P_{i,j}\right)^{e_{l,j}^3} \,,
\end{equation}
and sends the computed value to bidder $\psi(\alpha_l)$ or bidder $\phi(\alpha_l)$.
This process is then repeated until the last bidder in the ring transfer is reached.
The final value ${ }_{r+1} B_{i,j}$ or ${ }_{r'+1} P_{i,j}$ is published as $B_{i,j}$ or $P_{i,j}$, respectively.
\begin{figure}
\centering
\subfloat[]
{\includegraphics[width=2.7in]{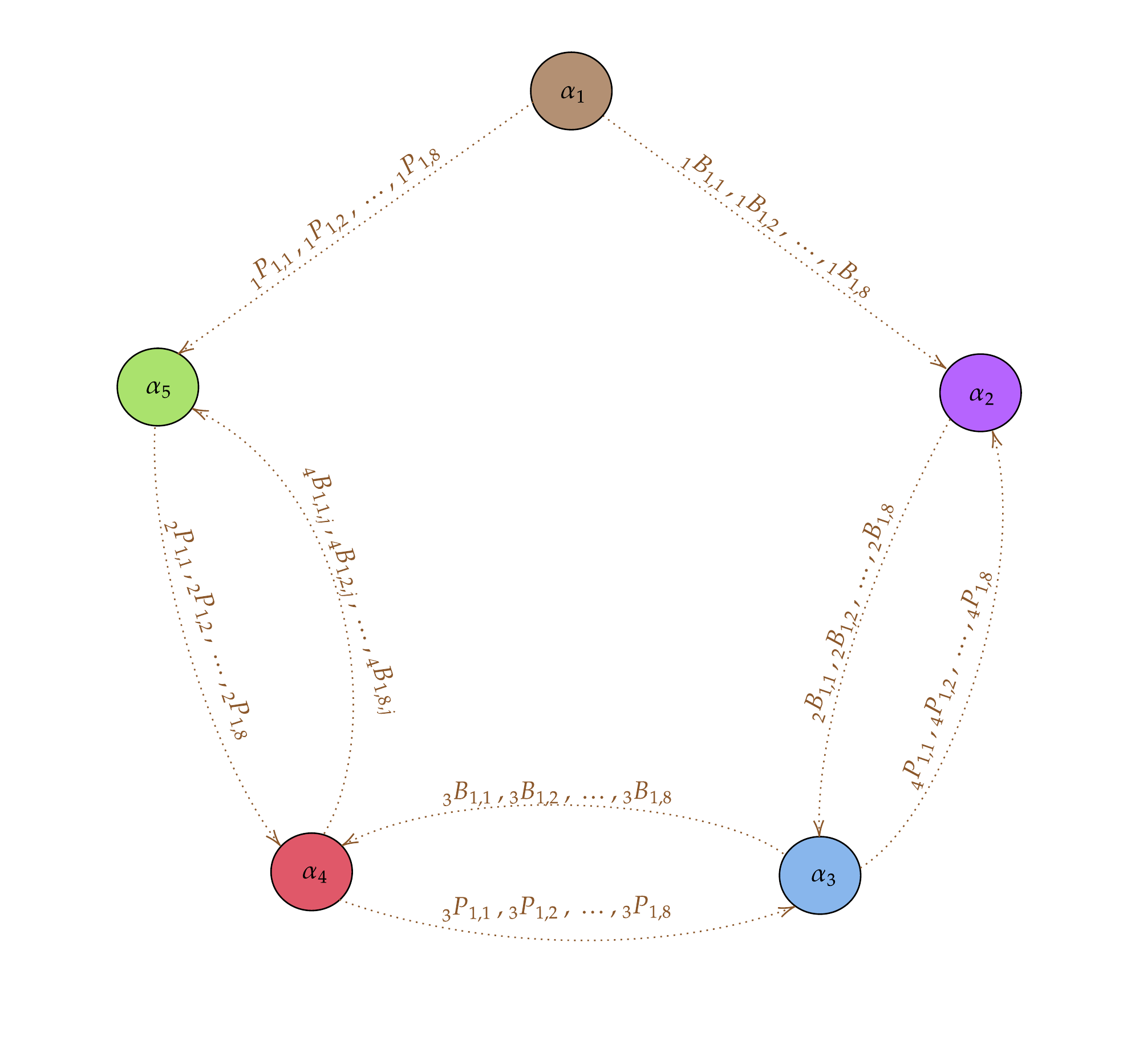}}
\label{figure_ring_1}
\hfil
\subfloat[]
{\includegraphics[width=2.8in]{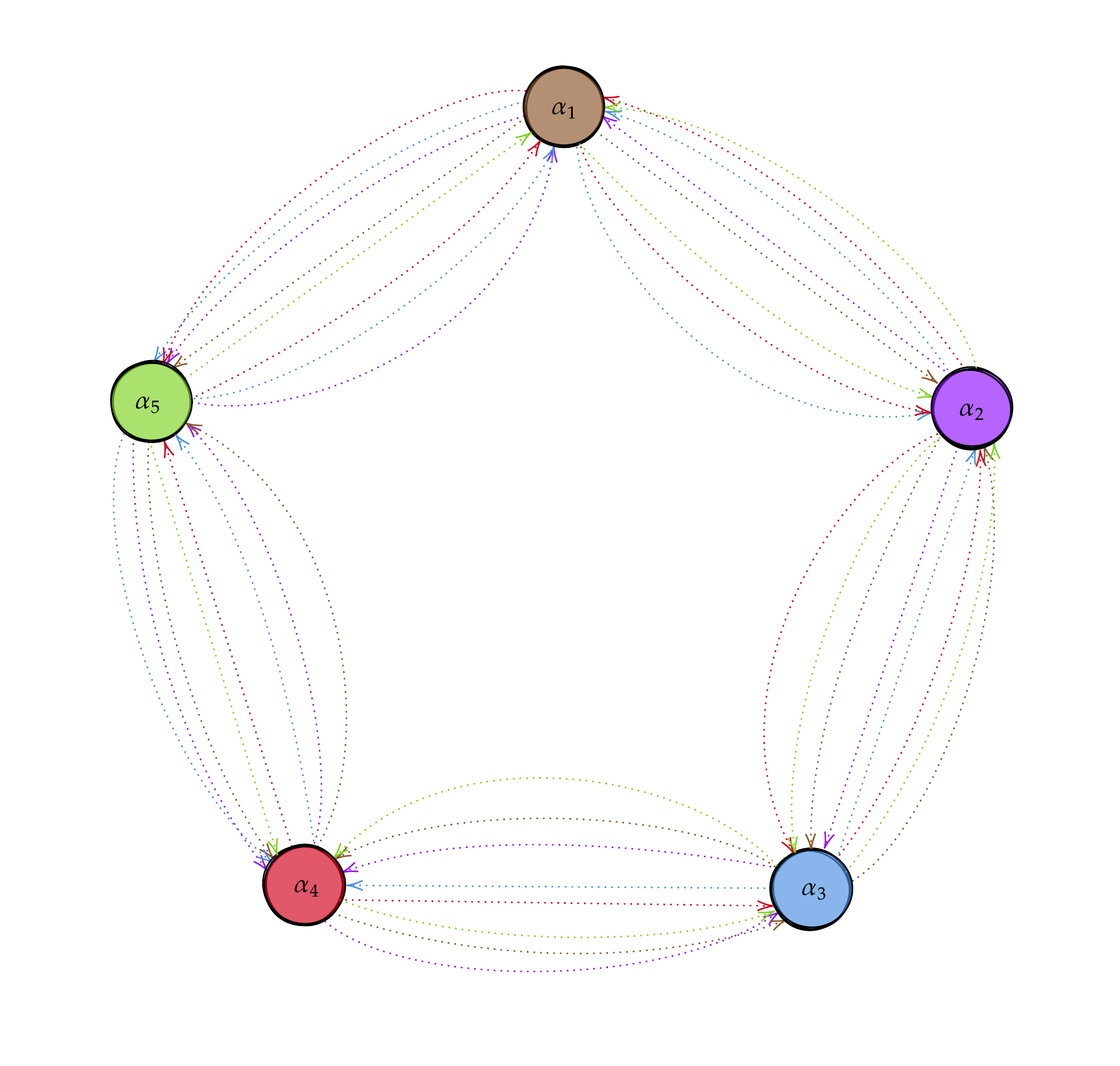}}
 \label{figure_ring_2}
\caption{
Ring transfer for generating $B_j$, $P_j$ and $D_j$ among five bidders. 
(a) depicts the transformation of values ${ }_r B_{1,j}$, ${ }_r P_{1,j}$ and ${ }_r D_{1,j}$, which were initially generated by bidder $\alpha_1$ in the ring transfer.
(b) illustrates the process of how values ${ }_r B_{i,j}$, ${ }_r P_{i,j}$ and ${ }_r D_{i,j}$ are generated by each of the five bidders $\alpha_l$ for $l=1,2,\dots,5$. 
Each bidder's contributions are represented by different colors to distinguish the path of initial values generated by each bidder.
}
\label{figure_ring_BP}
\end{figure}

Finally, each bidder $\alpha_l$ computes the final values $B_j$ and $P_j$ by taking the product of all published $B_{i,j}$ and $P_{i,j}$, 
using the formula
\begin{equation} \label{formula_B_j}
B_j=\prod_{i=1}^n B_{i,j}
\end{equation} 
and
\begin{equation} \label{formula_P_j}
P_j=\prod_{i=1}^n P_{i,j} \,.
\end{equation}

\bigskip
\begin{algorithm}[h]
\caption{Generating $B_j$ and $P_j$}
\label{ring_transfer_BP_algo} 

\SetKwFunction{RingTransferBP}{Algorithm~\ref{ring_transfer_BP_algo}} 
\DontPrintSemicolon
    \SetAlgoLined
    \KwInput{received bids $b_{l,i,j}$ and codes $e_{l,j}$ for $\forall~\alpha_l$, $\forall~i \in\{1,2, \dots, n\}$  and $\forall~j\in\{1,2, \dots, k\}$.}
    \KwOutput{$B_j$, $P_j$}

\vspace{0.1cm}
\Initialization{$r \gets 1$}{  
Send ${ }_{1} B_{l,j} \gets g^{\sum_{i=1}^n b_{l,i,j} e_{l,j}}$ to bidder $\psi(\alpha_l)$ \;
Send ${ }_{1} P_{l,j} \gets g^{\sum_{i=1}^n b_{l,i,j} e_{l,j}^3}$ to bidder $\phi(\alpha_l)$ \;
}

\vspace{0.1cm}
\Upon{ ${ }_r B_{i,j}$ or ${ }_{r'} P_{i,j}$}{
\vspace{0.05cm}
Bidder $\alpha_l$ computes one of the following updates: 
 ${ }_{r+1} B_{i,j} \gets \left({ }_{r} B_{i,j}\right)^{e_{l,j}}$ or 
 ${ }_{r'+1} P_{i,j} \gets \left({ }_{r'} P_{i,j}\right)^{e_{l,j}^3}$ \;

\eIf{$r<n-1$}{ 
  Send ${ }_{r+1} B_{i,j}$  to bidder $\psi(\alpha_l)$\;
            }{
 Publish $B_{i,j} \gets { }_{n} B_{i,j}$ \;
        }      

\eIf{$r'<n-1$}{ 
  Send ${ }_{r'+1} P_{i,j}$ to bidder $\phi(\alpha_l)$\;
            }{
 Publish $P_{i,j} \gets { }_{n} P_{i,j}$ \;
        }
        
    }
    \return{$B_j \gets \prod_{i=1}^n B_{i,j}$ and $P_j \gets \prod_{i=1}^n P_{i,j}$}

\end{algorithm} 
\bigskip

Let $\mathcal{Y}_j\subset I$ be the set of bidders whose $j^{th}$ digit of their bid is `1' ,
and let $\mathcal{N}_j\subset I$ be the set of bidders whose $j^{th}$ digit of their bid is `0'.

\begin{proposition}\label{proposition_2}
If $\alpha_l$ is the sole participant in digit $j$,
then their key value $K_{l,j}$ will match the value of $B_{j}$ as described in Algorithm \ref{ring_transfer_BP_algo}.
\end{proposition}
\begin{proof}

Given the formulas for $B_j$ in \autoref{formula_B_j}
\begin{equation*}  
        B_j=\prod_{i=1}^n B_{i,j}
\end{equation*}
Using a recursive substitution of \autoref{B_proof_eq_2} into \autoref{formula_B_j} and applying the initial condition of \autoref{B_proof_eq_1}, 
we have the following equation
\begin{equation}        \label{eqn_B_j_2}  
        B_j=\prod_{i=1}^n g^{(\sum_{u=1}^n b_{i,u,j}) \prod_{u=1}^n e_{u,j}}
        \end{equation}         
By transforming the product into a sum in the exponential in \autoref{eqn_B_j_2}, we obtain
\begin{equation}\label{eqn_B_j_3}  
B_j=g^{(\sum_{i=1}^n \sum_{u=1}^n b_{i,u,j}) \prod_{u=1}^n e_{u,j}}
\end{equation}
By rearranging \autoref{eqn_B_j_3}, we get
\begin{equation}              \label{eqn_B_j_4}    
B_j=g^{( \sum_{u=1}^n \sum_{i=1}^n b_{i,u,j}) \prod_{u=1}^n e_{u,j}} 
\end{equation}
By applying the formula of $Y$ and $N$ to \autoref{eqn_B_j_4}, we have
\begin{equation}  \label{eqn_B_j_5}  
B_j=g^{( \sum_{\alpha_u\in \mathcal{Y}_j} Y_{u,j}+\sum_{\alpha_v\in \mathcal{N}_j} N_{v,j}) \prod_{u=1}^n e_{u,j}}
\end{equation}
Finally, using the formula of $N$ in \autoref{eqn_B_j_5}, we arrive at
\begin{equation}  \label{eqn_B_j_6}  
        B_j=g^{( \sum_{\alpha_u\in \mathcal{Y}_j} Y_{u,j}-\sum_{\alpha_v\in \mathcal{N}_j} Y_{v,j}) \prod_{u=1}^n e_{u,j}}.  
        \end{equation}
The values of $B_j$ can be further simplified depending on the bids of the participants. Here are two cases:
\begin{enumerate}
\item
If all bidders bid `0' for digit $j$, then $\mathcal{N}_j=I$ and $B_j$ can be simplified to
\begin{equation} \label{eqn_B_j_7}  
    B_j=g^{( -\sum_{v=1}^n Y_{v,j}) \prod_{u=1}^n e_{u,j}} \,.  
\end{equation}
Value $B_j$ holds no significance for the bidders as they cannot decode it with their key value $K_{l,j}$.

\item If $\alpha_l$ is the only bidder who bids `1' in digit $j$, then $B_j$ simplifies to
\begin{equation} \label{eqn_B_j_8}  
    B_j=g^{( Y_{l,j}-\sum_{v\neq l} Y_{v,j}) \prod_{u=1}^n e_{u,j}} \,.  
\end{equation}   
In this case, the value of $B_j$ will equal the bidder's key value $K_{l,j}$ as described in Lemma \autoref{lemma_1}.
\end{enumerate}

\end{proof}

By performing the test in Proposition \autoref{proposition_2}, 
bidder $\alpha_l$ can confirm they are the only participant who bid `1' in the current digit. 
In the next stage (\Cref{sub_2nd_Price}), this bidder can use a fake key to manipulate the calculation of $P_j$ in such a way that it appears as if all bidders bid `0' for this current digit. 

\begin{remark}\label{remark_3}
In addition to the formula for $B_j$, 
we can also simplify the formula for $P_j$ using the same notation as the proof for Proposition \autoref{proposition_2}. 
Specifically, we can write $P_j$ as 
\begin{equation} \label{eqn_P_j_1}  
    P_j=g^{( \sum_{\alpha_u\in \mathcal{Y}_j} Y_{u,j}-\sum_{\alpha_v\in \mathcal{N}_j} Y_{v,j}) \prod_{u=1}^n e_{u,j}^3}.   
\end{equation}
where $Y_{u,j}$ and $e_{u,j}$ are defined as before. 
By simplifying the formula for $P_j$ in \autoref{formula_P_j}, we obtain
\begin{equation} \label{eqn_P_j_1}  
    P_j=g^{( \sum_{\alpha_u\in \mathcal{Y}_j} Y_{u,j}-\sum_{\alpha_v\in \mathcal{N}_j} Y_{v,j}) \prod_{u=1}^n e_{u,j}^3}.   
\end{equation}
We can also simplify $P_j$ further in two cases:
\begin{enumerate}
\item
If all bidders bid `0' for digit $j$, we have
\begin{equation}\label{form1P}
    P_j=g^{( -\sum_{v=1}^n Y_{v,j}) \prod_{u=1}^n e_{u,j}^3}.   
\end{equation}
It's important to note that the simplification of $B_j$ in \autoref{eqn_B_j_7} may not hold any significance for bidders, presented in the first case of the proof in Proposition \autoref{proposition_2}.
However, the values of $P_j$ can provide valuable information to bidders regarding the next step.

\item
If bidder $\alpha_l$ is the only one bidding `1' for digit $j$, then we have
\begin{equation} \label{form2P}
    P_j=g^{( Y_{l,j}-\sum_{v\neq l} Y_{v,j}) \prod_{u=1}^n e_{u,j}^3}.   
\end{equation}
\end{enumerate}
\end{remark}
This simplification in Remark \autoref{remark_3} will be used in the next stage (\Cref{sub_2nd_Price}) to check if there are at least 2 bidders bidding `1' in the $j^{th}$ digit.

\subsubsection{Calculating the Output Price}\label{sub_2nd_Price}

In this step, the calculation of the output price (which is the second-highest bid) is carried out in a confidential manner using Algorithm \ref{2nd_price_algo}.
The result of this computation represents the output price at the current digit and is visible to all participants.
To compute the output price, 
the values of $d_{i,l,j}$ are introduced, which are chosen randomly by the bidders using either their fake key or their check key, depending on their key values, 
as described by the equation
\begin{equation}
    \prod_{i=1}^n d_{i,l,j}=\begin{cases}
{F_{l,j},} & {\text{if}}\ K_{l,j}=B_j \\ 
{C_{l,j},} & {\text{otherwise.}} 
    \end{cases}
\end{equation}
The chosen $d_{i,l,j}$ values are sent to $\alpha_i$ for all $i \neq l$, 
and $\alpha_l$ receives $d_{l,i,j}$ for all $i \neq l$. 

To calculate $D_j$, $\alpha_l$ starts by computing 
\begin{equation}\label{D_ini}
{ }_{1} D_{l,j}=\prod_{i=1}^n d_{l,i,j}^{e_{l,j}}
\end{equation}
and sending it to bidder $\psi(\alpha_l)$. 
When $\alpha_l$ receives ${ }_r D_{i,j}$, 
it computes  
\begin{equation}\label{D_trans}
{ }_{r+1} D_{i,j}=\left({ }_{r} D_{i,j}\right)^{e_{l,j}} \,.
\end{equation}
This process is then repeated until the last bidder in the ring transfer.
If $r<n-1$, 
the computed value is sent to bidder $\psi(\alpha_l)$; 
otherwise it is published as $D_{i,j}$.
Each bidder can then compute 
\begin{equation}\label{def_D}
D_j=\prod_{i=1}^n D_{i,j}
\end{equation}
using the published $D_{i,j}$. 
This process is summarized in Algorithm \ref{ring_transfer_D_algo}.
An example of $D_j$ generation using a ring transfer is shown in \autoref{figure_ring_BP}.

\bigskip
\begin{algorithm}[h]
\caption{Generating $D_j$}
\label{ring_transfer_D_algo} 
\SetKwFunction{RingTransferD}{Algorithm~\ref{ring_transfer_D_algo}} 
\DontPrintSemicolon
    \SetAlgoLined
    \KwInput{received $d_{l,i,j}$ and codes $e_{l,j}$ for $\forall~\alpha_l$, $\forall~i \in\{1,2, \dots, n\}$  and $\forall~j\in\{1,2, \dots, k\}$.}
    \KwOutput{$D_j$}
    
\vspace{0.1cm}
\Initialization{$r \gets 1$}{
Send ${ }_{1} D_{l,j} \gets \prod_{i=1}^n d_{l,i,j}^{e_{l,j}}$ to bidder $\psi(\alpha_l)$} 

\vspace{0.1cm}
    \Upon{ ${ }_r D_{i,j}$}{
   
\eIf{$r<n-1$}{
Send ${ }_{r+1} D_{i,j} \gets \left({ }_{r} D_{i,j}\right)^{e_{l,j}}$ to bidder $\psi(\alpha_l)$\;
            }{
Publish $D_{i,j} \gets { }_{r+1} D_{i,j}$ \;   
            }
    }
    \Return{$D_j \gets \prod_{i=1}^n D_{i,j}$ }

\end{algorithm} 

\begin{figure}
\centering
\subfloat[]
{\includegraphics[width=2.5in]{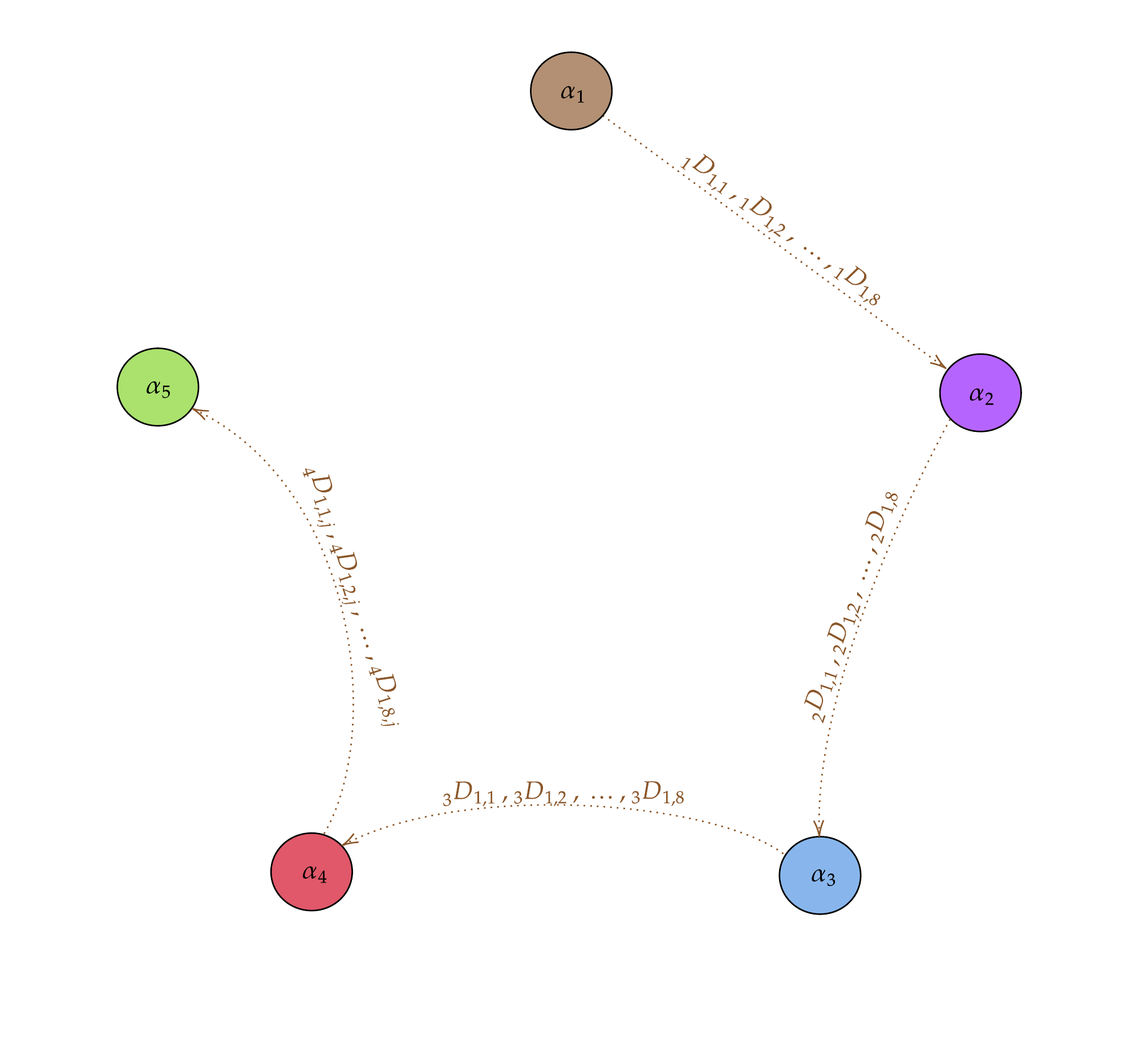}}
\label{figure_ring_1}
\hfil
\subfloat[]
{\includegraphics[width=2.6in]{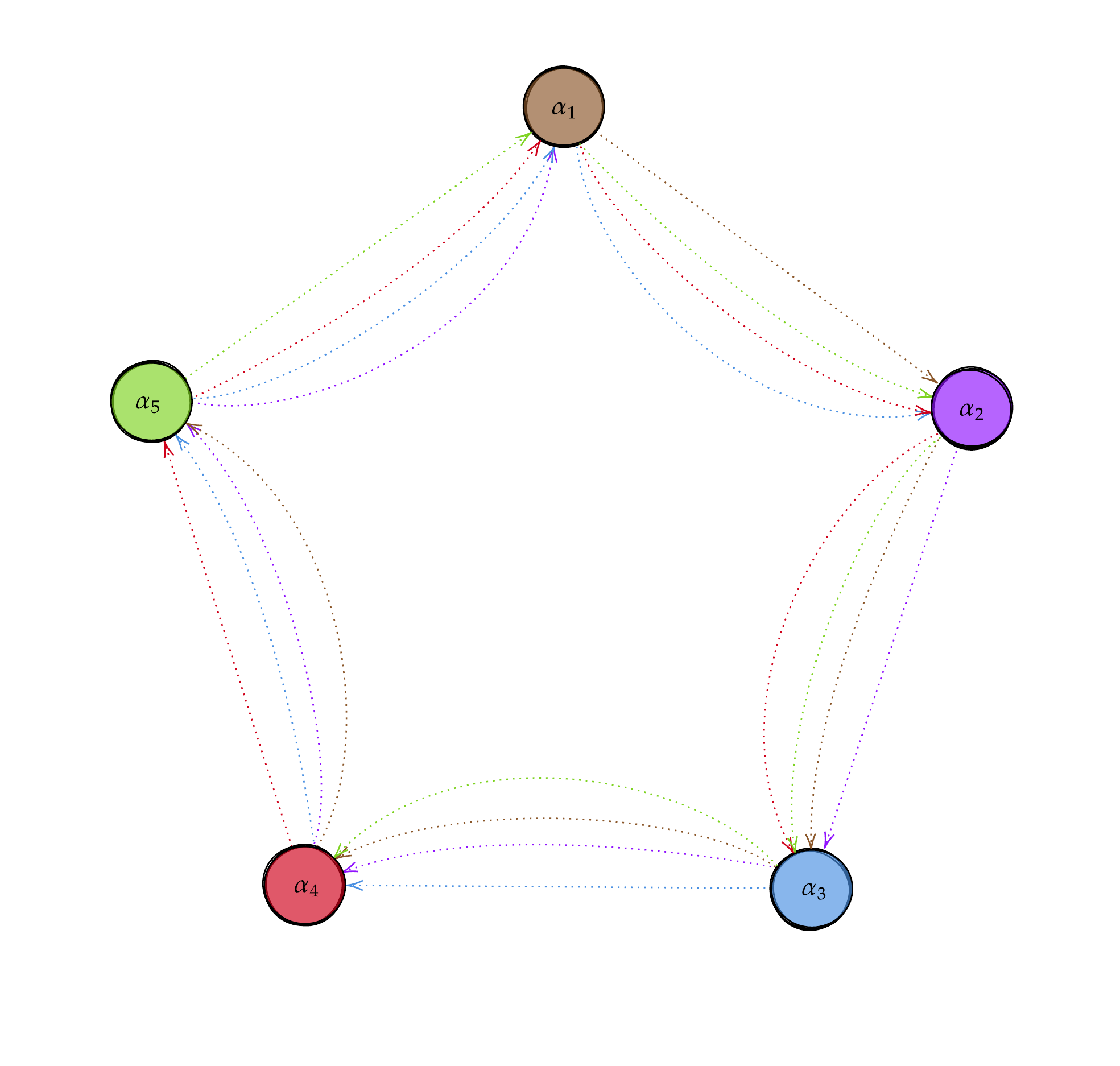}}
 \label{figure_ring_2}
\caption{Ring transfer for generating $D_j$ among five Bidders. 
(a) depicts the transformation of values ${ }_r D_{1,j}$ in the ring transfer, which were initially generated by bidder $\alpha_1$.
(b) illustrates the process of how values ${ }_r D_{l,j}$ are generated by each of the five bidders $\alpha_l$ for $l=1,2,\dots,5$.
The contributions of each bidder are distinguished by different colors, which each represent the path of initial values generated by each bidder.
}
\label{An_example_D_j_generation}
\end{figure}

\begin{proposition}\label{prop_1}
If $D_j = P_j^{n-2}$, 
then $j^{th}$ digit of the binary representation of the final output price $p_{0,j}$ will be equal to `0'.
\end{proposition}

\begin{proof} 
From the definition of $D_j$ in \autoref{def_D}, we have
\begin{equation*}
D_j=\prod_{i=1}^n D_{i,j}
\end{equation*}
Using a recursive substitution of \autoref{D_trans} into \autoref{def_D} and applying the initial condition of \autoref{D_ini}, 
we have 
\begin{equation}\label{case_1_D_j_2}
D_j=\prod_{i=1}^n (\prod_{u=1}^n d_{i,u,j})^{\prod_{u=1}^n e_{u,j}}
\end{equation}
By rearranging \autoref{case_1_D_j_2}, we get
\begin{equation}\label{case_1_D_j_3}
D_j= (\prod_{u=1}^n \prod_{i=1}^n d_{i,u,j})^{\prod_{u=1}^n e_{u,j}}
\end{equation}
There are two cases to consider: the first case is when someone uses a fake key, and the second case is when no one uses a fake key.

Case 1. 
Suppose $\alpha_l$ is the only bidder who bids `1' in the current round and submits a fake key instead of a check key.
Using the formula \ref{form2} for $F$ and $C$, 
\autoref{case_1_D_j_3} becomes
\begin{equation}\label{case_1_D_j_5}
D_j= (F_{l,j}\prod_{h\neq l }C_{h,j})^{\prod_{u=1}^n e_{u,j}}
\end{equation}
By rearranging \autoref{case_1_D_j_5}, we obtain
\begin{equation}\label{case_1_D_j_6}
D_j= F_{l,j}^{\prod_{u=1}^n e_{u,j}}\prod_{h\neq l }C_{h,j}^{\prod_{u=1}^n e_{u,j}}
\end{equation}
Using the formula \ref{form2} for $F$ and $C$, 
\autoref{case_1_D_j_6} becomes
\begin{equation}\label{case_1_D_j_7}
D_j=g^{\left((2n-3)Y_{l,j}-\sum_{h \neq l} Y_{h,j}\right) \prod_{u=1}^n e_{u,j}^3} 
 \prod_{h\neq l }g^{\left(Y_{h,j}-\sum_{v \neq h} Y_{v,j}\right) \prod_{u=1}^n e_{u,j}^3}
\end{equation}
Simplifying \autoref{case_1_D_j_7}, we get
\begin{equation}\label{case_1_D_j_8}
D_j=g^{\left((n-2)Y_{l,j}-\sum_{h \neq l} (n-2)Y_{h,j}\right) \prod_{u=1}^n e_{u,j}^3}
\end{equation}
Finally, using the formula \ref{form2P} for $P$, we arrive at
\begin{equation}\label{case_1_D_j_9}
D_j=P_j^{n-2}. 
\end{equation}
\autoref{case_1_D_j_9} indicates that there is one or no bidder present for the $j^{th}$ digit. 
In such a scenario, the $j^{th}$ digit of the binary representation of the second price will be equal to `0'.

Case 2. If no bidder bids `1' in the current digit, then all bidders will use their check key to compute $D_j$.
The procedure is similar to Case 1, but all bidders use their check key. 
As a result, we have
\begin{equation}\label{D_j_proof}
    \begin{aligned}
        D_j &= (\prod_{i=1}^n C_{i,j})^{\prod_{u=1}^n e_{u,j}}\\ 
        &=\prod_{i=1 }^n g^{\left(Y_{i,j}-\sum_{v \neq i} Y_{v,j}\right) \prod_{u=1}^n e_{u,j}^3}\\ 
        &=g^{\left(-\sum_{i=1}^n (n-2)Y_{i,j}\right) \prod_{u=1}^n e_{u,j}^3}. 
    \end{aligned}
\end{equation}
In this case, we have $D_j=P_j^{n-2}$ only if 
\begin{equation}
    P_j=g^{( -\sum_{v=1}^n Y_{v,j}) \prod_{u=1}^n e_{u,j}^3} \,,   
\end{equation}
which implies that if no one bids `1' at this digit, 
then the second-highest bid is `0'. 

\end{proof}

\begin{proposition}\label{prop_2}
If $D_j \neq P_j^{n-2}$, 
then there are two or more bidders for the $j^\text{th}$ digit, 
and the $j^{th}$ digit of the binary representation of the final output price $p_{0,j}$ is `1'.
\end{proposition}

\begin{proof} 
In the case where all bidders use their check keys, 
we can use the result from the proof of $D_j$ in \autoref{D_j_proof}.
We know that $D_j=P_j^{n-2}$ only if 
\begin{equation}
    P_j=g^{( -\sum_{v=1}^n Y_{v,j}) \prod_{u=1}^n e_{u,j}^3},   
\end{equation}
Therefore, if $D_j \neq P_j^{n-2}$, 
it implies that there are two or more bidders for the $j^{th}$ digit. 
In this case, the $j^{th}$ digit of the binary representation of the second price will be equal to `1'.

\end{proof}

\subsubsection{Bid Adjustment}\label{sub_Bids_Adjustment}

This step involves adjusting bids in order to maintain the privacy of bidders. 
There are two different situations where these adjustments may be required.

The first scenario arises when the winning bidder wishes to conceal their bid and maintain their leading position.
This can be achieved by adjusting all digits $w > j$ after the $j^{th}$ digit of bids $b_{l,l,w}$ such that
\begin{equation}\label{Bids_Adjustment_1}
\sum_{i=1}^n b_{i,l,w}=Y_{l,w} \,.
\end{equation}
This adjustment is described in \autoref{bids_adjust_1} of \autoref{2nd_price_algo}. 
It ensures that the bidder's winning price remains private and their success in future rounds $w > j$.

The second scenario is when unsuccessful bidders whose bids are lower than the second-highest bid quit the competition by changing their bids to `0'. 
In such cases, the adjustment is made by altering all digits after the $j^{th}$ digit of bids $b_{l,l,w} (w > j)$ so that the following equation is satisfied
\begin{equation}\label{Bids_Adjustment_2}
\sum_{i=1}^n b_{i,l,w}=N_{l,w} \,.
\end{equation}
This adjustment is described in \autoref{bids_adjust_2} of \autoref{2nd_price_algo}. 
By bidding a minimum price instead of quitting the auction, 
the bidder's ideal price remains undisclosed. 
Additionally, bidding a fake price instead of quitting directly protects privacy as others can determine an upper bound of their ideal price, which is less than the second-highest bid if they exit the auction.

\bigskip
\begin{algorithm}[h]
\caption{Price Determination}
\label{2nd_price_algo}
\DontPrintSemicolon
\SetAlgoLined

\For{j =1 to k}{
$[B_j, P_j] \gets \RingTransferBP$ \; 

\vspace{0.1cm}
\For{l =1 to n}{

$d_{i,l,j} \gets \Random(\mathbb{F})$ for $i \in\{1,2,... ,n-1\}$\;    
\vspace{0.1cm}
\eIf{$K_{l,j}=B_j$}{
$d_{n,l,j} \gets F_{l,j} / \prod_{i=1}^{n-1} d_{i,l,j}$ for each $j$ \label{operation_1} \;
$b_{l,l,w} \gets Y_{l,w}-\sum_{i=1}^n b_{i,l,w}$ for $\forall w > j$ \label{bids_adjust_1}
}{
$d_{n,l,j} \gets C_{l,j} / \prod_{i=1}^{n-1} d_{i,l,j}$ for each $j$ \label{operation_2} \; 
}
 
Send $d_{i,l,j}$ to bidder $\alpha_i$ for each $i \neq l$ \;
}
$D_j \gets \RingTransferD$ \;
 
\vspace{0.1cm}
 \eIf{$D_j=P_j^{n-2}$}{
$p_{0,j} \gets 0$
            }{
$p_{0,j} \gets 1$ \;
 \vspace{0.1cm}
\For{l =1 to n}{ \If{ $\sum_{i=1}^n b_{i,l,j}=N_{l,j} $ }
{$b_{l,l,w} \gets N_{l,w}-\sum_{i=1}^n b_{i,l,w}$ for $\forall w > j$ \label{bids_adjust_2}} 
 } 
}
 }
 
\end{algorithm} 
\bigskip

\subsection{Final Price Verification}\label{sub_Final_Price_Verification}

Although the bids adjustment process is designed to protect bidders' privacy, 
it also creates the possibility for fraudulent behavior. 
Therefore, 
to maintain the integrity of the auction, 
it is crucial to verify any adjusted bids that could result in an output price higher than the second-highest bid, 
so that in reality we could incorporate a sub-protocol that detects and penalizes dishonest bidders through fines or exclusion from the auction or other means, thus incentivizing them to follow the main protocol.
This verification is done by comparing the output price calculated in \Cref{sub_2nd_Price} to the committed bids in \Cref{sub_Bid_Commitment}.

Let the output price calculated in \Cref{sub_2nd_Price} be denoted by $p_0$. 
To achieve this verification, 
each bidder $\alpha_l$ follows the steps outlined in Algorithm \ref{Final_Price_Verification_algo}. 
The algorithm involves the following steps.
Compute
\begin{equation}\label{price_check_1}
    { }_1 \mathcal{P}=g^{p_0 c_{l,1}},
\end{equation}
and send it to bidder $\psi(\alpha_l)$.
Upon receipt of ${ }_r \mathcal{P}$, compute
\begin{equation}\label{price_check_2}
    { }_{r+1} \mathcal{P}={ }_r \mathcal{P}^{c_{l,r+1}}.
\end{equation}
If $r=n-1$, publish 
\begin{equation}\label{price_check_3}
\mathcal{P}_{l}={ }_{n} \mathcal{P}
\end{equation}
to the public channel.
Otherwise, send ${ }_{r+1} \mathcal{P}$ to bidder $\psi(\alpha_l)$.

If there exists $i \in \{1,2,\dots,n\}$ such that $\mathbf{P}_i \neq \mathcal{P}_i$, then $p_0$ is accepted. 
Otherwise, $p_0$ is rejected.

\bigskip
\begin{algorithm}[h]
\caption{Final Price Verification}
\label{Final_Price_Verification_algo}

\DontPrintSemicolon
    \SetAlgoLined
\vspace{0.1cm}
     \Initialization{$r = 1$}{
     Send ${ }_1 \mathcal{P} \gets g^{p_0 c_{l,1}}$ to bidder $\psi(\alpha_l)$.} 
\vspace{0.1cm}
    \Upon{ ${ }_r \mathcal{P}$}{
\eIf{$r<n-1$}{
Send ${ }_{r+1} \mathcal{P} \gets { }_r \mathcal{P}^{c_{l,r+1}}$ to bidder $\psi(\alpha_l)$\;
            }{
Publish $\mathcal{P}_{l} \gets { }_{n} \mathcal{P}$ to the public channel. \;     
            }
    }

\vspace{0.1cm}
 \eIf{$\exists~i \in \{1,2,\dots,n\} $ s.t. $\mathbf{P}_i=\mathcal{P}_i$ }
 {$p_0$ is accepted}
 {$p_0$ is rejected.}
 
\end{algorithm} 
\bigskip

\begin{proposition}\label{output_price_is_second_bids}
The output price $p_0$ is the original committed second-highest bid if and only if 
\begin{equation}\label{final_test}
\exists ~i\in \{1,2,\dots,n\}  ~\text{such that}~ ~\mathbf{P}_i=\mathcal{P}_i
\end{equation}
\end{proposition}
\begin{proof}

We can express the committed bids in Algorithm \ref{bid_commitment_algo} using \begin{equation}\label{price_commitment_4}
    \mathbf{P}_{l}=g^{p_{l+1} \prod_{1\leq i\leq n, j=i-l} c_{i,j}} \,.
\end{equation}   
which is obtained by recursively substituting \autoref{price_commitment_2} into \autoref{price_commitment_3} and applying the initial condition of \autoref{price_commitment_1}.

To verify the legitimacy of the final output price, 
we substitute \autoref{price_check_2} into \autoref{price_check_3} iteratively and use the initial condition \autoref{price_check_1} to represent the published $\mathcal{P}_{l}$ as given by 
\begin{equation}\label{eqn_P_l}
    \mathcal{P}_{l}=g^{p_0 \prod_{1\leq i\leq n, j=i-l} c_{i,j}},
\end{equation}

Comparing \autoref{eqn_P_l} to \autoref{price_commitment_4}, 
we can observe that if there exists $i \in \{1,2,\dots,n\}$ such that $\mathbf{P}_i=\mathcal{P}_i$, then this condition implies that $p_0=p_{i+1}$. 
Hence, $p_0$ must be one of the original committed bids.

Furthermore, 
any false bid adjustments can only result in the output price $p_0$ being equal to or higher than the original committed second-highest bid. 
This is because the mechanism for calculating the second-highest price in \Cref{sub_Price_Determination} can only maintain or increase the output price.
If the output price is equal to the original committed second-highest bid, 
this case can be omitted since this false bid adjustment does not affect the output price $p_0$. 
However, if the output price is higher than the original committed second-highest bid,
then no one can claim to be the winner of the auction in this situation, 
and the system automatically rejects this price.

Therefore, the condition \autoref{final_test} is enough to guarantee that $p_0$ is indeed the original committed second price.

\end{proof}

\begin{remark}[A Tie-Breaking Mechanism]\label{remark_4}

If there are multiple indices $i$ that satisfy the condition in \autoref{final_test},  
then multiple bidders have bid the output price (i.e., the second-highest bid). 
The possibility of ties in the auction outcome thus requires a tie-breaking mechanism:

\begin{enumerate}
\item
If any bidder claims to be the winner, 
they will be awarded the item and pay the output price. 
\item
If no bidder claims to be the winner, 
it means that there are multiple bidders who bid the highest bid, 
and the output price is also the highest bid. 
In this situation, 
 the system can randomly select one of the bidders from the set $\{\alpha_{i+1}\mid \mathbf{P}_i=\mathcal{P}_i\}$ to be the ultimate winner.
This random selection does not affect the final payment price, 
which remains the same as the highest bid.
\end{enumerate}
It is worth noting that such a tie-breaking mechanism does not affect the privacy of bidders' information.
\end{remark}


\subsection{Winner Determination}\label{sub_winner_determination}

In the final step of the auction, 
the winner comes forward to the seller with a proof to show that the win was real. 
On this basis, everyone is able to check if there is only one winner or if there is a tie. 
The highest bidder is hence determined and the item is awarded accordingly.

Let $j'$ be the largest index such that $p_{0, j'} = 0$. 
 The winner $\alpha_i$ sends the seller $K_{u,i,j'}$ for $\forall u$ and can be verified by checking 
\begin{equation}\label{winner_determination_1}
 \prod_{u=1}^n K_{u,i,j'}=B_{j'}
 \end{equation}


By demonstrating \autoref{winner_determination_1}, 
the winner $\alpha_i$ can claim that they are the only bidder who placed a `1' at the digit ${j'}$, 
and no one else can dispute this claim.
Moreover, the first bids adjustment rule guarantees that the other digits in $\alpha_i$'s bid remain undisclosed to both other bidders and the seller.
Therefore, by submitting $K_{u,i,j'}$, the winner $\alpha_i$ is not revealing any information regarding the other parts of their bid.
Rather, they are simply presenting their willingness to pay a price that is at least one bid higher than the second-highest bid.


\section{Analysis of Protocol}\label{section_4}

In this section, we will discuss the security and privacy analysis of our proposed protocol and demonstrate how it is resilient to various types of attacks, including collusion.

\begin{theorem}[Correctness]
The proposed auction protocol guarantees that the second-highest bid is the output of the auction, 
and the winner can verify success to the seller. 
Additionally, the protocol is secure and any dishonest behavior that affects the output price will be rejected.
\end{theorem}

\begin{proof}

Proposition \autoref{output_price_is_second_bids} guarantees the detection and rejection of a false output price $p_0$ that is not the second-highest bid.
If the output price $p_0$ is the second-highest bid, 
then applying Propositions \autoref{prop_1} and \autoref{prop_2} to all digits results in the full binary representation of the second-highest bid, which is the output price $p_0$.
Since these values are published, 
the second-highest bid becomes accessible to all bidders in the auction as the output price.

If there is a unique winner, 
 the bids adjustment step guarantees that the winner is the only bidder who bid `1' at a specific digit $j$. 
Proposition \autoref{proposition_2} enables winners to determine if they won at a specific digit $j$ by checking their key values. 
Moreover, 
the seller can verify the winner's success since the hash value of all the $K_{u,l,j}$ is sent to the seller at the end of \Cref{sub_Keys_Generations}.

If there is a tie, 
where multiple bidders bid the highest price and the second-highest bid equals the highest bid, 
Remark \autoref{remark_4} shows that the winners can be determined. 
This does not affect the final payment price since it remains the same as the output price. 

By combining the results of the propositions and lemmas presented above, 
we have shown that the auction protocol is secure and therefore achieves the desired outcome of producing the second-highest bid as the output price while allowing the winner to verify their success to the seller. 

\end{proof}

 \begin{theorem}[Privacy-Preserving]
The proposed auction protocol guarantees full privacy, 
as it ensures that no coalition of up to $n-1$ bidders can obtain any information about the other bidders, except for their own bid and the output price (i.e., the second-highest bid).
\end{theorem}

\begin{proof}
The bids adjustment rule in \autoref{Bids_Adjustment_2} guarantees privacy for bidders who bid less than the second-highest price by allowing them to bid `0' instead of quitting the auction and ensuring their bids and ranks remain undisclosed. 
Our only concern is the privacy of the winning bidder, denoted by $\alpha_1$.
Without loss of generality, 
let us assume that a coalition of up to $n-1$ bidders, controlled by $\alpha_2$, 
may attempt to obtain information about the bids or ranks of bidder $\alpha_1$.

To prove full privacy, 
we examine all potential methods that coalition $\alpha_2$ might employ to obtain information about the bids or ranks of $\alpha_1$.
For each method, we show that it is impossible for coalition $\alpha_2$ to obtain any information beyond the output price and its own bid.
\begin{enumerate}

\item
In \Cref{sub_Bid_Commitment},
 we consider a scenario where $\alpha_2$ attempts to obtain information about bidder $\alpha_1$'s bid, denoted by $p_1$. 
 One approach for $\alpha_2$ is to directly decrypt bidder $\alpha_1$'s bid from ${ }_1 \mathbf{p}$ as presented in \autoref{price_commitment_1} in the ring transfer of the prices-commitment step. 
 However, this approach is impossible because $c_{1,1}$, which is necessary for decryption, is never shared and remains unknown to $\alpha_2$.
 
\item
In \Cref{sub_Bids_Sharing},  
$\alpha_2$ may attempt to obtain information about the $j^{th}$ digit of $\alpha_1$'s bid during the bid-sharing step. 
One possible approach for $\alpha_2$ is to use \autoref{Y_form} and \autoref{def_share_bids} to check if $b_{1,1,j}+b_{2,1,j}=a_{1,1,j}+a_{2,1,j}$. 
However, this method is impossible because $b_{1,1,j}$ and $a_{1,1,j}$ are unknown to $\alpha_2$ and cannot be extracted from ${}_1B_{1,j}$, ${}_1K_{1,i,j}$, ${}_1C_{1,i,j}$ and ${}_1F_{1,i,j}$ due to the encryption using $e_{1,j}$.

\item 
Since we assumed at the beginning of this proof that $\alpha_1$ is the winner, 
we have already eliminated the case where $\alpha_2$ bids `1' and $\alpha_1$ bids `0'. 
Thus, if $\alpha_2$ bids `0' in the $j^{th}$ digit, there are two possible bids for $\alpha_1$: `0' or `1'. 
However, if $\alpha_2$ bids `1' in the $j^{th}$ digit, it is certain that $\alpha_1$ bids `1'.
In \Cref{sub_Checking_for_a_Sole_Participant},
our goal is to demonstrate that $\alpha_2$ cannot determine whether $\alpha_1$ bid `1' or `0' in the $j^{th}$ digit when their bid is `0' in following cases 

\begin{enumerate}
\item 
The first possibility for $\alpha_2$ to obtain information about $\alpha_1$'s bid is through the common value $B_j$ and attempting to find out more about $\alpha_1$'s key values.
Proposition \autoref{proposition_2} provides a test for bidder $\alpha_1$ to determine if they are the only bidder who bids `1' in that round.
If $B_j=K_{1,j}$ after the round transfer of $B_{i,j}$ in the price-determination step,
this indicates that $\alpha_1$ is the only bidder who bid `1'.
However, $\alpha_2$ cannot check this condition
because the encryption of $K_{1,j}$ in \autoref{K_formula_proof_1} using $e_{1,j}$ and $K_{2,1,j}$ are unknown to $\alpha_2$.

\item 
The second possibility for $\alpha_2$ to obtain information about $\alpha_1$'s bid is through equations \autoref{eqn_B_j_7} and \autoref{eqn_B_j_8}.
In \autoref{eqn_B_j_7}, 
$\alpha_2$ can determine if all bidders bid `0' for digit $j$ by testing if $B_j = g^{-(Y_{1,j}+Y_{2,j})e_{1,j}e_{2,j}}$ after the round transfer of $B_{i,j}$ in the price-determination step.
However, 
since $Y_{1,j}$ and $e_{1,j}$ are unknown to $\alpha_2$, 
the only way for them to recognize $g^{-(Y_{1,j}+Y_{2,j})e_{1,j}e_{2,j}}$ is to use their key $K_{2,j} = g^{(-Y_{1,j}+Y_{2,j})e_{1,j}e_{2,j}}$ from \autoref{K_formula_proof_1}.
This is not possible because $g^{2Y_{2,j}e_{1,j}e_{2,j}}$ cannot be extracted from the values $\alpha_2$ has.
While $\alpha_2$ can obtain $g^{2Y_{2,j}e_{1,j}^2e_{2,j}^2}$ from $C_{2,j}$ and $F_{2,j}$ using \autoref{form2}, this information is not helpful due to the existence of $e_{1,j}$, which remains unknown to $\alpha_2$.

\item 
Another possibility for $\alpha_2$ to obtain information about $\alpha_1$'s bid is by analyzing the value $P_j$ using \autoref{form1P} and \autoref{form2P}.
In \autoref{form1P}, 
$\alpha_2$ can check if $P_j=g^{-( Y_{1,j}+Y_{2,j}) e_{1,j}^3 e_{2,j}^3}$ after the round transfer of $P_{i,j}$ in the price-determination step. 
However, 
this verification process is more complex than checking $B_j$, as $\alpha_2$ does not have any information about $g^{e_{1,j}^3}$. 
\end{enumerate}

\item 
In \Cref{sub_2nd_Price}, 
$\alpha_2$ may try to determine whether $\alpha_1$ uses check keys $C_{1,j}$ or fake keys $F_{1,j}$ in the price-determination step.
However, this is impossible since the values $C_{2,1,j}$ in \autoref{C_formula} or $F_{2,1,j}$ in \autoref{F_formula} are unknown to $\alpha_2$. 

\item 
In \Cref{sub_Bids_Adjustment},
$\alpha_2$ may try to determine whether bidder $\alpha_1$ changes their $b_{1,1,j}$ using 
\autoref{Bids_Adjustment_1} or \autoref{Bids_Adjustment_2} for each digit. 
However, this is impossible because $\alpha_1$ has never revealed or committed its $b_{1,1,j}$ to $\alpha_2$. 
In fact, $\alpha_1$ has never published its $b_{1,1,j}$ to $\alpha_2$, so $\alpha_2$ has no knowledge of $\alpha_1$ 's bids at any step of the auction protocol.
    
\end{enumerate}

Based on the analysis of all possible scenarios throughout every step of our auction protocol, 
we have shown that for each potential method that a coalition of up to $n-1$ bidders might employ, 
it is impossible to obtain any information beyond the output price and their own bid. 
Therefore, we have proven that our proposed auction protocol guarantees full privacy.

\end{proof}

\section{Numerical Test}\label{section_5}

In this section, 
we evaluate the effectiveness of the proposed auction protocol by examining its performance based on complexity analysis and numerical simulations.

We first analyzed each step in detail and assessed the computational complexity and communication overhead of the algorithms involved. The protocol, discussed in detail in \Cref{section_3},
comprises several steps with corresponding algorithms: 
key generation (Algorithm \ref{keys_generation_algo}), 
bid commitment (Algorithm \ref{bid_commitment_algo}), 
bid sharing (handled by Algorithm \ref{bids_sharing_algo}), 
ring transfer functions to generate $B_j$, $P_j$, and $D_j$ values (Algorithms \ref{ring_transfer_BP_algo} and \ref{ring_transfer_D_algo}), 
price determination (Algorithm \ref{2nd_price_algo}), 
and final price verification (Algorithm \ref{Final_Price_Verification_algo}).
We measured computational complexity in terms of the number of operations performed by a single bidder and communication overhead in terms of the data size transmitted, which is summarized for each step in Table \ref{table:complexity}.

The total computational complexity of our protocol is estimated to be $O(n^2 k)$, with the key generation step (Algorithm \ref{keys_generation_algo}) being the most computationally expensive operation. 
However, participants who enter the auction again in the future can safely reuse the previously generated keys, 
leading to a reduced computational complexity of $O(n k)$.
Regarding the communication overhead, 
each bidder only communicates with their immediate neighbors in the ring, 
resulting in $O(n k)$ total messages sent. 
Therefore, our protocol is communication-efficient as well.

Our proposed Vickrey auction protocol achieves a computational complexity of $O(n k)$ for the main auction, 
matching the best-known computational complexity record for sealed-bid first-price auction schemes, 
as demonstrated in the study by Bag et al. \cite{bag2019seal}. 
Furthermore, our protocol establishes a new standard for Vickrey auction schemes by achieving the same level of computational efficiency for the main auction while maintaining full privacy. 
A detailed comparison of computational complexity is illustrated in \autoref{table:comparison}.

\begin{table*}[ht]
\center
\resizebox* {!} {3.6cm}{
\centering
\begin{tabular}{|cll|c|c|c|}
\hline
\multicolumn{3}{|c|}{\textbf{Protocol Steps [with Associated Algorithms]}}              & \textbf{Computation} & \textbf{Communication} \\ \hline
\multicolumn{1}{|l|}{{Preparation}}                  & \multicolumn{2}{l|}{[\ref{keys_generation_algo}] Key Generation } & $O(n^2 k)$ & $O(n k)$ \\ \hline
\multicolumn{1}{|l|}{\multirow{6}{*}{{Main Auction}}} & \multicolumn{2}{l|}{[\ref{bid_commitment_algo}] Bid Commitment}                   & $O(n)$                                & $O(n)$                                  \\ \cline{2-5} 
\multicolumn{1}{|c|}{}                   & \multicolumn{1}{l|}{}                                & [\ref{bids_sharing_algo}] Bid Sharing                  & $O(nk)$                                & $O(nk)$                                \\ \cline{3-5} 
\multicolumn{1}{|c|}{}                   & \multicolumn{1}{l|}{[\ref{2nd_price_algo}] Price}    &  [\ref{ring_transfer_BP_algo}] Generating $B_j$ and $P_j$   & $O(nk)$                                & $O(n k)$                                \\ \cline{3-5} 
\multicolumn{1}{|c|}{}                   & \multicolumn{1}{l|}{Determination}                   & [\ref{ring_transfer_D_algo}] Generating $D_j$             & $O(nk)$                                & $O(n k)$                                \\ \cline{3-5} 
\multicolumn{1}{|c|}{}                   & \multicolumn{1}{l|}{}                                & Bid Adjustment                & $O(k)$                                & $0$                                \\ \cline{2-5} 
\multicolumn{1}{|c|}{}                   & \multicolumn{2}{l|}{[\ref{Final_Price_Verification_algo}] Final Price Verification} & $O(n)$                                & $O(n)$                                  \\ \hline
\multicolumn{3}{|l|}{{\bf Total Complexity of Main Auction (excluding Preparation)}}    & $O(n k)$ & $O(nk)$ \\ \hline
\end{tabular}}
\caption{Computational complexity and communication overhead of proposed auction protocol algorithms, measured by the number of operations per bidder and data size transmitted, respectively.}
\label{table:complexity}
\end{table*}

\begin{figure}[ht]
\centering
\subfloat[]
{\includegraphics[width=0.49\textwidth, trim=2cm 10cm 2cm 10cm,clip]{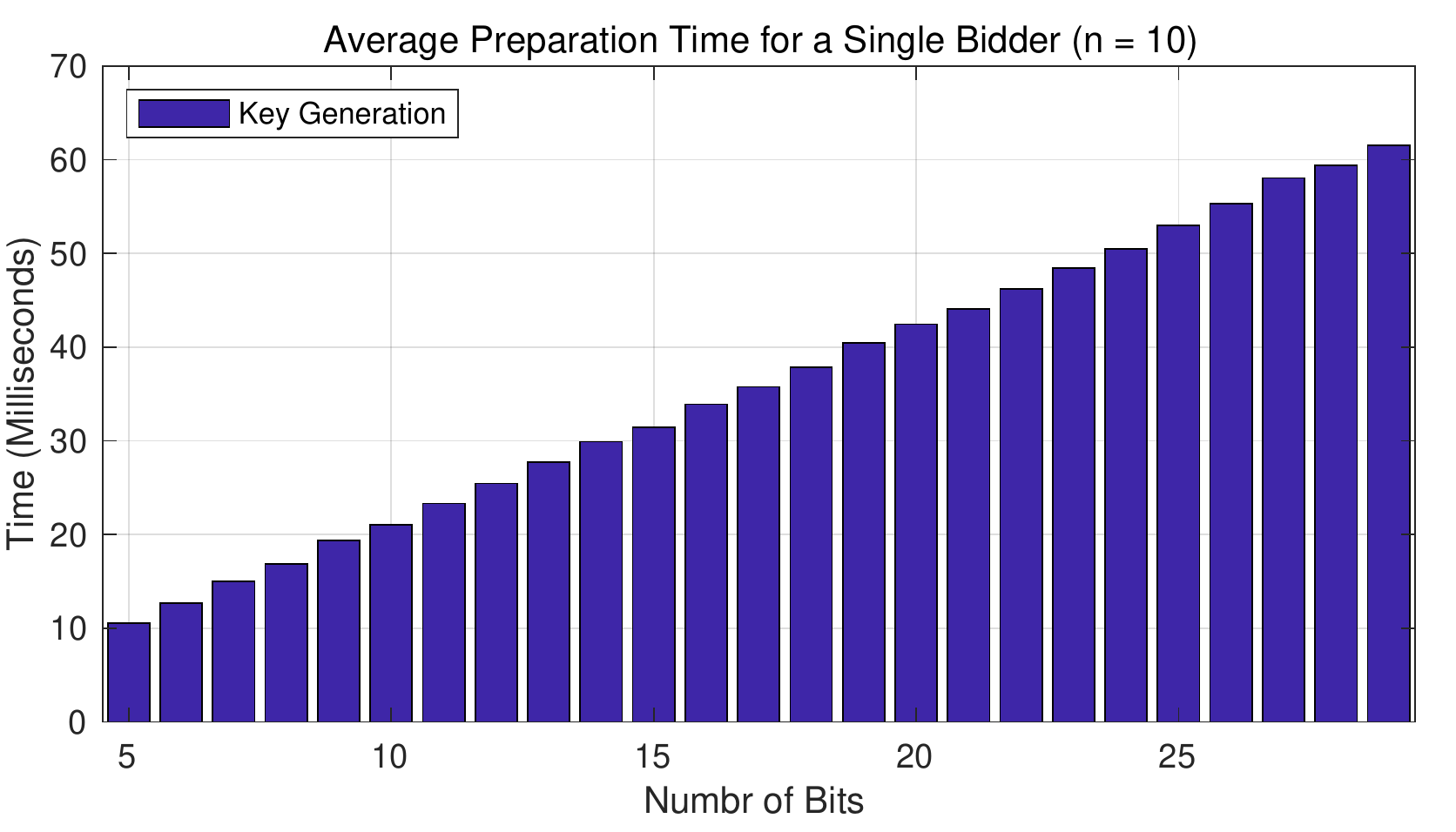}}
\label{figure_result_a}
\hfil
\subfloat[]
{\includegraphics[width=0.49\textwidth, trim=1.8cm 10cm 2cm 10cm,clip]{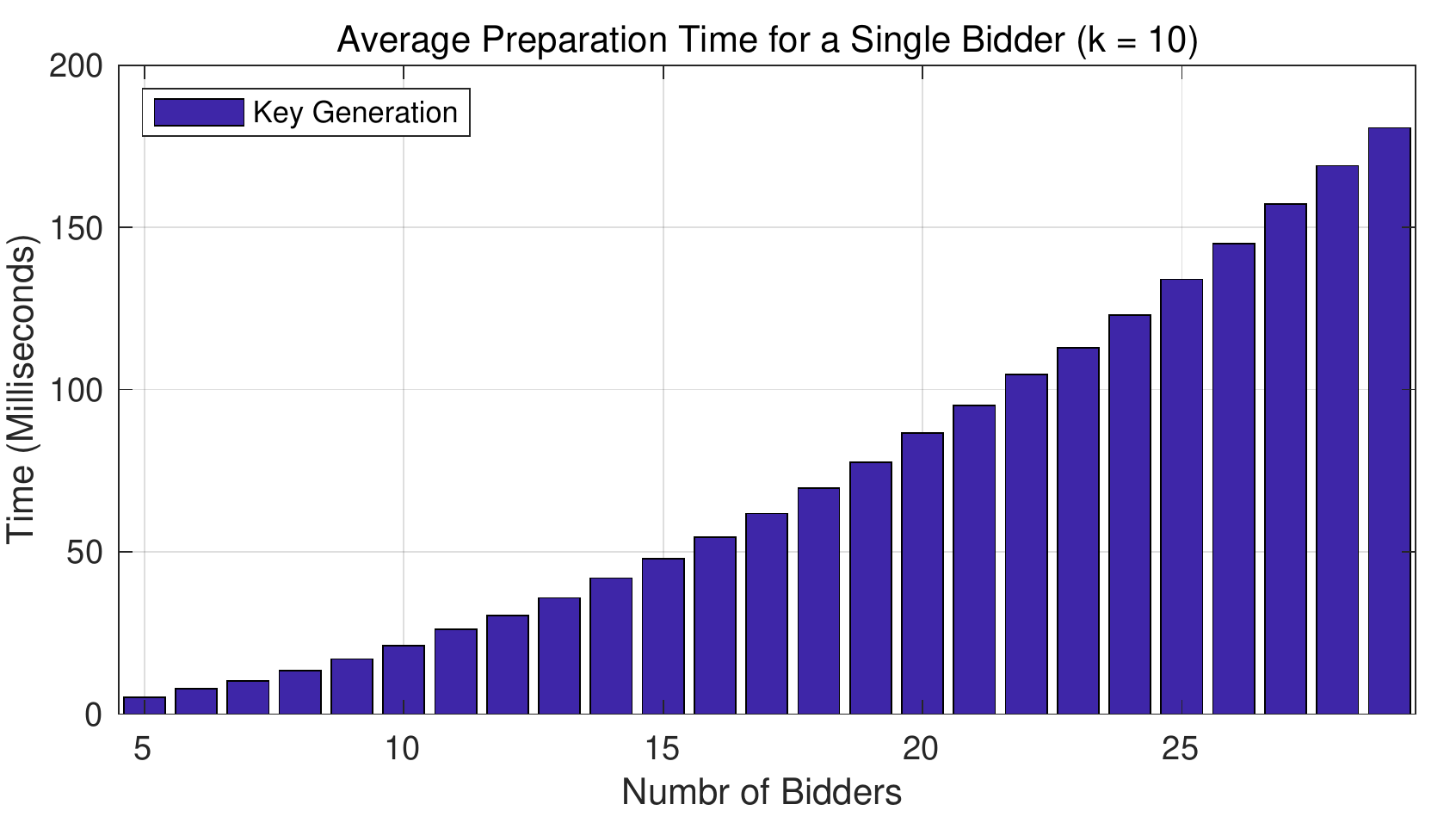}}
 \label{figure_result_b}
 \subfloat[]
{\includegraphics[width=0.49\textwidth, trim=1.8cm 10cm 2cm 10cm,clip]{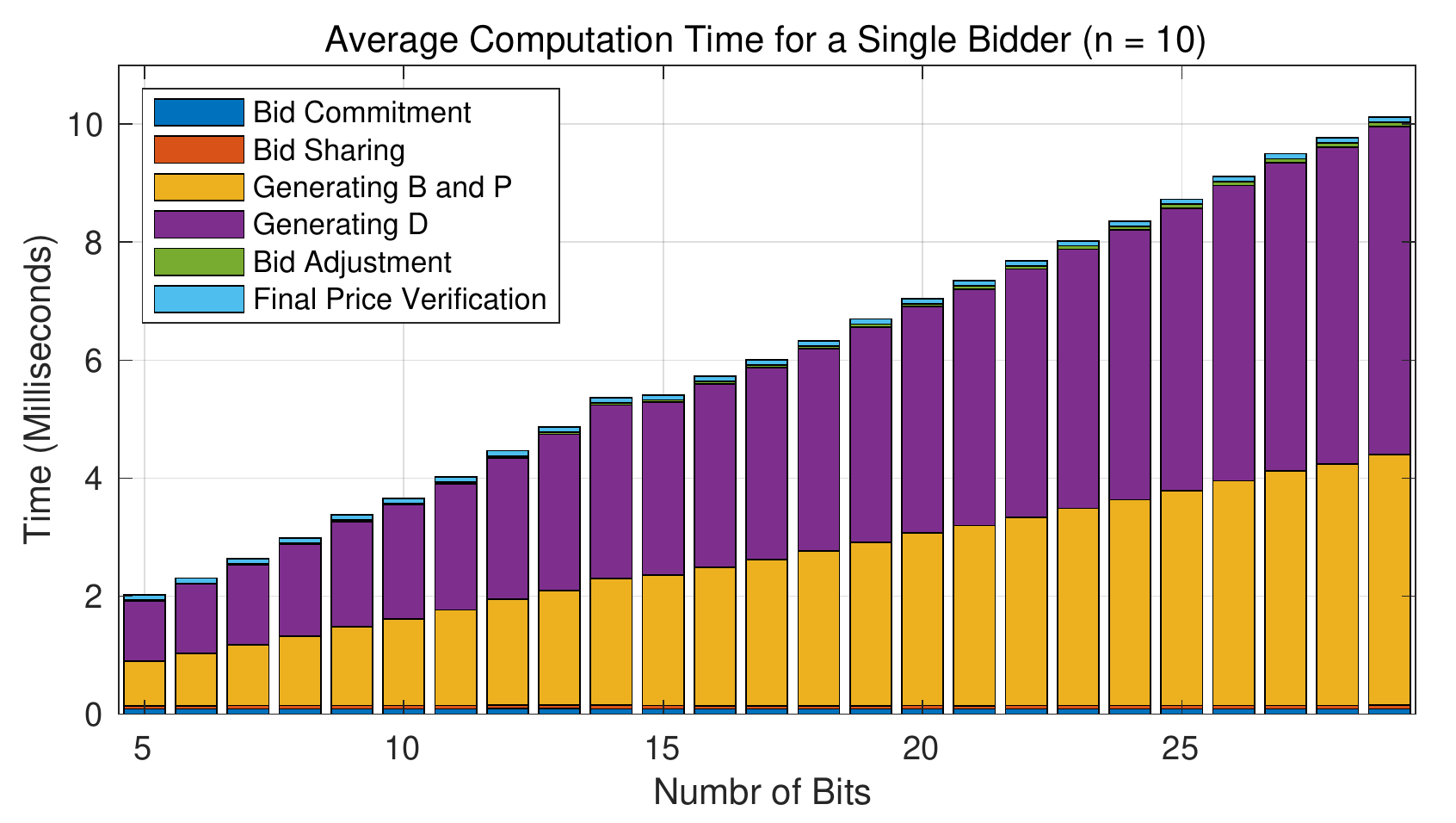}}
\label{figure_result_c}
\hfil
\subfloat[]
{\includegraphics[width=0.49\textwidth, trim=1.8cm 10cm 2cm 10cm,clip]{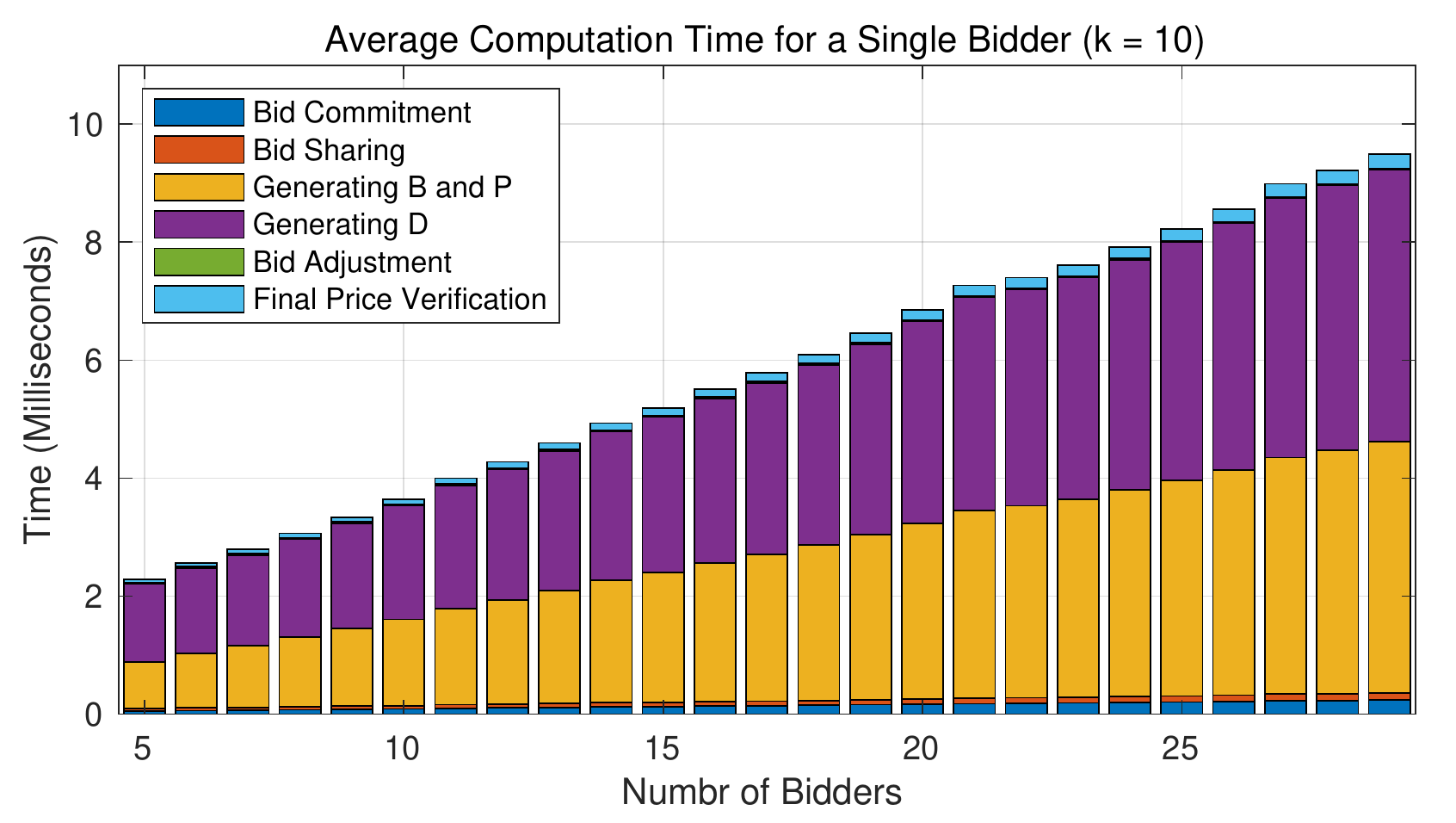}}
 \label{figure_result_d}
\caption{
This figure presents the average time taken by a single bidder to complete various stages of the auction protocol. 
Figure (a) illustrates the average time required for the preparation step as the number of bits $k$ varies in the case of 10 bidders,
while Figure (b) demonstrates the average time needed for the preparation step as the number of bidders $n$ changes, given 10 bits. 
Figures (c) and (d) display the average time taken by a single bidder to execute different auction steps when the number of bits $k$ or bidders $n$ varies, respectively. 
Each component of each auction step is represented by a different color.
}
\label{An_example_D_j_generation}
\end{figure}

To evaluate the performance of a proposed auction protocol, 
we conducted 1000 simulation trials using Matlab 2022b on a MacBook Pro equipped with an Apple M1 processor and 16 GB RAM. 
For each trial, 
we randomly selected a finite field $\mathbb{F}$ with an order that is a safe prime \cite{weisstein2002sophie} within the range of $2^{31}$ to $2^{32}-1$.
We then measured the average time required for the bidder to execute the various auction steps across all 1000 simulations. 
Firstly, we examined the average time needed for the preparation step, 
which is the key generation provided in Algorithm \ref{keys_generation_algo}. 
We plotted the average time against varying numbers of bits $k$ for $n=10$ bidders in Figure \ref{An_example_D_j_generation}(a), 
demonstrating a linear relationship with the value of $k$ when the number of bidders remains constant. 
Next, we plotted the average time against the varying numbers of bidders $n$ for a constant bit size $k=10$ in Figure \ref{An_example_D_j_generation}(b), 
illustrating a quadratic relationship with the value of $n$ when the number of bits is constant.
Furthermore, we investigated the average time necessary to complete all auction steps, 
including bid commitment, bid sharing, ring transfer functions to generate $B_j$, $P_j$, and $D_j$ values, price determination, and final price verification. 
We plotted the average time against the varying numbers of bits $k$ for $n=10$ bidders in Figure \ref{An_example_D_j_generation}(c), 
and against varying numbers of bidders $n$ for a constant bit size $k=10$ in Figure \ref{An_example_D_j_generation}(d). 
The bar plots provide insights into the time required for executing different steps of the protocol.
In particular, 
the expressions involving division operations in \autoref{operation_1} and \autoref{operation_2} of \autoref{2nd_price_algo} are the most expensive steps in finite field computations. 
One way to optimize these expressions is to eliminate the need for division and replace it with a multiplicative inverse with respect to the field modulus using an extended Euclidean algorithm.
We conclude that the overall time required to complete the auction protocol per bidder is linear with respect to the bit size $k$ and the number of bidders $n$. 
This corroborates the theoretical analysis presented in Table \ref{table:complexity}.

\begin{table}[h]
\center
\resizebox* {!} {1.8cm}{
\centering
\begin{tabular}{|c|c|}
\hline
 \textbf{Recorded Best Existing Protocols} & \textbf{Complexity in Main Auction} \\ \hline
Our Proposed Protocol & $O(n k)$ \\ \hline
Vickrey Auction \cite{brandt2002secure} & $O(n 2^k)$ \\ \hline
$1^\text{st}$-Price Sealed-bid Auction \cite{bag2019seal} & $O(n k)$ \\ \hline
\end{tabular}}
\caption{Comparison of Computational Complexity}
\label{table:comparison}
\end{table}

Overall, the proposed auction protocol demonstrates reasonable computational and communication efficiency, 
making it a feasible option for real-world applications, 
particularly when high privacy is essential or when dealing with high-value bids. 
However, in situations involving a large number of bidders, 
additional optimization may be required to alleviate communication burdens. 
Nevertheless, the decentralized nature of the protocol eliminates the need for an auctioneer, 
reducing potential single points of failure.

\section{Conclusion}\label{section_6}

In conclusion, the proposed Vickrey auction protocol is a promising solution for addressing security and privacy weaknesses in traditional Vickrey auctions. 
By using SMPC techniques, our protocol allows bidders to compute the second-highest bid in a decentralized setting while ensuring full privacy for all participants. 
Our protocol also offers resistance to attacks such as collusion and manipulation, 
with the auction outcome being publicly verifiable to ensure transparency and prevent dishonest bidders from altering the outcome. 
Additionally, a sub-protocol may be included to detect and penalize dishonest bidders, incentivizing all participants to follow the protocol. 
While our proposed protocol has reduced computational complexity, it may come with the limitation of requiring more communication among bidders, and it may work better with more expensive items. 
The choice of approach depends on the specific requirements and constraints of the application, such as the privacy requirement and the available computational resources. The proposed protocol has potential applications in a wide range of settings, including online auctions and government procurement processes. Future research should focus on developing more efficient and practical protocols that can scale to larger auctions with multiple bidders and items.

\section*{Acknowledgments}
This work is supported by the National Natural Science Foundation of China (No. 12225102 and 12050002).

\newpage

\bibliographystyle{plain}

\newpage

\begin{appendix}

\section{Visual Illustration of the Protocol Execution}\label{appendix_1}

This appendix includes a visual illustration that demonstrates the execution of the proposed protocol using a simple example.

We consider a set of bidders denoted by $$I = \{{\color{node_1} \alpha_{1}, \color{node_2} \alpha_{2}, \color{node_3} \alpha_{3}, \color{node_4} \alpha_{4}, \color{node_5} \alpha_{5}}\} \,.$$
Suppose that each bidder has a private valuation of 
$$ \begin{array}{c}
{\color{node_1} p_1 = 143 = (10001111)_{\,\text{base 2}}}\\ 
 {\color{node_2} p_2 = 124 = (01111100)_{\,\text{base 2}}}\\ 
 {\color{node_3} p_3 = 217 = (11011001)_{\,\text{base 2}}}\\ 
 {\color{node_4} p_4 = 222 = (11011110)_{\,\text{base 2}}}\\ 
 {\color{node_5} p_5 = 86 = (01010110)_{\,\text{base 2}}}\\ 
 \end{array}$$  
and the best response is to bid this amount,
regardless of the behavior of other bidders. 
Each bid is represented in binary with $8$ bits.

We choose the multiplicative group of $\mathbb{F}_{2063} = \mathbb{Z}/2063\mathbb{Z}$ with $5$ as one of its generators.
The purpose of this appendix is to simplify the demonstration of the protocol's execution and aid the reader's understanding.
In the visual illustration provided, some negative numbers in $\mathbb{F}_{2063}$ have intentionally been left as negative values for the convenience of the reader. 
However, it is important to note that leaving these negative values unchanged does not affect the correctness of the protocol, and it allows for easier verification of its accuracy.

Here's how the protocol works in a decentralized setting:

\noindent $\color{node_1} \alpha_{1} $  generates codes 
$$ \footnotesize \begin{array}{l}
{\color{node_1} \alpha _{1ij} = \begin{bmatrix}
847 & 1555 & 1944 & 1527 & 273 & 1103 & 611 & 31 \\
21 & 1835 & 1604 & 499 & 208 & 1542 & 501 & 1692 \\
306 & 1621 & 667 & 1471 & 827 & 1554 & 1937 & 231 \\
1173 & 1405 & 207 & 699 & 156 & 112 & 1176 & 198 \\
889 & 1939 & 1058 & 1003 & 135 & 1432 & 1679 & 909 \\
\end{bmatrix} } \\\\ 
{\color{node_1} c _{1i} = \begin{bmatrix}
673 & 162 & 624 & 149 & 1349 
 \end{bmatrix}^\intercal } \\\\ 
{\color{node_1} e _{1j} = \begin{bmatrix}
1079 & 1216 & 1810 & 536 & 413 & 795 & 1869 & 981 
 \end{bmatrix} }; \\\\ 
\end{array}$$ 

\noindent $\color{node_2} \alpha_{2} $  generates codes 
$$ \footnotesize \begin{array}{l}
{\color{node_2} \alpha _{2ij} = \begin{bmatrix}
1449 & 168 & 1354 & 861 & 1923 & 1405 & 1731 & 1515 \\
114 & 1851 & 1987 & 1260 & 1485 & 1662 & 832 & 1871 \\
145 & 904 & 1162 & 1224 & 640 & 1491 & 1568 & 1012 \\
1528 & 888 & 1206 & 1740 & 1283 & 798 & 101 & 1126 \\
1007 & 551 & 1205 & 1868 & 1978 & 570 & 1893 & 943 \\
\end{bmatrix} } \\\\ 
{\color{node_2} c _{2i} = \begin{bmatrix}
182 & 1030 & 1717 & 612 & 671 
 \end{bmatrix}^\intercal } \\\\ 
{\color{node_2} e _{2j} = \begin{bmatrix}
1270 & 1384 & 1174 & 867 & 1078 & 833 & 1370 & 999 
 \end{bmatrix} }; \\\\ 
\end{array}$$ 

\noindent $\color{node_3} \alpha_{3} $  generates codes 
$$ \footnotesize \begin{array}{l}
{\color{node_3} \alpha _{3ij} = \begin{bmatrix}
30 & 735 & 1621 & 1696 & 1284 & 639 & 1103 & 1031 \\
1989 & 1100 & 104 & 1255 & 997 & 1609 & 1253 & 1643 \\
1972 & 911 & 414 & 1738 & 759 & 1264 & 1014 & 2019 \\
1720 & 1378 & 1117 & 1517 & 1339 & 369 & 304 & 1510 \\
1016 & 205 & 1536 & 674 & 1741 & 710 & 991 & 1405 \\
\end{bmatrix} } \\\\ 
{\color{node_3} c _{3i} = \begin{bmatrix}
876 & 1623 & 113 & 68 & 535 
 \end{bmatrix}^\intercal } \\\\ 
{\color{node_3} e _{3j} = \begin{bmatrix}
1062 & 1465 & 11 & 1444 & 1739 & 1024 & 862 & 1708 
 \end{bmatrix} }; \\\\ 
\end{array}$$ 

\noindent $\color{node_4} \alpha_{4} $  generates codes 
$$ \footnotesize \begin{array}{l}
{\color{node_4} \alpha _{4ij} = \begin{bmatrix}
1594 & 1891 & 992 & 2026 & 1781 & 1295 & 1008 & 1300 \\
1921 & 1407 & 895 & 583 & 1926 & 70 & 290 & 827 \\
631 & 956 & 881 & 594 & 98 & 69 & 1961 & 1836 \\
1116 & 182 & 1128 & 408 & 837 & 14 & 1380 & 757 \\
1607 & 1736 & 1342 & 1951 & 1077 & 1658 & 392 & 925 \\
\end{bmatrix} } \\\\ 
{\color{node_4} c _{4i} = \begin{bmatrix}
1098 & 1026 & 1853 & 923 & 133 
 \end{bmatrix}^\intercal } \\\\ 
{\color{node_4} e _{4j} = \begin{bmatrix}
815 & 1417 & 1446 & 1112 & 479 & 1202 & 95 & 187 
 \end{bmatrix} }; \\\\ 
\end{array}$$ 

\noindent $\color{node_5} \alpha_{5} $  generates codes 
$$ \footnotesize \begin{array}{l}
{\color{node_5} \alpha _{5ij} = \begin{bmatrix}
1950 & 1690 & 1325 & 1994 & 250 & 283 & 1689 & 918 \\
99 & 769 & 1206 & 332 & 489 & 1143 & 346 & 726 \\
951 & 657 & 641 & 2023 & 889 & 148 & 1628 & 1307 \\
1229 & 827 & 922 & 151 & 1573 & 998 & 1182 & 928 \\
1174 & 331 & 1306 & 347 & 1389 & 1210 & 819 & 1307 \\
\end{bmatrix} } \\\\ 
{\color{node_5} c _{5i} = \begin{bmatrix}
711 & 1752 & 1497 & 2043 & 140 
 \end{bmatrix}^\intercal } \\\\ 
{\color{node_5} e _{5j} = \begin{bmatrix}
1030 & 301 & 388 & 1776 & 781 & 12 & 181 & 1756 
 \end{bmatrix} }; \\\\ 
\end{array}$$ 


\newpage
Based on the algorithm for key generation described in Algorithm \ref{keys_generation_algo},
we can create keys, fake keys, and check keys as shown in \autoref{table:Keys} , \autoref{table:Fake Keys} and  \autoref{table:Check Keys}.
 \bigskip
\begin{table}[h] 
 \centering
 \resizebox* {!} {3cm}{ 
 \begin{tabular}{|c|c|c|c|c|c|c|c|c|} 
 \hline & $j=1$ & $j=2$ & $j=3$ & $j=4$ & $j=5$ & $j=6$ & $j=7$ & $j=8$ \\ 
 \hline {\color{node_1}  $K_{1,j}$} 
& {\color{node_1}  823} & {\color{node_1}  1882} & {\color{node_1}  801} & {\color{node_1}  1014} & {\color{node_1}  1204} & {\color{node_1}  295} & {\color{node_1}  2058} & {\color{node_1}  274} \\ 
 \hline {\color{node_2}  $K_{2,j}$} 
& {\color{node_2}  790} & {\color{node_2}  766} & {\color{node_2}  361} & {\color{node_2}  28} & {\color{node_2}  274} & {\color{node_2}  88} & {\color{node_2}  1642} & {\color{node_2}  230} \\ 
 \hline {\color{node_3}  $K_{3,j}$} 
& {\color{node_3}  757} & {\color{node_3}  624} & {\color{node_3}  1042} & {\color{node_3}  641} & {\color{node_3}  844} & {\color{node_3}  1993} & {\color{node_3}  820} & {\color{node_3}  478} \\ 
 \hline {\color{node_4}  $K_{4,j}$} 
& {\color{node_4}  1825} & {\color{node_4}  961} & {\color{node_4}  983} & {\color{node_4}  13} & {\color{node_4}  1490} & {\color{node_4}  346} & {\color{node_4}  915} & {\color{node_4}  1213} \\ 
 \hline {\color{node_5}  $K_{5,j}$} 
& {\color{node_5}  821} & {\color{node_5}  170} & {\color{node_5}  1185} & {\color{node_5}  959} & {\color{node_5}  484} & {\color{node_5}  1604} & {\color{node_5}  1330} & {\color{node_5}  67} \\ 
 \hline \end{tabular}} 
\caption{Generating keys using Algorithm \ref{keys_generation_algo}.}
   \label{table:Keys} 
 \end{table} 
\bigskip
\begin{table}[h] 
 \centering
 \resizebox* {!} {3cm}{ 
 \begin{tabular}{|c|c|c|c|c|c|c|c|c|} 
 \hline & $j=1$ & $j=2$ & $j=3$ & $j=4$ & $j=5$ & $j=6$ & $j=7$ & $j=8$ \\ 
 \hline {\color{node_1}  $F_{1,j}$} 
& {\color{node_1}  641} & {\color{node_1}  1994} & {\color{node_1}  75} & {\color{node_1}  1519} & {\color{node_1}  33} & {\color{node_1}  1436} & {\color{node_1}  443} & {\color{node_1}  1735} \\ 
 \hline {\color{node_2}  $F_{2,j}$} 
& {\color{node_2}  1174} & {\color{node_2}  1285} & {\color{node_2}  1501} & {\color{node_2}  1927} & {\color{node_2}  1875} & {\color{node_2}  1612} & {\color{node_2}  862} & {\color{node_2}  774} \\ 
 \hline {\color{node_3}  $F_{3,j}$} 
& {\color{node_3}  1800} & {\color{node_3}  823} & {\color{node_3}  1786} & {\color{node_3}  1757} & {\color{node_3}  336} & {\color{node_3}  717} & {\color{node_3}  502} & {\color{node_3}  1645} \\ 
 \hline {\color{node_4}  $F_{4,j}$} 
& {\color{node_4}  1418} & {\color{node_4}  1194} & {\color{node_4}  323} & {\color{node_4}  13} & {\color{node_4}  24} & {\color{node_4}  946} & {\color{node_4}  275} & {\color{node_4}  117} \\ 
 \hline {\color{node_5}  $F_{5,j}$} 
& {\color{node_5}  1063} & {\color{node_5}  1689} & {\color{node_5}  1549} & {\color{node_5}  1605} & {\color{node_5}  1543} & {\color{node_5}  268} & {\color{node_5}  1646} & {\color{node_5}  1097} \\ 
 \hline \end{tabular}} 
\caption{Generating fake keys using Algorithm \ref{keys_generation_algo}.}
 \label{table:Fake Keys} 
 \end{table} 
 \bigskip
\begin{table}[h] 
 \centering
 \resizebox* {!} {3cm}{ 
 \begin{tabular}{|c|c|c|c|c|c|c|c|c|} 
 \hline & $j=1$ & $j=2$ & $j=3$ & $j=4$ & $j=5$ & $j=6$ & $j=7$ & $j=8$ \\ 
 \hline {\color{node_1}  $C_{1,j}$} 
& {\color{node_1}  849} & {\color{node_1}  704} & {\color{node_1}  251} & {\color{node_1}  1171} & {\color{node_1}  1914} & {\color{node_1}  1206} & {\color{node_1}  1741} & {\color{node_1}  412} \\ 
 \hline {\color{node_2}  $C_{2,j}$} 
& {\color{node_2}  1563} & {\color{node_2}  946} & {\color{node_2}  1343} & {\color{node_2}  858} & {\color{node_2}  1835} & {\color{node_2}  1987} & {\color{node_2}  1533} & {\color{node_2}  862} \\ 
 \hline {\color{node_3}  $C_{3,j}$} 
& {\color{node_3}  664} & {\color{node_3}  1835} & {\color{node_3}  601} & {\color{node_3}  1597} & {\color{node_3}  349} & {\color{node_3}  454} & {\color{node_3}  1421} & {\color{node_3}  832} \\ 
 \hline {\color{node_4}  $C_{4,j}$} 
& {\color{node_4}  1969} & {\color{node_4}  1029} & {\color{node_4}  1330} & {\color{node_4}  1580} & {\color{node_4}  900} & {\color{node_4}  1618} & {\color{node_4}  88} & {\color{node_4}  806} \\ 
 \hline {\color{node_5}  $C_{5,j}$} 
& {\color{node_5}  2058} & {\color{node_5}  104} & {\color{node_5}  99} & {\color{node_5}  1869} & {\color{node_5}  614} & {\color{node_5}  650} & {\color{node_5}  1825} & {\color{node_5}  1274} \\ 
 \hline \end{tabular}} 
\caption{Generating check keys using Algorithm \ref{keys_generation_algo}.}
  \label{table:Check Keys} 
 \end{table} 
\bigskip

Use the bid commitment algorithm outlined in Algorithm \ref{sub_Bid_Commitment} to make the following committed bids on the  public channel:
$$  \begin{array}{l}
{\mathbf{P}_l = \begin{bmatrix} 
{\color{node_1}  681} &  
 {\color{node_2}  528} &  
 {\color{node_3}  718} &  
 {\color{node_4}  32} &  
 {\color{node_5}  9 }  \\ 
 \end{bmatrix}^\intercal  }
\end{array}$$

%
%
%

\newpage
In the bid sharing step, as described in Algorithm \ref{bids_sharing_algo}, 
each bidder $\alpha_l$ sends their codes $a_{l,i,j}$ to bidder $\alpha_i$ for all $i \neq l$.

\noindent $\color{node_1} \alpha_{1} $  receives 
$$  \begin{array}{l}
{\color{node_1}  
 a_{i,1,j} = \begin{bmatrix} 
{\color{node_1} 847 } & {\color{node_1} 1555 } & {\color{node_1} 1944 } & {\color{node_1} 1527 } & {\color{node_1} 273 } & {\color{node_1} 1103 } & {\color{node_1} 611 } & {\color{node_1} 31 }  \\ 
 {\color{node_2} 1449 } & {\color{node_2} 168 } & {\color{node_2} 1354 } & {\color{node_2} 861 } & {\color{node_2} 1923 } & {\color{node_2} 1405 } & {\color{node_2} 1731 } & {\color{node_2} 1515 }  \\ 
 {\color{node_3} 30 } & {\color{node_3} 735 } & {\color{node_3} 1621 } & {\color{node_3} 1696 } & {\color{node_3} 1284 } & {\color{node_3} 639 } & {\color{node_3} 1103 } & {\color{node_3} 1031 }  \\ 
 {\color{node_4} 1594 } & {\color{node_4} 1891 } & {\color{node_4} 992 } & {\color{node_4} 2026 } & {\color{node_4} 1781 } & {\color{node_4} 1295 } & {\color{node_4} 1008 } & {\color{node_4} 1300 }  \\ 
 {\color{node_5} 1950 } & {\color{node_5} 1690 } & {\color{node_5} 1325 } & {\color{node_5} 1994 } & {\color{node_5} 250 } & {\color{node_5} 283 } & {\color{node_5} 1689 } & {\color{node_5} 918 }  \\ 
 \end{bmatrix} } 
\end{array}$$
\noindent $\color{node_2} \alpha_{2} $  receives 
$$  \begin{array}{l}
{\color{node_2}  
 a_{i,2,j} = \begin{bmatrix} 
{\color{node_1} 21 } & {\color{node_1} 1835 } & {\color{node_1} 1604 } & {\color{node_1} 499 } & {\color{node_1} 208 } & {\color{node_1} 1542 } & {\color{node_1} 501 } & {\color{node_1} 1692 }  \\ 
 {\color{node_2} 114 } & {\color{node_2} 1851 } & {\color{node_2} 1987 } & {\color{node_2} 1260 } & {\color{node_2} 1485 } & {\color{node_2} 1662 } & {\color{node_2} 832 } & {\color{node_2} 1871 }  \\ 
 {\color{node_3} 1989 } & {\color{node_3} 1100 } & {\color{node_3} 104 } & {\color{node_3} 1255 } & {\color{node_3} 997 } & {\color{node_3} 1609 } & {\color{node_3} 1253 } & {\color{node_3} 1643 }  \\ 
 {\color{node_4} 1921 } & {\color{node_4} 1407 } & {\color{node_4} 895 } & {\color{node_4} 583 } & {\color{node_4} 1926 } & {\color{node_4} 70 } & {\color{node_4} 290 } & {\color{node_4} 827 }  \\ 
 {\color{node_5} 99 } & {\color{node_5} 769 } & {\color{node_5} 1206 } & {\color{node_5} 332 } & {\color{node_5} 489 } & {\color{node_5} 1143 } & {\color{node_5} 346 } & {\color{node_5} 726 }  \\ 
 \end{bmatrix} } 
\end{array}$$
\noindent $\color{node_3} \alpha_{3} $  receives 
$$  \begin{array}{l}
{\color{node_3}  
 a_{i,3,j} = \begin{bmatrix} 
{\color{node_1} 306 } & {\color{node_1} 1621 } & {\color{node_1} 667 } & {\color{node_1} 1471 } & {\color{node_1} 827 } & {\color{node_1} 1554 } & {\color{node_1} 1937 } & {\color{node_1} 231 }  \\ 
 {\color{node_2} 145 } & {\color{node_2} 904 } & {\color{node_2} 1162 } & {\color{node_2} 1224 } & {\color{node_2} 640 } & {\color{node_2} 1491 } & {\color{node_2} 1568 } & {\color{node_2} 1012 }  \\ 
 {\color{node_3} 1972 } & {\color{node_3} 911 } & {\color{node_3} 414 } & {\color{node_3} 1738 } & {\color{node_3} 759 } & {\color{node_3} 1264 } & {\color{node_3} 1014 } & {\color{node_3} 2019 }  \\ 
 {\color{node_4} 631 } & {\color{node_4} 956 } & {\color{node_4} 881 } & {\color{node_4} 594 } & {\color{node_4} 98 } & {\color{node_4} 69 } & {\color{node_4} 1961 } & {\color{node_4} 1836 }  \\ 
 {\color{node_5} 951 } & {\color{node_5} 657 } & {\color{node_5} 641 } & {\color{node_5} 2023 } & {\color{node_5} 889 } & {\color{node_5} 148 } & {\color{node_5} 1628 } & {\color{node_5} 1307 }  \\ 
 \end{bmatrix} } 
\end{array}$$
\noindent $\color{node_4} \alpha_{4} $  receives 
$$  \begin{array}{l}
{\color{node_4}  
 a_{i,4,j} = \begin{bmatrix} 
{\color{node_1} 1173 } & {\color{node_1} 1405 } & {\color{node_1} 207 } & {\color{node_1} 699 } & {\color{node_1} 156 } & {\color{node_1} 112 } & {\color{node_1} 1176 } & {\color{node_1} 198 }  \\ 
 {\color{node_2} 1528 } & {\color{node_2} 888 } & {\color{node_2} 1206 } & {\color{node_2} 1740 } & {\color{node_2} 1283 } & {\color{node_2} 798 } & {\color{node_2} 101 } & {\color{node_2} 1126 }  \\ 
 {\color{node_3} 1720 } & {\color{node_3} 1378 } & {\color{node_3} 1117 } & {\color{node_3} 1517 } & {\color{node_3} 1339 } & {\color{node_3} 369 } & {\color{node_3} 304 } & {\color{node_3} 1510 }  \\ 
 {\color{node_4} 1116 } & {\color{node_4} 182 } & {\color{node_4} 1128 } & {\color{node_4} 408 } & {\color{node_4} 837 } & {\color{node_4} 14 } & {\color{node_4} 1380 } & {\color{node_4} 757 }  \\ 
 {\color{node_5} 1229 } & {\color{node_5} 827 } & {\color{node_5} 922 } & {\color{node_5} 151 } & {\color{node_5} 1573 } & {\color{node_5} 998 } & {\color{node_5} 1182 } & {\color{node_5} 928 }  \\ 
 \end{bmatrix} } 
\end{array}$$
\noindent $\color{node_5} \alpha_{5} $  receives 
$$  \begin{array}{l}
{\color{node_5}  
 a_{i,5,j} = \begin{bmatrix} 
{\color{node_1} 889 } & {\color{node_1} 1939 } & {\color{node_1} 1058 } & {\color{node_1} 1003 } & {\color{node_1} 135 } & {\color{node_1} 1432 } & {\color{node_1} 1679 } & {\color{node_1} 909 }  \\ 
 {\color{node_2} 1007 } & {\color{node_2} 551 } & {\color{node_2} 1205 } & {\color{node_2} 1868 } & {\color{node_2} 1978 } & {\color{node_2} 570 } & {\color{node_2} 1893 } & {\color{node_2} 943 }  \\ 
 {\color{node_3} 1016 } & {\color{node_3} 205 } & {\color{node_3} 1536 } & {\color{node_3} 674 } & {\color{node_3} 1741 } & {\color{node_3} 710 } & {\color{node_3} 991 } & {\color{node_3} 1405 }  \\ 
 {\color{node_4} 1607 } & {\color{node_4} 1736 } & {\color{node_4} 1342 } & {\color{node_4} 1951 } & {\color{node_4} 1077 } & {\color{node_4} 1658 } & {\color{node_4} 392 } & {\color{node_4} 925 }  \\ 
 {\color{node_5} 1174 } & {\color{node_5} 331 } & {\color{node_5} 1306 } & {\color{node_5} 347 } & {\color{node_5} 1389 } & {\color{node_5} 1210 } & {\color{node_5} 819 } & {\color{node_5} 1307 }  \\ 
 \end{bmatrix} } 
\end{array}$$

\newpage
Using these codes,  bidder $\alpha_l$ can use them to compute the indicators for each digit $j$ of their binary representation bid $p_{l,j}$, 
which they committed to in the bid commitment step outlined in Algorithm \ref{sub_Bid_Commitment}.

\begin{table}[h] 
 \centering
 \resizebox* {!} {3cm}{ 
 \begin{tabular}{|c|c|c|c|c|c|c|c|c|} 
 \hline $(Y | N)_{l,j}$ 
& $j=1$ & $j=2$ & $j=3$ & $j=4$ & $j=5$ & $j=6$ & $j=7$ & $j=8$ \\ 
 \hline {\color{node_1}  $l = 1$} 
& {\color{node_1}  5870} & {\color{node_1}  -6039} & {\color{node_1}  -7236} & {\color{node_1}  -8104} & {\color{node_1}  5511} & {\color{node_1}  4725} & {\color{node_1}  6142} & {\color{node_1}  4795} \\ 
 \hline {\color{node_2}  $l = 2$} 
& {\color{node_2}  -4144} & {\color{node_2}  6962} & {\color{node_2}  5796} & {\color{node_2}  3929} & {\color{node_2}  5105} & {\color{node_2}  6026} & {\color{node_2}  -3222} & {\color{node_2}  -6759} \\ 
 \hline {\color{node_3}  $l = 3$} 
& {\color{node_3}  4005} & {\color{node_3}  5049} & {\color{node_3}  -3765} & {\color{node_3}  7050} & {\color{node_3}  3213} & {\color{node_3}  -4526} & {\color{node_3}  -8108} & {\color{node_3}  6405} \\ 
 \hline {\color{node_4}  $l = 4$} 
& {\color{node_4}  6766} & {\color{node_4}  4680} & {\color{node_4}  -4580} & {\color{node_4}  4515} & {\color{node_4}  5188} & {\color{node_4}  2291} & {\color{node_4}  4143} & {\color{node_4}  -4519} \\ 
 \hline {\color{node_5}  $l = 5$} 
& {\color{node_5}  -5693} & {\color{node_5}  4762} & {\color{node_5}  -6447} & {\color{node_5}  5843} & {\color{node_5}  -6320} & {\color{node_5}  5580} & {\color{node_5}  5774} & {\color{node_5}  -5489} \\ 
 \hline \end{tabular}} 
 \caption{The indicators $Y_{l,j}$ and $N_{l,j}$ represent a bid of $p_{l,j}=1$ and $p_{l,j}=0$, respectively, for bidder $\alpha_l$.}
 \end{table} 

In the next step, 
each bidder $\alpha_l$ randomly chooses $b_{i,l,j}$, 
such that 
\begin{equation}
    \sum_{i=1}^n b_{i,l,j}=\begin{cases}
    {Y_{l,j},} & {\text{if}}\ p_{l,j}=1 \\ 
    {N_{l,j},} & {\text{otherwise.}}
    \end{cases}
\end{equation}

\noindent $\color{node_1} \alpha_{1} $  generates 
$$  \begin{array}{l}
{\color{node_1} b_{i,1,j} = \begin{bmatrix}
1050 & 1613 & 457 & 1472 & 219 & 442 & 839 & 891 \\
1779 & 1860 & 758 & 1694 & 535 & 779 & 1668 & 790 \\
1431 & 1776 & 1919 & 555 & 1169 & 930 & 1642 & 1407 \\
258 & 512 & 904 & 1466 & 246 & 1871 & 141 & 1210 \\
1352 & -11800 & -11274 & -13291 & 3342 & 703 & 1852 & 497 \\
\end{bmatrix} } \\\\ 
\end{array}$$
\noindent $\color{node_2} \alpha_{2} $  generates 
$$  \begin{array}{l}
{\color{node_2} b_{i,2,j} = \begin{bmatrix}
551 & 231 & 427 & 598 & 637 & 1879 & 1570 & 1544 \\
221 & 1905 & 478 & 1212 & 508 & 364 & 866 & 1801 \\
369 & 2 & 660 & 4 & 1607 & 790 & 920 & 1189 \\
320 & 1574 & 1962 & 2 & 836 & 428 & 412 & 166 \\
-5605 & 3250 & 2269 & 2113 & 1517 & 2565 & -6990 & -11459 \\
\end{bmatrix} } \\\\ 
\end{array}$$
\noindent $\color{node_3} \alpha_{3} $  generates 
$$  \begin{array}{l}
{\color{node_3} b_{i,3,j} = \begin{bmatrix}
1612 & 149 & 1327 & 498 & 437 & 571 & 1797 & 1317 \\
1643 & 1463 & 48 & 1352 & 1981 & 17 & 3 & 518 \\
375 & 607 & 1456 & 22 & 1179 & 967 & 98 & 98 \\
291 & 1471 & 1862 & 125 & 1396 & 175 & 1981 & 124 \\
84 & 1359 & -8458 & 5053 & -1780 & -6256 & -11987 & 4348 \\
\end{bmatrix} } \\\\ 
\end{array}$$
\noindent $\color{node_4} \alpha_{4} $  generates 
$$  \begin{array}{l}
{\color{node_4} b_{i,4,j} = \begin{bmatrix}
1064 & 1875 & 698 & 48 & 59 & 188 & 696 & 1334 \\
1349 & 135 & 1660 & 404 & 1019 & 1600 & 1904 & 1514 \\
1778 & 1909 & 1116 & 1411 & 1286 & 64 & 1587 & 1313 \\
977 & 717 & 1366 & 1641 & 487 & 1962 & 1759 & 1977 \\
1598 & 44 & -9420 & 1011 & 2337 & -1523 & -1803 & -10657 \\
\end{bmatrix} } \\\\ 
\end{array}$$
\noindent $\color{node_5} \alpha_{5} $  generates 
$$  \begin{array}{l}
{\color{node_5} b_{i,5,j} = \begin{bmatrix}
416 & 439 & 123 & 1009 & 329 & 347 & 916 & 1558 \\
1282 & 355 & 585 & 1186 & 748 & 1677 & 838 & 1966 \\
309 & 903 & 1112 & 1173 & 1715 & 1505 & 286 & 631 \\
1976 & 1517 & 825 & 539 & 1793 & 234 & 647 & 1927 \\
-9676 & 1548 & -9092 & 1936 & -10905 & 1817 & 3087 & -11571 \\
\end{bmatrix} } \\\\ 
\end{array}$$
Once bidders generate their random variables, 
$\alpha_l$ sends corresponding $b_{i,l,j}$ values to each bidder $\alpha_i$ for all $i\neq l$.

\noindent $\color{node_1} \alpha_{1} $  receives 
$$ \begin{array}{l}
{\color{node_1} b_{1,i,j} = \begin{bmatrix}{\color{node_1} 1050 } & {\color{node_1} 1613 } & {\color{node_1} 457 } & {\color{node_1} 1472 } & {\color{node_1} -10803 } & {\color{node_1} -9008 } & {\color{node_1} -11445 } & {\color{node_1} -8699 } \\ 
{\color{node_2} 551 } & {\color{node_2} 231 } & {\color{node_2} 427 } & {\color{node_2} 598 } & {\color{node_2} 637 } & {\color{node_2} 1879 } & {\color{node_2} 1570 } & {\color{node_2} 1544 } \\ 
{\color{node_3} 1612 } & {\color{node_3} 149 } & {\color{node_3} 1327 } & {\color{node_3} 498 } & {\color{node_3} 437 } & {\color{node_3} 571 } & {\color{node_3} 1797 } & {\color{node_3} 1317 } \\ 
{\color{node_4} 1064 } & {\color{node_4} 1875 } & {\color{node_4} 698 } & {\color{node_4} 48 } & {\color{node_4} 59 } & {\color{node_4} 188 } & {\color{node_4} 696 } & {\color{node_4} 1334 } \\ 
{\color{node_5} 416 } & {\color{node_5} 439 } & {\color{node_5} 123 } & {\color{node_5} 1009 } & {\color{node_5} 329 } & {\color{node_5} 347 } & {\color{node_5} 916 } & {\color{node_5} 1558 } \\ 
\end{bmatrix} } \\\\ 
\end{array}$$
\noindent $\color{node_2} \alpha_{2} $  receives 
$$ \begin{array}{l}
{\color{node_2} b_{2,i,j} = \begin{bmatrix}{\color{node_1} 1779 } & {\color{node_1} 1860 } & {\color{node_1} 758 } & {\color{node_1} 1694 } & {\color{node_1} 535 } & {\color{node_1} 779 } & {\color{node_1} 1668 } & {\color{node_1} 790 } \\ 
{\color{node_2} 221 } & {\color{node_2} -12019 } & {\color{node_2} -11114 } & {\color{node_2} -6646 } & {\color{node_2} -9702 } & {\color{node_2} -11688 } & {\color{node_2} 866 } & {\color{node_2} 1801 } \\ 
{\color{node_3} 1643 } & {\color{node_3} 1463 } & {\color{node_3} 48 } & {\color{node_3} 1352 } & {\color{node_3} 1981 } & {\color{node_3} 17 } & {\color{node_3} 3 } & {\color{node_3} 518 } \\ 
{\color{node_4} 1349 } & {\color{node_4} 135 } & {\color{node_4} 1660 } & {\color{node_4} 404 } & {\color{node_4} 1019 } & {\color{node_4} 1600 } & {\color{node_4} 1904 } & {\color{node_4} 1514 } \\ 
{\color{node_5} 1282 } & {\color{node_5} 355 } & {\color{node_5} 585 } & {\color{node_5} 1186 } & {\color{node_5} 748 } & {\color{node_5} 1677 } & {\color{node_5} 838 } & {\color{node_5} 1966 } \\ 
\end{bmatrix} } \\\\ 
\end{array}$$
\noindent $\color{node_3} \alpha_{3} $  receives 
$$ \begin{array}{l}
{\color{node_3} b_{3,i,j} = \begin{bmatrix}{\color{node_1} 1431 } & {\color{node_1} 1776 } & {\color{node_1} 1919 } & {\color{node_1} 555 } & {\color{node_1} 1169 } & {\color{node_1} 930 } & {\color{node_1} 1642 } & {\color{node_1} 1407 } \\ 
{\color{node_2} 369 } & {\color{node_2} 2 } & {\color{node_2} 660 } & {\color{node_2} 4 } & {\color{node_2} 1607 } & {\color{node_2} 790 } & {\color{node_2} 920 } & {\color{node_2} 1189 } \\ 
{\color{node_3} 375 } & {\color{node_3} 607 } & {\color{node_3} 1456 } & {\color{node_3} 22 } & {\color{node_3} 1179 } & {\color{node_3} 967 } & {\color{node_3} 98 } & {\color{node_3} 98 } \\ 
{\color{node_4} 1778 } & {\color{node_4} 1909 } & {\color{node_4} 1116 } & {\color{node_4} 1411 } & {\color{node_4} 1286 } & {\color{node_4} 64 } & {\color{node_4} 1587 } & {\color{node_4} 1313 } \\ 
{\color{node_5} 309 } & {\color{node_5} 903 } & {\color{node_5} 1112 } & {\color{node_5} 1173 } & {\color{node_5} 1715 } & {\color{node_5} 1505 } & {\color{node_5} 286 } & {\color{node_5} 631 } \\ 
\end{bmatrix} } \\\\ 
\end{array}$$
\noindent $\color{node_4} \alpha_{4} $  receives 
$$ \begin{array}{l}
{\color{node_4} b_{4,i,j} = \begin{bmatrix}{\color{node_1} 258 } & {\color{node_1} 512 } & {\color{node_1} 904 } & {\color{node_1} 1466 } & {\color{node_1} 246 } & {\color{node_1} 1871 } & {\color{node_1} 141 } & {\color{node_1} 1210 } \\ 
{\color{node_2} 320 } & {\color{node_2} 1574 } & {\color{node_2} 1962 } & {\color{node_2} 2 } & {\color{node_2} 836 } & {\color{node_2} 428 } & {\color{node_2} 412 } & {\color{node_2} 166 } \\ 
{\color{node_3} 291 } & {\color{node_3} 1471 } & {\color{node_3} 1862 } & {\color{node_3} 125 } & {\color{node_3} 1396 } & {\color{node_3} 175 } & {\color{node_3} 1981 } & {\color{node_3} 124 } \\ 
{\color{node_4} 977 } & {\color{node_4} 717 } & {\color{node_4} 1366 } & {\color{node_4} 1641 } & {\color{node_4} 487 } & {\color{node_4} 1962 } & {\color{node_4} 1759 } & {\color{node_4} 11015 } \\ 
{\color{node_5} 1976 } & {\color{node_5} 1517 } & {\color{node_5} 825 } & {\color{node_5} 539 } & {\color{node_5} 1793 } & {\color{node_5} 234 } & {\color{node_5} 647 } & {\color{node_5} 1927 } \\ 
\end{bmatrix} } \\\\ 
\end{array}$$
\noindent $\color{node_5} \alpha_{5} $  receives 
$$ \begin{array}{l}
{\color{node_5} b_{5,i,j} = \begin{bmatrix}{\color{node_1} 1352 } & {\color{node_1} -11800 } & {\color{node_1} -11274 } & {\color{node_1} -13291 } & {\color{node_1} 3342 } & {\color{node_1} 703 } & {\color{node_1} 1852 } & {\color{node_1} 497 } \\ 
{\color{node_2} -5605 } & {\color{node_2} 3250 } & {\color{node_2} 2269 } & {\color{node_2} 2113 } & {\color{node_2} 1517 } & {\color{node_2} 2565 } & {\color{node_2} -6990 } & {\color{node_2} -11459 } \\ 
{\color{node_3} 84 } & {\color{node_3} 1359 } & {\color{node_3} -8458 } & {\color{node_3} 5053 } & {\color{node_3} -1780 } & {\color{node_3} -6256 } & {\color{node_3} -11987 } & {\color{node_3} 4348 } \\ 
{\color{node_4} 1598 } & {\color{node_4} 44 } & {\color{node_4} -9420 } & {\color{node_4} 1011 } & {\color{node_4} 2337 } & {\color{node_4} -1523 } & {\color{node_4} -1803 } & {\color{node_4} -10657 } \\ 
{\color{node_5} -9676 } & {\color{node_5} -7976 } & {\color{node_5} -9092 } & {\color{node_5} -9750 } & {\color{node_5} -10905 } & {\color{node_5} -9343 } & {\color{node_5} -8461 } & {\color{node_5} -11571 } \\ 
\end{bmatrix} } \\\\ 
\end{array}$$
This completes the bid sharing step of the auction protocol.

To ensure the scalability of the auction, 
the protocol examines the binary representation of each bid, 
beginning with the first digit $j=1$ and progressing towards the right until $j=k$.

\newpage For $j=1$, Algorithm \ref{ring_transfer_BP_algo} generates $B_1=1621$ and $P_1=675$.
  
$d_{l,i,1}$ are chosen randomly by the bidders using either their fake key or their check key, depending on their key values
\begin{equation}
    \prod_{i=1}^n d_{i,l,j}=\begin{cases}
{F_{l,j},} & {\text{if}}\ K_{l,j}=B_j \\ 
{C_{l,j},} & {\text{otherwise.}} 
    \end{cases}
\end{equation}
$K_{1,1} \neq B_{1}$, 
$K_{2,1} \neq B_{1}$, 
$K_{3,1} \neq B_{1}$, 
$K_{4,1} \neq B_{1}$, 
$K_{5,1} \neq B_{1}$, 

\noindent $\color{node_1} \alpha_{1} $ uses $\color{node_1} C_{1,1}=849$ to generate 
$$ 
\begin{array}{l} {\color{node_1} d_{1,i,1} = 
 \begin{bmatrix} 
705 & 1479 & 307 & 1292 & 1112 \end{bmatrix}^\intercal }  
 \end{array}$$ 

\noindent $\color{node_2} \alpha_{2} $ uses $\color{node_2} C_{2,1}=1563$ to generate 
$$ \begin{array}{l} {\color{node_2} d_{2,i,1} = 
 \begin{bmatrix} 
1761 & 291 & 789 & 648 & 441 \end{bmatrix}^\intercal }  
 \end{array}$$ 

\noindent $\color{node_3} \alpha_{3} $ uses $\color{node_3} C_{3,1}=664$ to generate 
$$ \begin{array}{l} {\color{node_3} d_{3,i,1} = 
 \begin{bmatrix} 
563 & 747 & 1321 & 1363 & 261 \end{bmatrix}^\intercal }  
 \end{array}$$ 

\noindent $\color{node_4} \alpha_{4} $ uses $\color{node_4} C_{4,1}=1969$ to generate 
$$ \begin{array}{l} {\color{node_4} d_{4,i,1} = 
 \begin{bmatrix} 
1034 & 204 & 43 & 421 & 1016 \end{bmatrix}^\intercal }  
 \end{array}$$ 

\noindent $\color{node_5} \alpha_{5} $ uses $\color{node_5} C_{5,1}=2058$ to generate 
$$ \begin{array}{l} {\color{node_5} d_{5,i,1} = 
 \begin{bmatrix} 
305 & 1372 & 960 & 235 & 1205 \end{bmatrix}^\intercal }  
 \end{array}$$ 
 The chosen $d_{i,l,j}$ values are sent to $\alpha_i$ for all $i \neq l$, 
 so
 
 \noindent $\color{node_1} \alpha_{1} $  receives 
$$  \begin{array}{l}
{\color{node_1} d_{i,1,j} = \begin{bmatrix} 
{\color{node_1}  705} &  
 {\color{node_2}  1761} &  
 {\color{node_3}  563} &  
 {\color{node_4}  1034} &  
 {\color{node_5}  305 }  \\ 
 \end{bmatrix}^\intercal  } \quad 
\end{array}$$
\noindent $\color{node_2} \alpha_{2} $  receives 
$$  \begin{array}{l}
{\color{node_2} d_{i,2,j} = \begin{bmatrix} 
{\color{node_1}  1479} &  
 {\color{node_2}  291} &  
 {\color{node_3}  747} &  
 {\color{node_4}  204} &  
 {\color{node_5}  1372 }  \\ 
 \end{bmatrix}^\intercal  } \quad 
\end{array}$$
\noindent $\color{node_3} \alpha_{3} $  receives 
$$  \begin{array}{l}
{\color{node_3} d_{i,3,j} = \begin{bmatrix} 
{\color{node_1}  307} &  
 {\color{node_2}  789} &  
 {\color{node_3}  1321} &  
 {\color{node_4}  43} &  
 {\color{node_5}  960 }  \\ 
 \end{bmatrix}^\intercal  } \quad 
\end{array}$$
\noindent $\color{node_4} \alpha_{4} $  receives 
$$  \begin{array}{l}
{\color{node_4} d_{i,4,j} = \begin{bmatrix} 
{\color{node_1}  1292} &  
 {\color{node_2}  648} &  
 {\color{node_3}  1363} &  
 {\color{node_4}  421} &  
 {\color{node_5}  235 }  \\ 
 \end{bmatrix}^\intercal  } \quad 
\end{array}$$
\noindent $\color{node_5} \alpha_{5} $  receives 
$$  \begin{array}{l}
{\color{node_5} d_{i,5,j} = \begin{bmatrix} 
{\color{node_1}  1112} &  
 {\color{node_2}  441} &  
 {\color{node_3}  261} &  
 {\color{node_4}  1016} &  
 {\color{node_5}  1205 }  \\ 
 \end{bmatrix}^\intercal  } \,.
\end{array}$$\bigskip

According to Algorithm \ref{ring_transfer_D_algo}, $D_1=1407$ is generated.
We first calculate $P_1^{n-2}=1024$, and then deduce that $$D_1 \neq P_1^{n-2} \implies p_{0,1} = 1$$

\begin{table}[h] 
\center\resizebox* {!} {6cm}{ 
\begin{tabular}{  c | c c | c c | c c | c c | c c | c c | c c | c  c } 
& \multicolumn{2}{c|}{$j=1$} 
& \multicolumn{2}{c|}{$j=2$} 
& \multicolumn{2}{c|}{$j=3$} 
& \multicolumn{2}{c|}{$j=4$} 
& \multicolumn{2}{c|}{$j=5$} 
& \multicolumn{2}{c|}{$j=6$} 
& \multicolumn{2}{c|}{$j=7$} 
& \multicolumn{2}{c|}{$j=8$} 
\\ \cline{1-17} 
\multirow{2}{*}{{\color{node_1}  $p_1=143$}}   
& {\color{node_1}  1} &  
& {\color{node_1}  0} &  
& {\color{node_1}  0} &  
& {\color{node_1}  0} &  
& {\color{node_1}  1} &  
& {\color{node_1}  1} &  
& {\color{node_1}  1} &  
& {\color{node_1}  1} &  
\multicolumn{1}{c|}{} 
 \\ \cline{2-17} 
& \rotatebox[origin=c]{180}{$\Lsh$} & \multicolumn{1}{|c|}{\cellcolor {gray!30} 1} 
& \rotatebox[origin=c]{180}{$\Lsh$} & \multicolumn{1}{|c|}{\cellcolor {gray!30}  } 
& \rotatebox[origin=c]{180}{$\Lsh$} & \multicolumn{1}{|c|}{\cellcolor {gray!30}  } 
& \rotatebox[origin=c]{180}{$\Lsh$} & \multicolumn{1}{|c|}{\cellcolor {gray!30}  } 
& \rotatebox[origin=c]{180}{$\Lsh$} & \multicolumn{1}{|c|}{\cellcolor {gray!30}  } 
& \rotatebox[origin=c]{180}{$\Lsh$} & \multicolumn{1}{|c|}{\cellcolor {gray!30}  } 
& \rotatebox[origin=c]{180}{$\Lsh$} & \multicolumn{1}{|c|}{\cellcolor {gray!30}  } 
& \rotatebox[origin=c]{180}{$\Lsh$} & \multicolumn{1}{|c|}{\cellcolor {gray!30}  } 
\\ \cline{1-17} 
\multirow{2}{*}{{\color{node_2}  $p_2=124$}}   
& {\color{node_2}  0} &  
& {\color{node_2}  1} &  
& {\color{node_2}  1} &  
& {\color{node_2}  1} &  
& {\color{node_2}  1} &  
& {\color{node_2}  1} &  
& {\color{node_2}  0} &  
& {\color{node_2}  0} &  
\multicolumn{1}{c|}{} 
 \\
\end{tabular}} 
\caption{Two or more bidders bidding on digit $j=1$.}
 \end{table} 
 
 As described in \autoref{bids_adjust_2} of Algorithm \ref{2nd_price_algo},
 unsuccessful bidders whose bids are lower than the second-highest bid quit the competition by changing their bids to `0'. 
 The adjustment is made by altering all digits after the $j^{th}$ digit of bids $b_{l,l,w} (w > j)$ so that the following equation is satisfied
\begin{equation}
\sum_{i=1}^n b_{i,l,w}=N_{l,w} \,.
\end{equation}
This adjustment is
$$  \begin{array}{l}
{\color{node_2} b_{2,2,j} = 
 \begin{bmatrix} 
221 & 1905 & 478 & 1212 & 508 & 364 & 866 & 1801 \\
 \end{bmatrix} } 
\end{array}$$
$$ \begin{array}{l}
{\color{node_5} b_{5,5,j} = 
 \begin{bmatrix} 
-9676 & 1548 & -9092 & 1936 & -10905 & 1817 & 3087 & -11571 \\
 \end{bmatrix} } 
\end{array}$$

\newpage For $j=2$, Algorithm \ref{ring_transfer_BP_algo} generates $B_2=768$ and $P_2=888$. 
$K_{1,2} \neq B_{2}$, 
$K_{2,2} \neq B_{2}$, 
$K_{3,2} \neq B_{2}$, 
$K_{4,2} \neq B_{2}$, 
$K_{5,2} \neq B_{2}$, 

\noindent $\color{node_1} \alpha_{1} $ uses $\color{node_1} C_{1,2}=704$ to generate 
$$ \begin{array}{l} {\color{node_1} d_{1,i,2} = 
 \begin{bmatrix} 
1654 & 27 & 1638 & 739 & 542 \end{bmatrix}^\intercal }  
 \end{array}$$ 

\noindent $\color{node_2} \alpha_{2} $ uses $\color{node_2} C_{2,2}=946$ to generate 
$$ \begin{array}{l} {\color{node_2} d_{2,i,2} = 
 \begin{bmatrix} 
937 & 1248 & 321 & 2054 & 1603 \end{bmatrix}^\intercal }  
 \end{array}$$ 

\noindent $\color{node_3} \alpha_{3} $ uses $\color{node_3} C_{3,2}=1835$ to generate 
$$ \begin{array}{l} {\color{node_3} d_{3,i,2} = 
 \begin{bmatrix} 
1022 & 1964 & 833 & 268 & 940 \end{bmatrix}^\intercal }  
 \end{array}$$ 

\noindent $\color{node_4} \alpha_{4} $ uses $\color{node_4} C_{4,2}=1029$ to generate 
$$ \begin{array}{l} {\color{node_4} d_{4,i,2} = 
 \begin{bmatrix} 
1331 & 1951 & 433 & 794 & 304 \end{bmatrix}^\intercal }  
 \end{array}$$ 

\noindent $\color{node_5} \alpha_{5} $ uses $\color{node_5} C_{5,2}=104$ to generate 
$$ \begin{array}{l} {\color{node_5} d_{5,i,2} = 
 \begin{bmatrix} 
1355 & 1331 & 428 & 834 & 1092 \end{bmatrix}^\intercal }  
 \end{array}$$ 
\noindent $\color{node_1} \alpha_{1} $  receives 
$$  \begin{array}{l}
{\color{node_1} d_{i,1,j} = \begin{bmatrix} 
{\color{node_1}  1654} &  
 {\color{node_2}  937} &  
 {\color{node_3}  1022} &  
 {\color{node_4}  1331} &  
 {\color{node_5}  1355 }  \\ 
 \end{bmatrix}^\intercal  } \quad 
\end{array}$$
\noindent $\color{node_2} \alpha_{2} $  receives 
$$  \begin{array}{l}
{\color{node_2} d_{i,2,j} = \begin{bmatrix} 
{\color{node_1}  27} &  
 {\color{node_2}  1248} &  
 {\color{node_3}  1964} &  
 {\color{node_4}  1951} &  
 {\color{node_5}  1331 }  \\ 
 \end{bmatrix}^\intercal  } \quad 
\end{array}$$
\noindent $\color{node_3} \alpha_{3} $  receives 
$$  \begin{array}{l}
{\color{node_3} d_{i,3,j} = \begin{bmatrix} 
{\color{node_1}  1638} &  
 {\color{node_2}  321} &  
 {\color{node_3}  833} &  
 {\color{node_4}  433} &  
 {\color{node_5}  428 }  \\ 
 \end{bmatrix}^\intercal  } \quad 
\end{array}$$
\noindent $\color{node_4} \alpha_{4} $  receives 
$$  \begin{array}{l}
{\color{node_4} d_{i,4,j} = \begin{bmatrix} 
{\color{node_1}  739} &  
 {\color{node_2}  2054} &  
 {\color{node_3}  268} &  
 {\color{node_4}  794} &  
 {\color{node_5}  834 }  \\ 
 \end{bmatrix}^\intercal  } \quad 
\end{array}$$
\noindent $\color{node_5} \alpha_{5} $  receives 
$$  \begin{array}{l}
{\color{node_5} d_{i,5,j} = \begin{bmatrix} 
{\color{node_1}  542} &  
 {\color{node_2}  1603} &  
 {\color{node_3}  940} &  
 {\color{node_4}  304} &  
 {\color{node_5}  1092 }  \\ 
 \end{bmatrix}^\intercal  } \quad 
\end{array}.$$
\bigskip
 \newpage According to Algorithm \ref{ring_transfer_D_algo}, $D_1=1407$ is generated.

 We compute $P_2^{n-2}=1549$, and then deduct that
$$D_2 \neq P_2^{n-2} \implies p_{0,2} = 1 $$  

\begin{table}[h] 
\center\resizebox* {!} {6cm}{ 
\begin{tabular}{  c | c c | c c | c c | c c | c c | c c | c c | c  c } 
& \multicolumn{2}{c|}{$j=1$} 
& \multicolumn{2}{c|}{$j=2$} 
& \multicolumn{2}{c|}{$j=3$} 
& \multicolumn{2}{c|}{$j=4$} 
& \multicolumn{2}{c|}{$j=5$} 
& \multicolumn{2}{c|}{$j=6$} 
& \multicolumn{2}{c|}{$j=7$} 
& \multicolumn{2}{c|}{$j=8$} 
\\ \cline{1-17} 
\multirow{2}{*}{{\color{node_1}  $p_1=143$}}   
& {\color{node_1}  1} &  
& {\color{node_1}  0} &  
& {\color{node_1}  0} &  
& {\color{node_1}  0} &  
& {\color{node_1}  1} &  
& {\color{node_1}  1} &  
& {\color{node_1}  1} &  
& {\color{node_1}  1} &  
\multicolumn{1}{c|}{} 
 \\ \cline{2-17} 
& \rotatebox[origin=c]{180}{$\Lsh$} & \multicolumn{1}{|c|}{\cellcolor {gray!30} 1} 
& \rotatebox[origin=c]{180}{$\Lsh$} & \multicolumn{1}{|c|}{\cellcolor {gray!30} 0} 
& \rotatebox[origin=c]{180}{$\Lsh$} & \multicolumn{1}{|c|}{\cellcolor {gray!30} 0} 
& \rotatebox[origin=c]{180}{$\Lsh$} & \multicolumn{1}{|c|}{\cellcolor {gray!30} 0} 
& \rotatebox[origin=c]{180}{$\Lsh$} & \multicolumn{1}{|c|}{\cellcolor {gray!30} 0} 
& \rotatebox[origin=c]{180}{$\Lsh$} & \multicolumn{1}{|c|}{\cellcolor {gray!30} 0} 
& \rotatebox[origin=c]{180}{$\Lsh$} & \multicolumn{1}{|c|}{\cellcolor {gray!30} 0} 
& \rotatebox[origin=c]{180}{$\Lsh$} & \multicolumn{1}{|c|}{\cellcolor {gray!30} 0} 
\\ \cline{1-17} 
\multirow{2}{*}{{\color{node_2}  $p_2=124$}}   
& {\color{node_2}  0} &  
& {\color{node_2}  1} &  
& {\color{node_2}  1} &  
& {\color{node_2}  1} &  
& {\color{node_2}  1} &  
& {\color{node_2}  1} &  
& {\color{node_2}  0} &  
& {\color{node_2}  0} &  
\multicolumn{1}{c|}{} 
 \\
\end{tabular}} 
 \caption{Two or more bidders for $j=2$} 
 \end{table} 
Bid adjustment is
$$  \begin{array}{l}
{\color{node_1} b_{1,1,j} = 
 \begin{bmatrix} 
1050 & 1613 & 457 & 1472 & 219 & 442 & 839 & 891 \\
 \end{bmatrix} } 
\end{array}$$

\bigskip

\newpage For $j=3$, Algorithm \ref{ring_transfer_BP_algo} generates $B_3=1394$ and $P_3=534$.  
$K_{1,3} \neq B_{3}$, 
$K_{2,3} \neq B_{3}$, 
$K_{3,3} \neq B_{3}$, 
$K_{4,3} \neq B_{3}$, 
$K_{5,3} \neq B_{3}$, 

\noindent $\color{node_1} \alpha_{1} $ uses $\color{node_1} C_{1,3}=251$ to generate 
$$ \begin{array}{l} {\color{node_1} d_{1,i,3} = 
 \begin{bmatrix} 
421 & 1328 & 965 & 1399 & 249 \end{bmatrix}^\intercal }  
 \end{array}$$ 

\noindent $\color{node_2} \alpha_{2} $ uses $\color{node_2} C_{2,3}=1343$ to generate 
$$ \begin{array}{l} {\color{node_2} d_{2,i,3} = 
 \begin{bmatrix} 
1534 & 1142 & 1065 & 2032 & 1836 \end{bmatrix}^\intercal }  
 \end{array}$$ 

\noindent $\color{node_3} \alpha_{3} $ uses $\color{node_3} C_{3,3}=601$ to generate 
$$ \begin{array}{l} {\color{node_3} d_{3,i,3} = 
 \begin{bmatrix} 
1872 & 481 & 1320 & 1629 & 849 \end{bmatrix}^\intercal }  
 \end{array}$$ 

\noindent $\color{node_4} \alpha_{4} $ uses $\color{node_4} C_{4,3}=1330$ to generate 
$$ \begin{array}{l} {\color{node_4} d_{4,i,3} = 
 \begin{bmatrix} 
680 & 1202 & 710 & 979 & 84 \end{bmatrix}^\intercal }  
 \end{array}$$ 

\noindent $\color{node_5} \alpha_{5} $ uses $\color{node_5} C_{5,3}=99$ to generate 
$$ \begin{array}{l} {\color{node_5} d_{5,i,3} = 
 \begin{bmatrix} 
1700 & 130 & 2061 & 1052 & 1608 \end{bmatrix}^\intercal }  
 \end{array}$$ 
\noindent $\color{node_1} \alpha_{1} $  receives 
$$  \begin{array}{l}
{\color{node_1} d_{i,1,j} = \begin{bmatrix} 
{\color{node_1}  421} &  
 {\color{node_2}  1534} &  
 {\color{node_3}  1872} &  
 {\color{node_4}  680} &  
 {\color{node_5}  1700 }  \\ 
 \end{bmatrix}^\intercal  } \quad 
\end{array}$$
\noindent $\color{node_2} \alpha_{2} $  receives 
$$  \begin{array}{l}
{\color{node_2} d_{i,2,j} = \begin{bmatrix} 
{\color{node_1}  1328} &  
 {\color{node_2}  1142} &  
 {\color{node_3}  481} &  
 {\color{node_4}  1202} &  
 {\color{node_5}  130 }  \\ 
 \end{bmatrix}^\intercal  } \quad 
\end{array}$$
\noindent $\color{node_3} \alpha_{3} $  receives 
$$  \begin{array}{l}
{\color{node_3} d_{i,3,j} = \begin{bmatrix} 
{\color{node_1}  965} &  
 {\color{node_2}  1065} &  
 {\color{node_3}  1320} &  
 {\color{node_4}  710} &  
 {\color{node_5}  2061 }  \\ 
 \end{bmatrix}^\intercal  } \quad 
\end{array}$$
\noindent $\color{node_4} \alpha_{4} $  receives 
$$  \begin{array}{l}
{\color{node_4} d_{i,4,j} = \begin{bmatrix} 
{\color{node_1}  1399} &  
 {\color{node_2}  2032} &  
 {\color{node_3}  1629} &  
 {\color{node_4}  979} &  
 {\color{node_5}  1052 }  \\ 
 \end{bmatrix}^\intercal  } \quad 
\end{array}$$
\noindent $\color{node_5} \alpha_{5} $  receives 
$$  \begin{array}{l}
{\color{node_5} d_{i,5,j} = \begin{bmatrix} 
{\color{node_1}  249} &  
 {\color{node_2}  1836} &  
 {\color{node_3}  849} &  
 {\color{node_4}  84} &  
 {\color{node_5}  1608 }  \\ 
 \end{bmatrix}^\intercal  } \quad 
\end{array}.$$

 \newpage According to Algorithm \ref{ring_transfer_D_algo}, $D_{3} = 1211$ is generated. 

 We compute $P_3^{n-2}=1211$, and then deduct that
$$D_3 = P_3^{n-2} \implies p_{0,3} = 0 $$ 
\begin{table}[h] 
\center\resizebox* {!} {6cm}{ 
\begin{tabular}{  c | c c | c c | c c | c c | c c | c c | c c | c  c } 
& \multicolumn{2}{c|}{$j=1$} 
& \multicolumn{2}{c|}{$j=2$} 
& \multicolumn{2}{c|}{$j=3$} 
& \multicolumn{2}{c|}{$j=4$} 
& \multicolumn{2}{c|}{$j=5$} 
& \multicolumn{2}{c|}{$j=6$} 
& \multicolumn{2}{c|}{$j=7$} 
& \multicolumn{2}{c|}{$j=8$} 
\\ \cline{1-17} 
\multirow{2}{*}{{\color{node_1}  $p_1=143$}}   
& {\color{node_1}  1} &  
& {\color{node_1}  0} &  
& {\color{node_1}  0} &  
& {\color{node_1}  0} &  
& {\color{node_1}  1} &  
& {\color{node_1}  1} &  
& {\color{node_1}  1} &  
& {\color{node_1}  1} &  
\multicolumn{1}{c|}{} 
 \\ \cline{2-17} 
& \rotatebox[origin=c]{180}{$\Lsh$} & \multicolumn{1}{|c|}{\cellcolor {gray!30} 1} 
& \rotatebox[origin=c]{180}{$\Lsh$} & \multicolumn{1}{|c|}{\cellcolor {gray!30} 0} 
& \rotatebox[origin=c]{180}{$\Lsh$} & \multicolumn{1}{|c|}{\cellcolor {gray!30} 0} 
& \rotatebox[origin=c]{180}{$\Lsh$} & \multicolumn{1}{|c|}{\cellcolor {gray!30} 0} 
& \rotatebox[origin=c]{180}{$\Lsh$} & \multicolumn{1}{|c|}{\cellcolor {gray!30} 0} 
& \rotatebox[origin=c]{180}{$\Lsh$} & \multicolumn{1}{|c|}{\cellcolor {gray!30} 0} 
& \rotatebox[origin=c]{180}{$\Lsh$} & \multicolumn{1}{|c|}{\cellcolor {gray!30} 0} 
& \rotatebox[origin=c]{180}{$\Lsh$} & \multicolumn{1}{|c|}{\cellcolor {gray!30} 0} 
\\ \cline{1-17} 
\multirow{2}{*}{{\color{node_2}  $p_2=124$}}   
& {\color{node_2}  0} &  
& {\color{node_2}  1} &  
& {\color{node_2}  1} &  
& {\color{node_2}  1} &  
& {\color{node_2}  1} &  
& {\color{node_2}  1} &  
& {\color{node_2}  0} &  
& {\color{node_2}  0} &  
\multicolumn{1}{c|}{} 
 \\
\end{tabular}} 
 \caption{One or no bidder for $j=3$} 
 \end{table} 
\bigskip

\newpage For $j=4$, Algorithm \ref{ring_transfer_BP_algo} generates $B_4=1674$ and $P_4=163$.  
$K_{1,4} \neq B_{4}$, 
$K_{2,4} \neq B_{4}$, 
$K_{3,4} \neq B_{4}$, 
$K_{4,4} \neq B_{4}$, 
$K_{5,4} \neq B_{4}$, 

\noindent $\color{node_1} \alpha_{1} $ uses $\color{node_1} C_{1,4}=1171$ to generate 
$$ \begin{array}{l} {\color{node_1} d_{1,i,4} = 
 \begin{bmatrix} 
1043 & 462 & 1366 & 1115 & 1911 \end{bmatrix}^\intercal }  
 \end{array}$$ 

\noindent $\color{node_2} \alpha_{2} $ uses $\color{node_2} C_{2,4}=858$ to generate 
$$ \begin{array}{l} {\color{node_2} d_{2,i,4} = 
 \begin{bmatrix} 
32 & 760 & 482 & 1779 & 1008 \end{bmatrix}^\intercal }  
 \end{array}$$ 

\noindent $\color{node_3} \alpha_{3} $ uses $\color{node_3} C_{3,4}=1597$ to generate 
$$ \begin{array}{l} {\color{node_3} d_{3,i,4} = 
 \begin{bmatrix} 
992 & 1756 & 1635 & 656 & 434 \end{bmatrix}^\intercal }  
 \end{array}$$ 

\noindent $\color{node_4} \alpha_{4} $ uses $\color{node_4} C_{4,4}=1580$ to generate 
$$ \begin{array}{l} {\color{node_4} d_{4,i,4} = 
 \begin{bmatrix} 
1593 & 2005 & 1264 & 1854 & 357 \end{bmatrix}^\intercal }  
 \end{array}$$ 

\noindent $\color{node_5} \alpha_{5} $ uses $\color{node_5} C_{5,4}=1869$ to generate 
$$ \begin{array}{l} {\color{node_5} d_{5,i,4} = 
 \begin{bmatrix} 
458 & 385 & 306 & 994 & 1604 \end{bmatrix}^\intercal }  
 \end{array}$$ 
\noindent $\color{node_1} \alpha_{1} $  receives 
$$  \begin{array}{l}
{\color{node_1} d_{i,1,j} = \begin{bmatrix} 
{\color{node_1}  1043} &  
 {\color{node_2}  32} &  
 {\color{node_3}  992} &  
 {\color{node_4}  1593} &  
 {\color{node_5}  458 }  \\ 
 \end{bmatrix}^\intercal  } \quad 
\end{array}$$
\noindent $\color{node_2} \alpha_{2} $  receives 
$$  \begin{array}{l}
{\color{node_2} d_{i,2,j} = \begin{bmatrix} 
{\color{node_1}  462} &  
 {\color{node_2}  760} &  
 {\color{node_3}  1756} &  
 {\color{node_4}  2005} &  
 {\color{node_5}  385 }  \\ 
 \end{bmatrix}^\intercal  } \quad 
\end{array}$$
\noindent $\color{node_3} \alpha_{3} $  receives 
$$  \begin{array}{l}
{\color{node_3} d_{i,3,j} = \begin{bmatrix} 
{\color{node_1}  1366} &  
 {\color{node_2}  482} &  
 {\color{node_3}  1635} &  
 {\color{node_4}  1264} &  
 {\color{node_5}  306 }  \\ 
 \end{bmatrix}^\intercal  } \quad 
\end{array}$$
\noindent $\color{node_4} \alpha_{4} $  receives 
$$  \begin{array}{l}
{\color{node_4} d_{i,4,j} = \begin{bmatrix} 
{\color{node_1}  1115} &  
 {\color{node_2}  1779} &  
 {\color{node_3}  656} &  
 {\color{node_4}  1854} &  
 {\color{node_5}  994 }  \\ 
 \end{bmatrix}^\intercal  } \quad 
\end{array}$$
\noindent $\color{node_5} \alpha_{5} $  receives 
$$  \begin{array}{l}
{\color{node_5} d_{i,5,j} = \begin{bmatrix} 
{\color{node_1}  1911} &  
 {\color{node_2}  1008} &  
 {\color{node_3}  434} &  
 {\color{node_4}  357} &  
 {\color{node_5}  1604 }  \\ 
 \end{bmatrix}^\intercal  } \quad 
\end{array}.$$
 \newpage According to Algorithm \ref{ring_transfer_D_algo}, $D_{4} = 1772$ is generated. 

 We compute $P_4^{n-2}=510$, and then deduct that
$$D_4 \neq P_4^{n-2} \implies p_{0,4} = 1 $$

\begin{table}[h] 
\center\resizebox* {!} {6cm}{ 
\begin{tabular}{  c | c c | c c | c c | c c | c c | c c | c c | c  c } 
& \multicolumn{2}{c|}{$j=1$} 
& \multicolumn{2}{c|}{$j=2$} 
& \multicolumn{2}{c|}{$j=3$} 
& \multicolumn{2}{c|}{$j=4$} 
& \multicolumn{2}{c|}{$j=5$} 
& \multicolumn{2}{c|}{$j=6$} 
& \multicolumn{2}{c|}{$j=7$} 
& \multicolumn{2}{c|}{$j=8$} 
\\ \cline{1-17} 
\multirow{2}{*}{{\color{node_1}  $p_1=143$}}   
& {\color{node_1}  1} &  
& {\color{node_1}  0} &  
& {\color{node_1}  0} &  
& {\color{node_1}  0} &  
& {\color{node_1}  1} &  
& {\color{node_1}  1} &  
& {\color{node_1}  1} &  
& {\color{node_1}  1} &  
\multicolumn{1}{c|}{} 
 \\ \cline{2-17} 
& \rotatebox[origin=c]{180}{$\Lsh$} & \multicolumn{1}{|c|}{\cellcolor {gray!30} 1} 
& \rotatebox[origin=c]{180}{$\Lsh$} & \multicolumn{1}{|c|}{\cellcolor {gray!30} 0} 
& \rotatebox[origin=c]{180}{$\Lsh$} & \multicolumn{1}{|c|}{\cellcolor {gray!30} 0} 
& \rotatebox[origin=c]{180}{$\Lsh$} & \multicolumn{1}{|c|}{\cellcolor {gray!30} 0} 
& \rotatebox[origin=c]{180}{$\Lsh$} & \multicolumn{1}{|c|}{\cellcolor {gray!30} 0} 
& \rotatebox[origin=c]{180}{$\Lsh$} & \multicolumn{1}{|c|}{\cellcolor {gray!30} 0} 
& \rotatebox[origin=c]{180}{$\Lsh$} & \multicolumn{1}{|c|}{\cellcolor {gray!30} 0} 
& \rotatebox[origin=c]{180}{$\Lsh$} & \multicolumn{1}{|c|}{\cellcolor {gray!30} 0} 
\\ \cline{1-17} 
\multirow{2}{*}{{\color{node_2}  $p_2=124$}}   
& {\color{node_2}  0} &  
& {\color{node_2}  1} &  
& {\color{node_2}  1} &  
& {\color{node_2}  1} &  
& {\color{node_2}  1} &  
& {\color{node_2}  1} &  
& {\color{node_2}  0} &  
& {\color{node_2}  0} &  
\multicolumn{1}{c|}{} 
 \\
\end{tabular}} 
 \caption{Two or more bidders for $j=4$} 
 \end{table}

\bigskip

\newpage For $j=5$, Algorithm \ref{ring_transfer_BP_algo} generates $B_5=454$ and $P_5=2025$.  
$K_{1,5} \neq B_{5}$, 
$K_{2,5} \neq B_{5}$, 
$K_{3,5} \neq B_{5}$, 
$K_{4,5} \neq B_{5}$, 
$K_{5,5} \neq B_{5}$, 

\noindent $\color{node_1} \alpha_{1} $ uses $\color{node_1} C_{1,5}=1914$ to generate 
$$ \begin{array}{l} {\color{node_1} d_{1,i,5} = 
 \begin{bmatrix} 
1771 & 1014 & 1060 & 876 & 1474 \end{bmatrix}^\intercal }  
 \end{array}$$ 

\noindent $\color{node_2} \alpha_{2} $ uses $\color{node_2} C_{2,5}=1835$ to generate 
$$ \begin{array}{l} {\color{node_2} d_{2,i,5} = 
 \begin{bmatrix} 
846 & 249 & 23 & 182 & 1616 \end{bmatrix}^\intercal }  
 \end{array}$$ 

\noindent $\color{node_3} \alpha_{3} $ uses $\color{node_3} C_{3,5}=349$ to generate 
$$ \begin{array}{l} {\color{node_3} d_{3,i,5} = 
 \begin{bmatrix} 
1086 & 2034 & 1059 & 351 & 136 \end{bmatrix}^\intercal }  
 \end{array}$$ 

\noindent $\color{node_4} \alpha_{4} $ uses $\color{node_4} C_{4,5}=900$ to generate 
$$ \begin{array}{l} {\color{node_4} d_{4,i,5} = 
 \begin{bmatrix} 
494 & 1940 & 219 & 177 & 1045 \end{bmatrix}^\intercal }  
 \end{array}$$ 

\noindent $\color{node_5} \alpha_{5} $ uses $\color{node_5} C_{5,5}=614$ to generate 
$$ \begin{array}{l} {\color{node_5} d_{5,i,5} = 
 \begin{bmatrix} 
1947 & 931 & 860 & 13 & 94 \end{bmatrix}^\intercal }  
 \end{array}$$ 
\noindent $\color{node_1} \alpha_{1} $  receives 
$$  \begin{array}{l}
{\color{node_1} d_{i,1,j} = \begin{bmatrix} 
{\color{node_1}  1771} &  
 {\color{node_2}  846} &  
 {\color{node_3}  1086} &  
 {\color{node_4}  494} &  
 {\color{node_5}  1947 }  \\ 
 \end{bmatrix}^\intercal  } \quad 
\end{array}$$
\noindent $\color{node_2} \alpha_{2} $  receives 
$$  \begin{array}{l}
{\color{node_2} d_{i,2,j} = \begin{bmatrix} 
{\color{node_1}  1014} &  
 {\color{node_2}  249} &  
 {\color{node_3}  2034} &  
 {\color{node_4}  1940} &  
 {\color{node_5}  931 }  \\ 
 \end{bmatrix}^\intercal  } \quad 
\end{array}$$
\noindent $\color{node_3} \alpha_{3} $  receives 
$$  \begin{array}{l}
{\color{node_3} d_{i,3,j} = \begin{bmatrix} 
{\color{node_1}  1060} &  
 {\color{node_2}  23} &  
 {\color{node_3}  1059} &  
 {\color{node_4}  219} &  
 {\color{node_5}  860 }  \\ 
 \end{bmatrix}^\intercal  } \quad 
\end{array}$$
\noindent $\color{node_4} \alpha_{4} $  receives 
$$  \begin{array}{l}
{\color{node_4} d_{i,4,j} = \begin{bmatrix} 
{\color{node_1}  876} &  
 {\color{node_2}  182} &  
 {\color{node_3}  351} &  
 {\color{node_4}  177} &  
 {\color{node_5}  13 }  \\ 
 \end{bmatrix}^\intercal  } \quad 
\end{array}$$
\noindent $\color{node_5} \alpha_{5} $  receives 
$$  \begin{array}{l}
{\color{node_5} d_{i,5,j} = \begin{bmatrix} 
{\color{node_1}  1474} &  
 {\color{node_2}  1616} &  
 {\color{node_3}  136} &  
 {\color{node_4}  1045} &  
 {\color{node_5}  94 }  \\ 
 \end{bmatrix}^\intercal  } \quad 
\end{array}.$$
 \newpage According to Algorithm \ref{ring_transfer_D_algo}, $D_{5} = 599$ is generated. 

 We compute $P_5^{n-2}=829$, and then deduct that
$$D_5 \neq P_5^{n-2} \implies p_{0,5} = 1 $$ 

\begin{table}[h] 
\center\resizebox* {!} {6cm}{ 
\begin{tabular}{  c | c c | c c | c c | c c | c c | c c | c c | c  c } 
& \multicolumn{2}{c|}{$j=1$} 
& \multicolumn{2}{c|}{$j=2$} 
& \multicolumn{2}{c|}{$j=3$} 
& \multicolumn{2}{c|}{$j=4$} 
& \multicolumn{2}{c|}{$j=5$} 
& \multicolumn{2}{c|}{$j=6$} 
& \multicolumn{2}{c|}{$j=7$} 
& \multicolumn{2}{c|}{$j=8$} 
\\ \cline{1-17} 
\multirow{2}{*}{{\color{node_1}  $p_1=143$}}   
& {\color{node_1}  1} &  
& {\color{node_1}  0} &  
& {\color{node_1}  0} &  
& {\color{node_1}  0} &  
& {\color{node_1}  1} &  
& {\color{node_1}  1} &  
& {\color{node_1}  1} &  
& {\color{node_1}  1} &  
\multicolumn{1}{c|}{} 
 \\ \cline{2-17} 
& \rotatebox[origin=c]{180}{$\Lsh$} & \multicolumn{1}{|c|}{\cellcolor {gray!30} 1} 
& \rotatebox[origin=c]{180}{$\Lsh$} & \multicolumn{1}{|c|}{\cellcolor {gray!30} 0} 
& \rotatebox[origin=c]{180}{$\Lsh$} & \multicolumn{1}{|c|}{\cellcolor {gray!30} 0} 
& \rotatebox[origin=c]{180}{$\Lsh$} & \multicolumn{1}{|c|}{\cellcolor {gray!30} 0} 
& \rotatebox[origin=c]{180}{$\Lsh$} & \multicolumn{1}{|c|}{\cellcolor {gray!30} 0} 
& \rotatebox[origin=c]{180}{$\Lsh$} & \multicolumn{1}{|c|}{\cellcolor {gray!30} 0} 
& \rotatebox[origin=c]{180}{$\Lsh$} & \multicolumn{1}{|c|}{\cellcolor {gray!30} 0} 
& \rotatebox[origin=c]{180}{$\Lsh$} & \multicolumn{1}{|c|}{\cellcolor {gray!30} 0} 
\\ \cline{1-17} 
\multirow{2}{*}{{\color{node_2}  $p_2=124$}}   
& {\color{node_2}  0} &  
& {\color{node_2}  1} &  
& {\color{node_2}  1} &  
& {\color{node_2}  1} &  
& {\color{node_2}  1} &  
& {\color{node_2}  1} &  
& {\color{node_2}  0} &  
& {\color{node_2}  0} &  
\multicolumn{1}{c|}{} 
 \\
\end{tabular}} 
 \caption{Two or more bidders for $j=5$} 
 \end{table}

\bigskip

\newpage For $j=6$, Algorithm \ref{ring_transfer_BP_algo} generates $B_6=346$ and $P_6=820$.  
$K_{1,6} \neq B_{6}$, 
$K_{2,6} \neq B_{6}$, 
$K_{3,6} \neq B_{6}$, 
$K_{5,6} \neq B_{6}$, 
$K_{4,6} = B_{6}$, $\implies$ $\alpha_{4}$ is the only bidder for $j=6$
$$ \begin{array}{l}
{\color{node_4} b_{4,4,j} = 
 \begin{bmatrix} 
977 & 717 & 1366 & 1641 & 487 & 1962 & 1759 & 1977 \\
 \end{bmatrix} } 
\end{array}$$

\noindent $\color{node_1} \alpha_{1} $ uses $\color{node_1} C_{1,6}=1206$ to generate 
$$ \begin{array}{l} {\color{node_1} d_{1,i,6} = 
 \begin{bmatrix} 
1149 & 1790 & 769 & 1280 & 362 \end{bmatrix}^\intercal }  
 \end{array}$$ 

\noindent $\color{node_2} \alpha_{2} $ uses $\color{node_2} C_{2,6}=1987$ to generate 
$$ \begin{array}{l} {\color{node_2} d_{2,i,6} = 
 \begin{bmatrix} 
1628 & 1384 & 1744 & 108 & 531 \end{bmatrix}^\intercal }  
 \end{array}$$ 

\noindent $\color{node_3} \alpha_{3} $ uses $\color{node_3} C_{3,6}=454$ to generate 
$$ \begin{array}{l} {\color{node_3} d_{3,i,6} = 
 \begin{bmatrix} 
123 & 449 & 1934 & 702 & 1110 \end{bmatrix}^\intercal }  
 \end{array}$$ 

\noindent $\color{node_4} \alpha_{4} $ uses $\color{node_4} F_{4,6}=946$ to generate 
$$ \begin{array}{l} {\color{node_4} d_{4,i,6} = 
 \begin{bmatrix} 
1289 & 1265 & 324 & 1168 & 1094 \end{bmatrix}^\intercal }  
 \end{array}$$ 

\noindent $\color{node_5} \alpha_{5} $ uses $\color{node_5} C_{5,6}=650$ to generate 
$$ \begin{array}{l} {\color{node_5} d_{5,i,6} = 
 \begin{bmatrix} 
604 & 1621 & 1261 & 1846 & 890 \end{bmatrix}^\intercal }  
 \end{array}$$ 
\noindent $\color{node_1} \alpha_{1} $  receives 
$$  \begin{array}{l}
{\color{node_1} d_{i,1,j} = \begin{bmatrix} 
{\color{node_1}  1149} &  
 {\color{node_2}  1628} &  
 {\color{node_3}  123} &  
 {\color{node_4}  1289} &  
 {\color{node_5}  604 }  \\ 
 \end{bmatrix}^\intercal  } \quad 
\end{array}$$
\noindent $\color{node_2} \alpha_{2} $  receives 
$$  \begin{array}{l}
{\color{node_2} d_{i,2,j} = \begin{bmatrix} 
{\color{node_1}  1790} &  
 {\color{node_2}  1384} &  
 {\color{node_3}  449} &  
 {\color{node_4}  1265} &  
 {\color{node_5}  1621 }  \\ 
 \end{bmatrix}^\intercal  } \quad 
\end{array}$$
\noindent $\color{node_3} \alpha_{3} $  receives 
$$  \begin{array}{l}
{\color{node_3} d_{i,3,j} = \begin{bmatrix} 
{\color{node_1}  769} &  
 {\color{node_2}  1744} &  
 {\color{node_3}  1934} &  
 {\color{node_4}  324} &  
 {\color{node_5}  1261 }  \\ 
 \end{bmatrix}^\intercal  } \quad 
\end{array}$$
\noindent $\color{node_4} \alpha_{4} $  receives 
$$  \begin{array}{l}
{\color{node_4} d_{i,4,j} = \begin{bmatrix} 
{\color{node_1}  1280} &  
 {\color{node_2}  108} &  
 {\color{node_3}  702} &  
 {\color{node_4}  1168} &  
 {\color{node_5}  1846 }  \\ 
 \end{bmatrix}^\intercal  } \quad 
\end{array}$$
\noindent $\color{node_5} \alpha_{5} $  receives 
$$  \begin{array}{l}
{\color{node_5} d_{i,5,j} = \begin{bmatrix} 
{\color{node_1}  362} &  
 {\color{node_2}  531} &  
 {\color{node_3}  1110} &  
 {\color{node_4}  1094} &  
 {\color{node_5}  890 }  \\ 
 \end{bmatrix}^\intercal  } \quad 
\end{array}.$$
 \newpage According to Algorithm \ref{ring_transfer_D_algo}, $D_{6} = 305$ is generated. 

 We compute $P_6^{n-2}=305$, and then deduct that
$$D_6 = P_6^{n-2} \implies p_{0,6} = 0 $$ 
\begin{table}[h] 
\center\resizebox* {!} {6cm}{ 
\begin{tabular}{  c | c c | c c | c c | c c | c c | c c | c c | c  c } 
& \multicolumn{2}{c|}{$j=1$} 
& \multicolumn{2}{c|}{$j=2$} 
& \multicolumn{2}{c|}{$j=3$} 
& \multicolumn{2}{c|}{$j=4$} 
& \multicolumn{2}{c|}{$j=5$} 
& \multicolumn{2}{c|}{$j=6$} 
& \multicolumn{2}{c|}{$j=7$} 
& \multicolumn{2}{c|}{$j=8$} 
\\ \cline{1-17} 
\multirow{2}{*}{{\color{node_1}  $p_1=143$}}   
& {\color{node_1}  1} &  
& {\color{node_1}  0} &  
& {\color{node_1}  0} &  
& {\color{node_1}  0} &  
& {\color{node_1}  1} &  
& {\color{node_1}  1} &  
& {\color{node_1}  1} &  
& {\color{node_1}  1} &  
\multicolumn{1}{c|}{} 
 \\ \cline{2-17} 
& \rotatebox[origin=c]{180}{$\Lsh$} & \multicolumn{1}{|c|}{\cellcolor {gray!30} 1} 
& \rotatebox[origin=c]{180}{$\Lsh$} & \multicolumn{1}{|c|}{\cellcolor {gray!30} 0} 
& \rotatebox[origin=c]{180}{$\Lsh$} & \multicolumn{1}{|c|}{\cellcolor {gray!30} 0} 
& \rotatebox[origin=c]{180}{$\Lsh$} & \multicolumn{1}{|c|}{\cellcolor {gray!30} 0} 
& \rotatebox[origin=c]{180}{$\Lsh$} & \multicolumn{1}{|c|}{\cellcolor {gray!30} 0} 
& \rotatebox[origin=c]{180}{$\Lsh$} & \multicolumn{1}{|c|}{\cellcolor {gray!30} 0} 
& \rotatebox[origin=c]{180}{$\Lsh$} & \multicolumn{1}{|c|}{\cellcolor {gray!30} 0} 
& \rotatebox[origin=c]{180}{$\Lsh$} & \multicolumn{1}{|c|}{\cellcolor {gray!30} 0} 
\\ \cline{1-17} 
\multirow{2}{*}{{\color{node_2}  $p_2=124$}}   
& {\color{node_2}  0} &  
& {\color{node_2}  1} &  
& {\color{node_2}  1} &  
& {\color{node_2}  1} &  
& {\color{node_2}  1} &  
& {\color{node_2}  1} &  
& {\color{node_2}  0} &  
& {\color{node_2}  0} &  
\multicolumn{1}{c|}{} 
 \\
\end{tabular}} 
 \caption{One or no bidder for $j=6$} 
 \end{table} 
\bigskip

\newpage For $j=7$, Algorithm \ref{ring_transfer_BP_algo} generates $B_7=915$ and $P_7=1070$.  
$K_{1,7} \neq B_{7}$, 
$K_{2,7} \neq B_{7}$, 
$K_{3,7} \neq B_{7}$, 
$K_{5,7} \neq B_{7}$, 

$K_{4,7} = B_{7}$, $\implies$ $\alpha_{4}$ is the only bidder for $j=7$

\noindent $\color{node_1} \alpha_{1} $ uses $\color{node_1} C_{1,7}=1741$ to generate 
$$ \begin{array}{l} {\color{node_1} d_{1,i,7} = 
 \begin{bmatrix} 
1227 & 110 & 660 & 154 & 1633 \end{bmatrix}^\intercal }  
 \end{array}$$ 

\noindent $\color{node_2} \alpha_{2} $ uses $\color{node_2} C_{2,7}=1533$ to generate 
$$ \begin{array}{l} {\color{node_2} d_{2,i,7} = 
 \begin{bmatrix} 
1559 & 414 & 1870 & 1912 & 1689 \end{bmatrix}^\intercal }  
 \end{array}$$ 

\noindent $\color{node_3} \alpha_{3} $ uses $\color{node_3} C_{3,7}=1421$ to generate 
$$ \begin{array}{l} {\color{node_3} d_{3,i,7} = 
 \begin{bmatrix} 
724 & 1493 & 88 & 1058 & 965 \end{bmatrix}^\intercal }  
 \end{array}$$ 

\noindent $\color{node_4} \alpha_{4} $ uses $\color{node_4} F_{4,7}=275$ to generate 
$$ \begin{array}{l} {\color{node_4} d_{4,i,7} = 
 \begin{bmatrix} 
637 & 374 & 1388 & 289 & 638 \end{bmatrix}^\intercal }  
 \end{array}$$ 

\noindent $\color{node_5} \alpha_{5} $ uses $\color{node_5} C_{5,7}=1825$ to generate 
$$ \begin{array}{l} {\color{node_5} d_{5,i,7} = 
 \begin{bmatrix} 
1922 & 1277 & 1184 & 1144 & 306 \end{bmatrix}^\intercal }  
 \end{array}$$ 
 \noindent $\color{node_1} \alpha_{1} $  receives 
$$  \begin{array}{l}
{\color{node_1} d_{i,1,j} = \begin{bmatrix} 
{\color{node_1}  1227} &  
 {\color{node_2}  1559} &  
 {\color{node_3}  724} &  
 {\color{node_4}  637} &  
 {\color{node_5}  1922 }  \\ 
 \end{bmatrix}^\intercal  } \quad 
\end{array}$$
\noindent $\color{node_2} \alpha_{2} $  receives 
$$  \begin{array}{l}
{\color{node_2} d_{i,2,j} = \begin{bmatrix} 
{\color{node_1}  110} &  
 {\color{node_2}  414} &  
 {\color{node_3}  1493} &  
 {\color{node_4}  374} &  
 {\color{node_5}  1277 }  \\ 
 \end{bmatrix}^\intercal  } \quad 
\end{array}$$
\noindent $\color{node_3} \alpha_{3} $  receives 
$$  \begin{array}{l}
{\color{node_3} d_{i,3,j} = \begin{bmatrix} 
{\color{node_1}  660} &  
 {\color{node_2}  1870} &  
 {\color{node_3}  88} &  
 {\color{node_4}  1388} &  
 {\color{node_5}  1184 }  \\ 
 \end{bmatrix}^\intercal  } \quad 
\end{array}$$
\noindent $\color{node_4} \alpha_{4} $  receives 
$$  \begin{array}{l}
{\color{node_4} d_{i,4,j} = \begin{bmatrix} 
{\color{node_1}  154} &  
 {\color{node_2}  1912} &  
 {\color{node_3}  1058} &  
 {\color{node_4}  289} &  
 {\color{node_5}  1144 }  \\ 
 \end{bmatrix}^\intercal  } \quad 
\end{array}$$
\noindent $\color{node_5} \alpha_{5} $  receives 
$$  \begin{array}{l}
{\color{node_5} d_{i,5,j} = \begin{bmatrix} 
{\color{node_1}  1633} &  
 {\color{node_2}  1689} &  
 {\color{node_3}  965} &  
 {\color{node_4}  638} &  
 {\color{node_5}  306 }  \\ 
 \end{bmatrix}^\intercal  } \quad 
\end{array}.$$
 \newpage According to Algorithm \ref{ring_transfer_D_algo}, $D_{7} = 592$ is generated. 

 We compute $P_7^{n-2}=592$, and then deduct that
$$D_7 = P_7^{n-2} \implies p_{0,7} = 0 $$ 
\begin{table}[h] 
\center\resizebox* {!} {6cm}{ 
\begin{tabular}{  c | c c | c c | c c | c c | c c | c c | c c | c  c } 
& \multicolumn{2}{c|}{$j=1$} 
& \multicolumn{2}{c|}{$j=2$} 
& \multicolumn{2}{c|}{$j=3$} 
& \multicolumn{2}{c|}{$j=4$} 
& \multicolumn{2}{c|}{$j=5$} 
& \multicolumn{2}{c|}{$j=6$} 
& \multicolumn{2}{c|}{$j=7$} 
& \multicolumn{2}{c|}{$j=8$} 
\\ \cline{1-17} 
\multirow{2}{*}{{\color{node_1}  $p_1=143$}}   
& {\color{node_1}  1} &  
& {\color{node_1}  0} &  
& {\color{node_1}  0} &  
& {\color{node_1}  0} &  
& {\color{node_1}  1} &  
& {\color{node_1}  1} &  
& {\color{node_1}  1} &  
& {\color{node_1}  1} &  
\multicolumn{1}{c|}{} 
 \\ \cline{2-17} 
& \rotatebox[origin=c]{180}{$\Lsh$} & \multicolumn{1}{|c|}{\cellcolor {gray!30} 1} 
& \rotatebox[origin=c]{180}{$\Lsh$} & \multicolumn{1}{|c|}{\cellcolor {gray!30} 0} 
& \rotatebox[origin=c]{180}{$\Lsh$} & \multicolumn{1}{|c|}{\cellcolor {gray!30} 0} 
& \rotatebox[origin=c]{180}{$\Lsh$} & \multicolumn{1}{|c|}{\cellcolor {gray!30} 0} 
& \rotatebox[origin=c]{180}{$\Lsh$} & \multicolumn{1}{|c|}{\cellcolor {gray!30} 0} 
& \rotatebox[origin=c]{180}{$\Lsh$} & \multicolumn{1}{|c|}{\cellcolor {gray!30} 0} 
& \rotatebox[origin=c]{180}{$\Lsh$} & \multicolumn{1}{|c|}{\cellcolor {gray!30} 0} 
& \rotatebox[origin=c]{180}{$\Lsh$} & \multicolumn{1}{|c|}{\cellcolor {gray!30} 0} 
\\ \cline{1-17} 
\multirow{2}{*}{{\color{node_2}  $p_2=124$}}   
& {\color{node_2}  0} &  
& {\color{node_2}  1} &  
& {\color{node_2}  1} &  
& {\color{node_2}  1} &  
& {\color{node_2}  1} &  
& {\color{node_2}  1} &  
& {\color{node_2}  0} &  
& {\color{node_2}  0} &  
\multicolumn{1}{c|}{} 
 \\
\end{tabular}} 
 \caption{One or no bidder for $j=7$} 
 \end{table} 
\bigskip

\newpage For $j=8$, Algorithm \ref{ring_transfer_BP_algo} generates $B_8=1843$ and $P_8=567$.  
$K_{1,8} \neq B_{8}$, 
$K_{2,8} \neq B_{8}$, 
$K_{3,8} \neq B_{8}$, 
$K_{4,8} \neq B_{8}$, 
$K_{5,8} \neq B_{8}$, 

\noindent $\color{node_1} \alpha_{1} $ uses $\color{node_1} C_{1,8}=412$ to generate 
$$ \begin{array}{l} {\color{node_1} d_{1,i,8} = 
 \begin{bmatrix} 
1957 & 1460 & 2042 & 118 & 1282 \end{bmatrix}^\intercal }  
 \end{array}$$ 

\noindent $\color{node_2} \alpha_{2} $ uses $\color{node_2} C_{2,8}=862$ to generate 
$$ \begin{array}{l} {\color{node_2} d_{2,i,8} = 
 \begin{bmatrix} 
181 & 495 & 1687 & 1333 & 179 \end{bmatrix}^\intercal }  
 \end{array}$$ 

\noindent $\color{node_3} \alpha_{3} $ uses $\color{node_3} C_{3,8}=832$ to generate 
$$ \begin{array}{l} {\color{node_3} d_{3,i,8} = 
 \begin{bmatrix} 
1364 & 78 & 1282 & 880 & 1324 \end{bmatrix}^\intercal }  
 \end{array}$$ 

\noindent $\color{node_4} \alpha_{4} $ uses $\color{node_4} C_{4,8}=806$ to generate 
$$ \begin{array}{l} {\color{node_4} d_{4,i,8} = 
 \begin{bmatrix} 
1559 & 187 & 250 & 174 & 76 \end{bmatrix}^\intercal }  
 \end{array}$$ 

\noindent $\color{node_5} \alpha_{5} $ uses $\color{node_5} C_{5,8}=1274$ to generate 
$$ \begin{array}{l} {\color{node_5} d_{5,i,8} = 
 \begin{bmatrix} 
900 & 9 & 47 & 469 & 639 \end{bmatrix}^\intercal }  
 \end{array}$$ 
\noindent $\color{node_1} \alpha_{1} $  receives 
$$  \begin{array}{l}
{\color{node_1} d_{i,1,j} = \begin{bmatrix} 
{\color{node_1}  1957} &  
 {\color{node_2}  181} &  
 {\color{node_3}  1364} &  
 {\color{node_4}  1559} &  
 {\color{node_5}  900 }  \\ 
 \end{bmatrix}^\intercal  } \quad 
\end{array}$$
\noindent $\color{node_2} \alpha_{2} $  receives 
$$  \begin{array}{l}
{\color{node_2} d_{i,2,j} = \begin{bmatrix} 
{\color{node_1}  1460} &  
 {\color{node_2}  495} &  
 {\color{node_3}  78} &  
 {\color{node_4}  187} &  
 {\color{node_5}  9 }  \\ 
 \end{bmatrix}^\intercal  } \quad 
\end{array}$$
\noindent $\color{node_3} \alpha_{3} $  receives 
$$  \begin{array}{l}
{\color{node_3} d_{i,3,j} = \begin{bmatrix} 
{\color{node_1}  2042} &  
 {\color{node_2}  1687} &  
 {\color{node_3}  1282} &  
 {\color{node_4}  250} &  
 {\color{node_5}  47 }  \\ 
 \end{bmatrix}^\intercal  } \quad 
\end{array}$$
\noindent $\color{node_4} \alpha_{4} $  receives 
$$  \begin{array}{l}
{\color{node_4} d_{i,4,j} = \begin{bmatrix} 
{\color{node_1}  118} &  
 {\color{node_2}  1333} &  
 {\color{node_3}  880} &  
 {\color{node_4}  174} &  
 {\color{node_5}  469 }  \\ 
 \end{bmatrix}^\intercal  } \quad 
\end{array}$$
\noindent $\color{node_5} \alpha_{5} $  receives 
$$  \begin{array}{l}
{\color{node_5} d_{i,5,j} = \begin{bmatrix} 
{\color{node_1}  1282} &  
 {\color{node_2}  179} &  
 {\color{node_3}  1324} &  
 {\color{node_4}  76} &  
 {\color{node_5}  639 }  \\ 
 \end{bmatrix}^\intercal  } \quad 
\end{array}.$$
 \newpage According to Algorithm \ref{ring_transfer_D_algo}, $D_{8} = 725$ is generated. 

 We compute $P_8^{n-2}=1709$, and then deduct that
$$D_8 \neq P_8^{n-2} \implies p_{0,8} = 1 $$ 

\begin{table}[h] 
\center\resizebox* {!} {6cm}{ 
\begin{tabular}{  c | c c | c c | c c | c c | c c | c c | c c | c  c } 
& \multicolumn{2}{c|}{$j=1$} 
& \multicolumn{2}{c|}{$j=2$} 
& \multicolumn{2}{c|}{$j=3$} 
& \multicolumn{2}{c|}{$j=4$} 
& \multicolumn{2}{c|}{$j=5$} 
& \multicolumn{2}{c|}{$j=6$} 
& \multicolumn{2}{c|}{$j=7$} 
& \multicolumn{2}{c|}{$j=8$} 
\\ \cline{1-17} 
\multirow{2}{*}{{\color{node_1}  $p_1=143$}}   
& {\color{node_1}  1} &  
& {\color{node_1}  0} &  
& {\color{node_1}  0} &  
& {\color{node_1}  0} &  
& {\color{node_1}  1} &  
& {\color{node_1}  1} &  
& {\color{node_1}  1} &  
& {\color{node_1}  1} &  
\multicolumn{1}{c|}{} 
 \\ \cline{2-17} 
& \rotatebox[origin=c]{180}{$\Lsh$} & \multicolumn{1}{|c|}{\cellcolor {gray!30} 1} 
& \rotatebox[origin=c]{180}{$\Lsh$} & \multicolumn{1}{|c|}{\cellcolor {gray!30} 0} 
& \rotatebox[origin=c]{180}{$\Lsh$} & \multicolumn{1}{|c|}{\cellcolor {gray!30} 0} 
& \rotatebox[origin=c]{180}{$\Lsh$} & \multicolumn{1}{|c|}{\cellcolor {gray!30} 0} 
& \rotatebox[origin=c]{180}{$\Lsh$} & \multicolumn{1}{|c|}{\cellcolor {gray!30} 0} 
& \rotatebox[origin=c]{180}{$\Lsh$} & \multicolumn{1}{|c|}{\cellcolor {gray!30} 0} 
& \rotatebox[origin=c]{180}{$\Lsh$} & \multicolumn{1}{|c|}{\cellcolor {gray!30} 0} 
& \rotatebox[origin=c]{180}{$\Lsh$} & \multicolumn{1}{|c|}{\cellcolor {gray!30} 0} 
\\ \cline{1-17} 
\multirow{2}{*}{{\color{node_2}  $p_2=124$}}   
& {\color{node_2}  0} &  
& {\color{node_2}  1} &  
& {\color{node_2}  1} &  
& {\color{node_2}  1} &  
& {\color{node_2}  1} &  
& {\color{node_2}  1} &  
& {\color{node_2}  0} &  
& {\color{node_2}  0} &  
\multicolumn{1}{c|}{} 
 \\
\end{tabular}} 
 \caption{Two or more bidders for $j=8$} 
 \end{table} 

\newpage
The output price of this Vickrey auction is
$$ \begin{array}{c} 
p_0 = (1  1  0  1  1  0  0  1)_{\text{base 2}} = 217 \\ 
 \end{array}$$ 
According to the final price verification Algorithm \ref{Final_Price_Verification_algo}, we have
$$  \begin{array}{l}
{\mathcal{P}_l = \begin{bmatrix} 
{\color{node_1}  1301} &  
 {\color{node_2}  705} &  
 {\color{node_3}  718} &  
 {\color{node_4}  511} &  
 {\color{node_5}  150 }  \\ 
 \end{bmatrix}^\intercal  }
\end{array}$$
Moreover, since there exists $\mathbf{P}_3=\mathcal{P}_3$, 
the output price $p_0$ is the original committed second-highest bid.
This information about bidders is also protected by the anonymity of the Internet.

Let $j'=7$ be the largest index such that $p_{0, j'} = 0$.
The winner $\alpha_4$ sends the seller $K_{u,4,7}$ for $\forall u$ which can be verified by checking if
\begin{equation}
 \prod_{u=1}^n K_{u,4,7}=B_{7}
 \end{equation}


\begin{table}[h] 
\center\resizebox* {!} {6cm}{ 
\begin{tabular}{  c | c c | c c | c c | c c | c c | c c | c c | c  c } 
& \multicolumn{2}{c|}{$j=1$} 
& \multicolumn{2}{c|}{$j=2$} 
& \multicolumn{2}{c|}{$j=3$} 
& \multicolumn{2}{c|}{$j=4$} 
& \multicolumn{2}{c|}{$j=5$} 
& \multicolumn{2}{c|}{$j=6$} 
& \multicolumn{2}{c|}{$j=7$} 
& \multicolumn{2}{c|}{$j=8$} 
\\ \cline{1-17} 
\multirow{2}{*}{{\color{node_1}  $p_1=143$}}
& {\color{node_1}  1} &  
& {\color{node_1}  0} &  
& {\color{node_1}  0} &  
& {\color{node_1}  0} &  
& {\color{node_1}  1} &  
& {\color{node_1}  1} &  
& {\color{node_1}  1} &  
& {\color{node_1}  1} &  
\multicolumn{1}{c|}{} 
 \\ \cline{2-17} 
& \rotatebox[origin=c]{180}{$\Lsh$} & \multicolumn{1}{|c|}{\cellcolor {gray!30} 1} 
& \rotatebox[origin=c]{180}{$\Lsh$} & \multicolumn{1}{|c|}{\cellcolor {gray!30} 0} 
& \rotatebox[origin=c]{180}{$\Lsh$} & \multicolumn{1}{|c|}{\cellcolor {gray!30} 0} 
& \rotatebox[origin=c]{180}{$\Lsh$} & \multicolumn{1}{|c|}{\cellcolor {gray!30} 0} 
& \rotatebox[origin=c]{180}{$\Lsh$} & \multicolumn{1}{|c|}{\cellcolor {gray!30} 0} 
& \rotatebox[origin=c]{180}{$\Lsh$} & \multicolumn{1}{|c|}{\cellcolor {gray!30} 0} 
& \rotatebox[origin=c]{180}{$\Lsh$} & \multicolumn{1}{|c|}{\cellcolor {gray!30} 0} 
& \rotatebox[origin=c]{180}{$\Lsh$} & \multicolumn{1}{|c|}{\cellcolor {gray!30} 0} 
\\ \cline{1-17} 
\multirow{2}{*}{{\color{node_2}  $p_2=124$}}   
& {\color{node_2}  0} &  
& {\color{node_2}  1} &  
& {\color{node_2}  1} &  
& {\color{node_2}  1} &  
& {\color{node_2}  1} &  
& {\color{node_2}  1} &  
& {\color{node_2}  0} &  
& {\color{node_2}  0} &  
\multicolumn{1}{c|}{} 
 \\
\end{tabular}} 
\vspace{0.3cm}
 \caption{
In the yellow cell at the bottom, it can be observed that the winner $\alpha_4$ is the only bidder who placed a `1' at digit $j'=7$.
Therefore, $\alpha_4$ is willing to pay a price that is at least one bid higher than the second-highest bid, which is represented by 
$p_4 ~\geq~ p_0+1 = (1  1  0  1  1  0  1  0)_{\text{base 2}} $ 
} 
\label{table_example}
\end{table}

\end{appendix}

\end{document}